\documentclass[11pt,a4paper]{amsart}

\usepackage{latexsym}
\usepackage{bbm}
\usepackage[pdftex]{hyperref,color,graphicx}

\usepackage{tikz} 
\usepackage{comment} 
\usepackage{multirow}
\usepackage{dcolumn}
\newcolumntype{d}{D{.}{.}{3.1}}

\usepackage{multirow}
\usepackage{booktabs}
\usepackage{bigstrut}

\usepackage{lscape}

\usepackage{amsmath,amsfonts,amssymb,amsthm} 

\usepackage{hyperref,color}

\newtheorem{theorem}{Theorem}[section]
\newtheorem{corollary}[theorem]{Corollary}
\newtheorem{lemma}[theorem]{Lemma}

\newtheorem{proposition}[theorem]{Proposition}

\theoremstyle{definition}

\theoremstyle{remark}
\newtheorem{remark}[theorem]{Remark}

\numberwithin{equation}{section}

\newcommand{\Law}{\mathrm{Law}}

\newcommand{\ud}{\,\mathrm{d}}

\newcommand{\e}{\mathrm{e}}

\DeclareMathOperator*{\esssup}{ess\,sup}

\begin{document}

\title[ESO with drift change point]{Executive stock option exercise
  with full and partial information on a drift change point}


\author[Vicky Henderson]{Vicky Henderson}

\address{Vicky Henderson\\
Department of Statistics \\
Zeeman Building \\
University Of Warwick \\
Coventry CV4 7AL \\
UK}

\email[]{Vicky.Henderson@warwick.ac.uk}

\author[Kamil Klad{\'i}vko]{Kamil Klad{\'i}vko}

\address{Kamil Klad{\'i}vko\\
School of Business\\
\"{O}rebro University \\
701 82 \"{O}rebro, Sweden}

\email[]{kladivko@gmail.com}

\author[Michael Monoyios]{Michael Monoyios}

\address{Michael Monoyios\\
Mathematical Institute \\ 
University of Oxford \\
Radcliffe Observatory Quarter \\
Woodstock Road\\ 
Oxford OX2 6GG \\
UK}

\email[]{monoyios@maths.ox.ac.uk}

\author[Christoph Reisinger]{Christoph Reisinger}

\address{Christoph Reisinger\\
Mathematical Institute \\ 
University of Oxford \\
Radcliffe Observatory Quarter \\
Woodstock Road\\ 
Oxford OX2 6GG \\
UK}

\email[]{reisinge@maths.ox.ac.uk}

\date{\today}

\begin{abstract}

 We analyse the optimal exercise of an American call executive stock
option (ESO) written on a stock whose drift parameter falls to a lower
value at a \emph{change point}, an exponentially distributed random
time independent of the Brownian motion driving the stock. Two agents,
who do not trade the stock, have differing information on the change
point, and seek to optimally exercise the option by maximising its
discounted payoff under the physical measure. The first agent has full
information, and observes the change point. The second agent has
partial information and filters the change point from price
observations. This scenario is designed to mimic the positions of two
employees of varying seniority, a fully informed executive and a
partially informed less senior employee, each of whom receives an
ESO\@. The partial information scenario yields a model under the
observation filtration $\widehat{\mathbb{F}}$ in which the stock drift
becomes a diffusion driven by the innovations process, an
$\widehat{\mathbb{F}}$-Brownian motion also driving the stock under
$\widehat{\mathbb{F}}$, and the partial information optimal stopping
value function has two spatial dimensions. We rigorously characterise
the free boundary PDEs for both agents, establish shape and regularity
properties of the associated optimal exercise boundaries, and prove
the smooth pasting property in both information scenarios, exploiting
some stochastic flow ideas to do so in the partial information
case. We develop finite difference algorithms to numerically solve
both agents' exercise and valuation problems and illustrate that the
additional information of the fully informed agent can result in
exercise patterns which exploit the information on the change point,
lending credence to empirical studies which suggest that privileged
information of bad news is a factor leading to early exercise of ESOs
prior to poor stock price performance.

\end{abstract}

\maketitle

{\small \noindent Keywords: optimal stopping, free boundary problems,
executive stock options, American options, smooth pasting,
stochastic flows, Kalman-Bucy filter
\newline AMS Subject classifications: 91G80, 93E11, 93E20}

\section{Introduction}
\label{sec:intro}


In this paper we consider two pure optimal stopping problems involving
a constant volatility stock whose drift parameter suffers a
\emph{change point}. At an exponentially distributed random time
$\theta$ (the change point), independent of the Brownian motion $W$
driving the stock, its drift falls from its initial constant value
$\mu_{0}$ to a lower constant value $\mu_{1}<\mu_{0}$. The two
problems we study are distinguished by \emph{full information}, in
which the change point is observed, or by \emph{partial information},
in which the change point is not observable, and so is filtered from
observations of the stock price.

The optimal stopping problems arise from the exercise of an executive
stock option (ESO), an American call on a stock that is not traded by
the option holders. Such a scenario is sometimes referred to as a
``pure buyer's position'', wherein an agent acquires an option, is not
able to hedge the option due to trading
restrictions, and seeks only to optimally exercise the
claim. The objective we use for this completely unhedgeable payoff is
to maximise the discounted payoff under the physical measure
$\mathbb{P}$ over stopping times of the agent's filtration. Our two
ESO-holding agents thus differ only in the respective filtrations to
which each has access, and one of our goals is to understand
how this information differential affects their exercise
strategies. 
Our aim is to capture a firm specific disastrous event, which 
happens at a random time, and is immediately known by the firm's top executives, but it is not revealed publicly, at least not immediately, and thus it is unknown to less senior employees. 
Recent examples of such disastrous events could be the Volkswagen emissions scandal (Dieselgate), the Facebook-Cambridge Analytica data scandal, or 
Boeing 737 MAX groundings. 

The first agent has ``full information''. He observes the change point
process $Y\in\{0,1\}$ (the indicator that the change point has
occurred) as well as the Brownian motion $W$, so his filtration,
$\mathbb{F}$ (the ``large'' filtration, or \emph{background
  filtration}), is the augmentation of the filtration generated by
$(W,Y)$. In this case, the (random) drift process of the stock is
$\mu(Y)$, given by a linear function of the change point process
$Y\in\{0,1\}$, such that at all times the drift is equal to one of the
distinct values ($\mu_{0}$ before the change point, $\mu_{1}$
afterwards, see \eqref{eq:mufi}).

The second agent has ``partial information''. She does not observe the
change point, and filters $Y$ (and thus the change point) from stock
price observations. The partially informed agent's filtration,
$\widehat{\mathbb{F}}$ (the \emph{observation filtration}), is thus
the augmentation of the stock price filtration, and
$\widehat{\mathbb{F}}\subset\mathbb{F}$. In this partial information
scenario, the filtered change point process $\widehat{Y}$ turns out to
be a diffusion in $[0,1]$ driven by the innovations process
$\widehat{W}$, which is the $\widehat{\mathbb{F}}$-Brownian motion
also driving the stock under the observation filtration. In this case,
the random drift turns out to be $\mu(\widehat{Y})$, featuring the
same linear function as in the full information case, but now of the
filtered process $\widehat{Y}$ (see \eqref{eq:dSobs}). The process
$\widehat{Y}$, adapted to the stock price filtration, turns out to be
a functional of the path-history of the stock price.

For both the full and partial information problems, we carry out a
detailed and rigorous free boundary analysis of the associated value
function for the option. For each problem this involves a classical
program of steps, which we generalise from the (typical) constant
drift case to each of our two random drift scenarios, as follows. The
two-state drift of the full information problem naturally leads to a
pair of value functions (one for each possible initial drift state
$i\in\{0,1\}$) characterising the ESO value. Equally naturally, in the
partial information problem, dependent on the diffusion
$\widehat{Y}\in[0,1]$, the value function depends on a variable
$y\in[0,1]$, representing the initial value of the change point
process (in addition to the usual temporal and stock price
dependence).

We first derive basic convexity, monotonicity and time decay
properties of the value functions (Lemma \ref{lem:cmtdfi} (full
information) and Lemma \ref{lem:cmtdpi} (partial information)), the
latter using some stochastic flow ideas applied to $\widehat{Y}(y)$,
the filtered change point process viewed as a function of its initial
value $y$. From these results we infer the form of the continuation
and stopping regions, the existence and form of optimal exercise
thresholds and (later) their limiting values as we approach the ESO
maturity time.

We show that, for the full information problem, there are a pair of
ordered, non-increasing, time-dependent exercise boundaries
$x^{*}_{0}(\cdot)\geq x^{*}_{1}(\cdot)$, such that optimal early
exercise can occur in the state where the drift is
$\mu_{i},i\in\{0,1\}$, when the stock breaches $x^{*}_{i}(\cdot)$ from
below, or if such a breach is triggered by the change point. On the
other hand, in the partial information case the exercise boundary
$x^{*}(\cdot,\cdot)$ is a \emph{surface}, with an additional spatial,
non-increasing dependence on the variable $y\in[0,1]$, arising from
the dependence of the drift on the filtered change point process, and
such that the partial information exercise surface lies between the
full information exercise thresholds. This can lead to an interesting
range of possible exercise patterns (such as immediate exercise by the
fully informed agent in response to the change point, a strategy
unavailable to the agent who does not see the jump in drift), which we
describe (and later examine numerically). We also consider how our
stopping problems are changed with the inclusion of an option vesting
period. In practice, vesting periods during which the option holder is
not permitted to exercise, are used by the company to maintain the
employee's incentives or exposure to the stock price.

We then give a rigorous characterisation of the ESO value functions in
terms of free boundary PDEs (Proposition \ref{prop:fbpfivf} (full
information) and Proposition \ref{prop:fbppivf} (partial information))
with associated smooth pasting conditions at the exercise thresholds
(Theorem \ref{thm:spfi} (full information) and Theorem \ref{thm:sppi}
(partial information)). Using these results we are able to derive
Doob-Meyer decompositions of the supermartingales which represent the
discounted ESO value processes (Theorem \ref{thm:dmdecomp} (full
information) and Lemma \ref{lem:dmdpise} (partial information)). These
in turn are used in proving the results on the limiting values of the
boundaries as we approach maturity $T$ (Proposition \ref{prop:tvfull}
(full information) and Lemma \ref{lem:tvpartial} (partial
information)). Although not needed elsewhere, we also show that the
boundaries for the full information problem are continuous over
$[0,T)$, as stated in Proposition \ref{prop:tvfull}.

Our mathematical results are obtained by implementing, broadly
speaking, the classical program for obtaining properties of American
options (see for example Karatzas and Shreve \cite[Chapter 2]{ks98}
for the American put in the Black-Scholes model), and carefully
modifying and extending these arguments to our random drift scenarios,
augmenting them in places with new tools, such as the stochastic flow
ideas mentioned above. These results are novel compared to existing
literature, as we now describe.

The full information case has some similarities with papers on
American option valuation with regime switching, such as the infinite
horizon put in Guo and Zhang \cite{gz04} and the finite horizon put in
Buffington and Elliott \cite{be02} (who assume all required regularity
properties of the value function). Le and Wang \cite{lw10} also treat
the American put with regime switching, and do prove the smooth
pasting property, by extending a fairly involved iterative procedure
originally due to Bayraktar \cite{bayraktar09}. As well as being
lengthy, some steps exploit the boundedness of the put payoff
function, so it is not clear if they are directly applicable to our
model. Here, therefore, we exploit our explicit one-switch scenario
and show how more classical techniques can be extended to the random
drift case, both for the free boundary characterisation, and then for
the smooth pasting property. The latter requires an analysis of the
optimal stopping time given a particular starting state, and here we
use our derived structures for the stopping and continuation regions.

In the partial information case, our results are entirely new. The
rigorous characterisation of the value function as a solution of a
free boundary PDE with an associated smooth pasting condition, has not
been demonstrated before to the best of our knowledge. We achieve
this, also show that the exercise surface is decreasing in time and in
the initial value $y\in[0,1]$ of the filtered change point process,
and give its limiting terminal value. An infinite horizon American
put with partial information on a switching dividend process was
studied by Gapeev \cite{gapeev12}, but the regularity of the value
function and the smooth pasting property were assumed to hold. 
We resolve these issues in our partial information
problem. Note that, with our objective of maximising the discounted
expected payoff under the physical measure, our problems map to
conventional American option pricing problems under a martingale
measure, but with a random dividend yield. Thus, our results also give
the required regularity for the problems studied in \cite{gapeev12}.

Finally, there is a strand of papers (D{\'e}camps et
al.~\cite{dmv05,dmv09}, Klein \cite{klein09}, Ekstr{\"o}m and Lu,
\cite{el11} Ekstr{\"o}m and Vannest\aa l \cite{ev19}) which study
optimal stopping problems in a partial information scenario when a
drift parameter is assumed to take on one of two values, but the agent
is unsure which value pertains in reality. These models correspond to
the limit that the parameter of the exponential time in our model
approaches zero, so an explicit change point is absent (they are
models of an uncertain drift, as opposed to uncertainty in the timing
of a change of drift). This renders them simpler than our partial
information model, because the dependence of the filtered process on
the entire history of the stock disappears. These papers are then able
to reduce the dimensionality of the problem under some circumstances,
a simplification not available in our model.

We complete the picture by solving both problems numerically, using
finite difference schemes, and carry out simulations to illustrate
some of the exercise patterns that can occur. The partial information
case is substantially more difficult numerically due to the second
spatial dimension, but with a single Brownian driver, resulting in a
reduced rank diffusion matrix, and the degeneracy of some of the
diffusion and drift coefficients at certain boundaries of the domain.
This setting requires a novel, tailored approximation scheme for the
efficient numerical solution.  We propose a first order monotone and a
second order non-monotone penalised backward diifferentiation formulae
(BDF) scheme on non-uniform meshes and prove convergence for the
former. Numerical tests demonstrate the stability and achievable
accuracy for the scheme.

One of our motivations for studying these issues is a strand of
literature in empirical finance which attributes early ESO exercise
prior to poor stock performance in part to privileged information,
particularly on imminent bad news. Early studies (Huddart and Lang
\cite{huddartlang03}, Carpenter and Remmers \cite{carpenterremmers01})
provide some evidence that this is the case. More recent works that
partition the exercises according to the particular exercise strategy
employed find much stronger evidence of informed exercise (Brooks et
al.~\cite{brooks2012}, Cicero \cite{cicero09}, Aboody et
al.~\cite{aboody08}): exercises accompanied by a sale of stock are
followed by negative abnormal returns (while other exercises are not).
We were thus motivated to construct a model where complete or
incomplete information on an adverse event could be compared in the
exercise of an American call.  Here, we think of the fully informed
agent as a senior executive who observes the change point, while the
partially informed agent is thought of as a less senior employee who
is not privy to board meetings sharing imminent bad news. Our setup
considers a stock price whose drift will jump to, and remain at, a
lower value. We do not consider a model where the drift can switch
repeatedly between two values, as this would not capture a seismic
piece of adverse news, though a rigorous analysis of such a model
would be interesting, and could potentially be built upon our analysis
here.

We use our model to conduct a study of mean post-exercise returns for
agents with full and partial information, motivated by the empirical
work of Brooks et al.~\cite{brooks2012}.  Our simulations (in Section
\ref{sec:numex}) support the conjecture that indeed, the difference
between average post-exercise returns for fully and partially informed
agents is significantly negative.  For our simulations, the difference
between mean post-exercise returns for fully and partially informed
agents varies between about -3.8\% and -9.7\%, depending on the
expected stock return $\mu_0$ and volatility, covering the range of
values reported by Brooks et al.~\cite{brooks2012}. Our model thus
provides theoretical support for the tests conducted in the empirical
literature to evidence so-called insider exercises.

Our analysis leads to our being able to characterise exercise
scenarios, and to point out scenarios where the change point can
induce exercise for the fully informed agent, but of course not
necessarily for the partially informed agent, since the change point
is not seen. We illustrate this in Section \ref{sec:numex} where we
provide simulations of various exercise scenarios and show the agent
with full information has considerable advantage in exercise
timing. An exercise surface $x^*(t,y); t \in [0,T], y \in [0,1] $ for
the agent with partial information, and thresholds
$x^*_0(t), x^*_1(t); t \in [0,T]$ for the full information case are
computed and shown to be consistent with the theoretical results in
earlier sections.

The informational advantage demonstrated in the exercise strategies is
reflected in the respective ESO values the agents place on their
options.  We document that the additional value the agent with full
information places on his ESO is significant in magnitude. The early
exercise value as a proportion of the European value can be very many
times greater for the agent with full rather than only partial
information.  In Table \ref{table_comp}, we also report comparative
statics for the ESO value as we vary stock parameters
$\mu_0, \mu_1, \sigma$, and $\lambda$.  ESO values for both agents
decrease as the magnitude of the expected return in the bad state,
$\mu_1$, increases or there is a greater probability of a downward
jump. However, the early exercise values increase, indicating that the
ability to time the exercise of the option is more valuable when the
expected return following the change point is worse, or when the
chance of entering the bad state is higher.  We also report ESO values
when option vesting is included in the model and note, as expected,
the early exercise value drops for both agents, whilst the
informational advantage of the agent with full information is still
present.

The rest of the paper is organised as follows. In Section
\ref{sec:sdcp} we introduce the model and the optimal stopping
problems under both information scenarios, and carry out a filtering
procedure to derive the model dynamics with respect to the stock price
filtration. In Sections \ref{sec:fiesop} and \ref{sec:piesop} we
analyse the full and partial information problems, respectively.
Section \ref{sec:vesting} gives a brief discussion of how a vesting
period impacts upon exercise.  In Section \ref{sec:nrsim} we construct
and describe numerical methods for solving the two optimal stopping
problems, including convergence results.  We apply the finite
difference methodology in Section \ref{sec:numex} to perform
simulations to compare the exercise patterns of the agents, undertake
an analysis of post-exercise returns, and provide ESO valuation.

\section{Stock price with a drift change point}
\label{sec:sdcp}

We model a stock price whose drift will jump to a lower value at a
random time (a \emph{change point}). The goal is to investigate
differences in the ESO exercise strategy between a fully informed
agent who observes the change point, and a partially informed agent
who has to filter the change point from stock price observations. In
particular, we seek to explore whether the fully informed agent can
exploit his additional information in the exercise strategy.

The setting is a complete probability space
$(\Omega,\mathcal{F},\mathbb{P})$ equipped with a filtration
$\mathbb{F}:=(\mathcal{F}_{t})_{t\in\mathbf{T}}$ satisfying the usual
hypotheses of right-continuity and augmentation by all the
$\mathbb{P}$-null sets of $\mathcal{F}$. The time set $\mathbf{T}$
will be the finite interval $\mathbf{T}=[0,T]$, for some
$T<\infty$. The filtration $\mathbb{F}$ will sometimes be referred to
as the {\em background filtration}. It represents the large filtration
available to a perfectly informed agent, and all processes will be
assumed to be $\mathbb{F}$-adapted in what follows.

Let $W$ denote a standard $(\mathbb{P},\mathbb{F})$-Brownian motion.
Let $\theta\in\mathbb{R}_{+}$ be a non-negative random time,
independent of $W$, with initial distribution
$\mathbb{P}[\theta=0]=:y_{0}\in[0,1)$ and subsequent distribution
\begin{equation*}
\mathbb{P}[\theta> t|\theta>0] = \e^{-\lambda t}, \quad \lambda\geq 0,
\quad t\in\mathbf{T}.   
\end{equation*}
Thus, conditional on the event
$\{\omega\in\Omega:\theta(\omega)>0\}\equiv\{\theta >0\}$, $\theta$
has exponential distribution with parameter $\lambda$. Define the
single-jump c\`adl\`ag process $Y$ by
\begin{equation}
Y_{t} := \mathbbm{1}_{\{t\geq\theta\}}, \quad t\in\mathbf{T},  
\label{eq:signal}
\end{equation}
so that $Y_{0}=\mathbbm{1}_{\{\theta=0\}}$ with
$\mathbb{E}[Y_{0}]=y_{0}$. We may (and do) take $\mathbb{F}$ to be the
$\mathbb{P}$-augmentation of $\mathbb{F}^{W,Y}$, the filtration
generated by the pair $(W,Y)$. By Karatzas and Shreve
\cite[Proposition 2.7.7] {ks91} this filtration is indeed
right-continuous, because $(W,Y)$ is a strong Markov process.

We associate with $Y$ the $(\mathbb{P},\mathbb{F})$-martingale
$M^{(Y)}$ (the compensated jump process), defined by
\begin{equation}
M^{(Y)}_{t} := Y_{t} - Y_{0} - \lambda\int_{0}^{t}(1-Y_{s})\ud s, \quad
t\in\mathbf{T}.  
\label{eq:M}
\end{equation}

A stock price process $X$ with constant volatility $\sigma>0$ has a
drift which depends on the process $Y$. We are given two real
constants $\mu_{0}>\mu_{1}$ such that the drift value falls from
$\mu_{0}$ to the lower value $\mu_{1}$ at the change point. Define the
constant $\eta>0$ by
\begin{equation}
\eta := \frac{\mu_{0}-\mu_{1}}{\sigma}.
\label{eq:eta}
\end{equation}
The stock price dynamics with respect to $(\mathbb{P},\mathbb{F})$ are
given by
\begin{equation}
\ud X_{t} = (\mu_{0} -\sigma\eta Y_{t})X_{t}\ud t + \sigma X_{t}\ud W_{t}. 
\label{eq:dS}
\end{equation}
Thus, the drift process $\mu(Y)$ of the stock is given by
\begin{equation}
\mu(Y_{t}) := \mu_{0}-\sigma\eta Y_{t} = \mu_{0}(1-Y_{t}) +
\mu_{1}Y_{t} = \left\{\begin{array}{ccc}
\mu_{0}, & \mbox{on} & \{t<\theta\}=\{Y_{t}=0\}, \\
\mu_{1}, & \mbox{on} & \{t\geq\theta\}=\{Y_{t}=1\},
\end{array}\right. \quad t\in\mathbf{T}.
\label{eq:mufi}                    
\end{equation}
Note in particular that for $y_{0}=0$ the change point $\theta$ is
almost surely strictly positive, and the stock evolution almost surely
begins with the higher drift value $\mu_{0}$.

We assume that the values of the constants
$y_{0},\mu_{0},\mu_{1},\sigma,\lambda$ are given. Finally, there is
also a cash account paying a constant interest rate $r\geq 0$.
Dividends could also be included, and there are several possibilities
as to how these could be modelled, but we do not do so for
simplicity. For example, a constant dividend yield could be included
with minor adjustments by re-interpreting the drifts as being net of
dividends.

We may write the stock price evolution as
\begin{equation}
\ud X_{t} = \sigma X_{t}\ud \xi_{t}, 
\label{eq:dX}
\end{equation}
where $\xi$ is the volatility-scaled return process given by
\begin{equation}
\xi_{t} :=
\frac{1}{\sigma}\int_{0}^{t}\frac{\ud X_{s}}{X_{s}} =
\left(\frac{\mu_{0}}{\sigma}\right) t - \eta\int_{0}^{t}Y_{s}\ud s +
W_{t} =: \int_{0}^{t}h_{s}\ud s + W_{t} , \quad t\in\mathbf T,
\label{eq:observation}
\end{equation}
with the process $h$ defined by
\begin{equation}
h_{t} := \frac{\mu_{0}}{\sigma} - \eta Y_{t}, \quad t\in\mathbf T,
\label{eq:h}
\end{equation}
so $h$ and $W$ are independent. The process $\xi $ will be used as an
observation process in a filtering algorithm in Section
\ref{subsec:duof}.

Define the {\em observation filtration}
$\widehat{\mathbb{F}}=(\widehat{\mathcal{F}}_{t})_{t\in\mathbf{T}}$ as
the $\mathbb{P}$-augmentation of the filtration generated by the stock
price (equivalently by the process $\xi$ in \eqref{eq:observation}):
\begin{equation*}
\widehat{\mathcal{F}}_{t} := \sigma(\mathcal{F}^{X}_{t}\cup
\mathcal{N}),
\quad t\in\mathbf{T},
\end{equation*}
where $\mathcal{F}^{X}_{t}:=\sigma(X_{s}:0\leq s\leq t)$, and
$\mathcal{N}$ denotes the $\mathbb{P}$-null sets of $\mathcal{F}$. We
have $\widehat{\mathbb{F}}\subset\mathbb{F}$ and, moreover, it turns
out that the filtration $\widehat{\mathbb{F}}$ is
right-continuous,\footnote{This is a consequence of the
  strong Markov property of the pair $(X,\widehat{Y})$, where
  $\widehat{Y}$ is the filtered estimate of $Y$ given
  $\widehat{\mathbb{F}}$.} as we shall justify in Remark
\ref{rem:rc}.

An executive stock option (ESO) on $X$ is an American call option with
strike $K\geq 0$ and maturity $T$, so has payoff $(X_{t}-K)^{+}$ if
exercised at $t\in\mathbf{T}$. We assume the ESO holder receives the
cash payoff on exercise. We consider two agents in this scenario,
each of whom is awarded at time zero an ESO on $X$, and who have
access to different filtrations, but are identical in other
respects. In practice, employees holding such ESOs are prohibited from
trading the company stock $X$ (see Carpenter \cite{carpenter98} and
Section 16c of the Securities and Exchange Act), and this motivates
our assumption that neither agent trades the stock.

The first agent has {\em full information}. He knows the values of all
the model parameters and has full access to the background filtration
$\mathbb{F}$, so in particular can observe the Brownian motion $W$ and
the one-jump process $Y$. The second agent has {\em partial
  information}. She also knows the values of the constant model
parameters, and observes the stock price $X$, but not the one-jump
process $Y$. The partially informed agent's filtration is therefore
the observation filtration $\widehat{\mathbb{F}}$. The only difference
between the agents is that the partially informed agent does not know
the value of the process $Y$, which she will filter from stock price
observations.

We have assumed that the stock volatility is constant, and in
particular does not depend on the single-jump process $Y$. If we
allowed the volatility process to depend on $Y$, then with continuous
stock price observations the partially informed agent could infer the
value of $Y$ from the rate of increase of the quadratic variation of
the stock. This would remove the distinction between the agents and
thus nullify our intention of building a model where the agents have
distinctly different information on the performance of the stock. In
principle, the constant volatility assumption could be relaxed to
allow the volatility to depend on $Y$, but only at the expense of
requiring a necessarily more complicated model of differential
information between the agents. For instance, the partially informed
agent could be rendered ignorant of the values $\mu_{0},\mu_{1}$, so
these could be modelled (for example) as random variables whose values
would be filtered from price observations. This would have significant
ramifications for the tractability of the ESO optimal stopping
problems, and our constant volatility model is the simplest one can
envisage with differential information on a change point.

\subsection{The ESO optimal stopping problems}
\label{subsec:osp}

We assume that each agent will maximise, over stopping times of their
respective filtration, the discounted expectation of the ESO payoff
under the physical measure $\mathbb{P}$. Given the absence of trading
opportunities, the ESO payoff constitutes a completely unhedgeable
claim, so the agents each face a pure exercise decision. In this case,
for simplicity, we take the most straightforward objective
possible. This objective was used in Monoyios and Ng \cite{mman11},
where ESO valuation with inside information was considered. It also
appears in works which consider American options in the absence of
classical hedging opportunities, sometimes called a pure buyer's
position: an agent holds a long position in an American option but,
for reasons of (say) liquidity or transaction costs, does not hedge
this position (see Ekstr\"{o}m and Vannest\aa l \cite{ev19} for
example). If we were to allow the agents to trade other securities,
one could envisage adding risk aversion by considering utility-based
valuation and hedging, yielding combined optimal stopping and control
problems. Such ESO problems have been considered for constant drift
models by Leung and Sircar \cite{ls091,ls092} and Grasselli and
Henderson \cite{gh09} using classical utility, and by Leung, Sircar
and Zariphopoulou \cite{lsz12} using forward utility. These works take
the required regularity of value functions as given. Utility-based
valuation of European claims on non-traded assets in a random
parameter framework has been considered by Monoyios \cite{mmamf10},
where both traded and non-traded assets are geometric Brownian motions
with unobserved constant drifts modelled as Gaussian random
variables. Filtering then leads to a random parameter basis risk model
that is significantly less tractable than its constant parameter
counterpart. As both our information models have random parameters,
their rigorous treatment via a risk-averse utility-based methodology,
including verification of regularity where needed, is an open problem
left for future research. Our contribution here is thus to use our
risk-neutral objective, in a random parameter framework, to give a
fully rigorous free boundary PDE treatment of both the full and
partial information ESO problems. The absence of risk aversion in our
model gives us the tractability we need for our analysis, and arguably
focuses on the informational, as opposed to risk aversion, aspects of
the agents' exercise and valuation decisions.

For $t\in[0,T]$, let $\mathcal{T}_{t,T}$ denote the set of
$\mathbb F$-stopping times with values in $[t,T]$, and let
$\widehat{\mathcal{T}}_{t,T}$ denote the corresponding set of
$\widehat{\mathbb F}$-stopping times. For any such starting time
$t\in[0,T]$, the fully informed agent's ESO value process is $V$, an
$\mathbb{F}$-adapted process defined by
\begin{equation}
V_{t} := \esssup_{\tau\in\mathcal{T}_{t,T}}\mathbb{E}\left[
\left.\e^{-r(\tau-t)}(X_{\tau}-K)^{+}\right\vert\mathcal{F}_{t}\right],
\quad t\in[0,T]. 
\label{eq:fiproblem}
\end{equation}
We shall call (\ref{eq:fiproblem}) the full information problem.

Similarly, the partially informed agent's ESO value process is $U$, an
$\widehat{\mathbb{F}}$-adapted process defined by
\begin{equation} 
U_{t} := \esssup_{\tau\in\widehat{\mathcal{T}}_{t,T}}\mathbb{E}\left[
\left.\e^{-r(\tau-t)}(X_{\tau}-K)^{+}\right\vert
\widehat{\mathcal{F}}_{t}\right], \quad t\in[0,T].
\label{eq:piproblem}
\end{equation}
We shall call (\ref{eq:piproblem}) the partial information problem.

Naturally, the salient distinction between (\ref{eq:fiproblem}) and
(\ref{eq:piproblem}) is the filtration with respect to which the
stopping time and essential supremum are defined. For the full
information problem (\ref{eq:fiproblem}) the stock dynamics will be
(\ref{eq:dS}). For the partial information problem
(\ref{eq:piproblem}) we must derive the model dynamics under the
observation filtration. This is done in Section \ref{subsec:duof}
below.

Recipients of company ESOs are often contractually restricted from
exercising their options during a vesting period, $[0,t_v)$ so that
stopping times may lie in the interval $[t_v, T]$, see for example,
Carpenter et al.~\cite{csw2015}.  Later, in Section~\ref{sec:vesting},
we outline how the problems may be modified to incorporate vesting,
and in Section \ref{subsec:val} we demonstrate the impact of vesting
on ESO values.

\begin{remark}[Formal equivalence to random-dividend no-arbitrage
  valuation]
\label{rem:maptoclassical}
  
The optimal stopping problems \eqref{eq:fiproblem} and
\eqref{eq:piproblem}, formulated under the physical measure
$\mathbb{P}$ with some random stock drift $\mu(\cdot)$, of course map
formally to problems written under a martingale measure $\mathbb{Q}$
where the stock drift will be $r-\delta(\cdot)$, for some random
dividend yield $\delta(\cdot)$, related to $\mu(\cdot)$ by
$\mu(\cdot)=r-\delta(\cdot)$. The results we obtain are thus
applicable to classical no-arbitrage valuation with a random dividend
yield.

\end{remark}

The scenario we have set up, with a drift value for a log-Brownian
motion which switches at a random time to a new value, has obvious
similarities with the so-called ``quickest detection of a Wiener
process'' problem, which has a long history and is discussed in
Chapter VI of Peskir and Shiryaev \cite{ps06} (see Gapeev and Shiryaev
\cite{gs13} for a recent example involving diffusion processes). The
difference between these problems and ours is that our objective
functional will be the expected discounted payoff of an ESO, so errors
in detecting the change point are transmitted through the prism of the
ESO exercise decision. In contrast, the classical change point
detection problem has some explicit objective functional which
directly penalises a detection delay or a false alarm (where the
change point is incorrectly deduced to have occurred).

\subsection{Dynamics under the observation filtration}
\label{subsec:duof}

Let the signal process be $Y$ in (\ref{eq:signal}), and take the
observation process to be $\xi$ in (\ref{eq:observation}), with the
augmented filtration generated by $\xi$ equivalent to the augmented
stock price filtration $\widehat{\mathbb{F}}$.

Introduce the notation
$\widehat{\phi}_{t}:=\mathbb{E}[\phi_{t}|\widehat{\mathcal{F}}_{t}]$,
$t\in\mathbf{T}$, for any process $\phi$. In particular, we are
interested in the filtered estimate of $Y$, defined by
\begin{equation*}
\widehat{Y}_{t} := \mathbb{E}[Y_{t}|\widehat{\mathcal{F}}_{t}], \quad
t\in\mathbf{T}.  
\end{equation*}
A standard filtering procedure gives the stock price dynamics with
respect to the observation filtration $\widehat{\mathbb{F}}$, along
with the dynamics of $\widehat{Y}$, resulting in the following
lemma. We give a short proof for completeness.

\begin{lemma}[Observation filtration dynamics]
\label{lem:ofd}

With respect to the observation filtration $\widehat{\mathbb{F}}$ the
stock price follows
\begin{equation}
\ud X_{t} = (\mu_{0} - \sigma\eta\widehat{Y}_{t})X_{t}\ud t + \sigma
X_{t}\ud\widehat{W}_{t}, 
\label{eq:dSobs}  
\end{equation}
where $\widehat{W}$ is the innovations process, given by
\begin{equation}
\widehat{W}_{t} := \xi_{t} - \int_{0}^{t}\widehat{h}_{s}\ud s =
\xi_{t} -
\frac{\mu_{0}}{\sigma}t +
\eta\int_{0}^{t}\widehat{Y}_{s}\ud s, \quad t\in\mathbf{T},   
\label{eq:innovation}
\end{equation}
where analogously to (\ref{eq:h}),
$\widehat{h}_{t}:=\frac{\mu_{0}}{\sigma}-\eta\widehat{Y}_{t}$,
$t\in\mathbf{T}$, and $\widehat{W}$ is a
$(\mathbb{P},\mathbb{\widehat{F}})$-Brownian motion.

The filtered process $\widehat{Y}$ has dynamics given by
\begin{equation}
\ud\widehat{Y}_{t} = \lambda(1-\widehat{Y}_{t})\ud t -
\eta\widehat{Y}_{t}(1-\widehat{Y}_{t})\ud\widehat{W}_{t}, \quad
\widehat{Y}_{0} = \mathbb{E}[Y_{0}] = y_{0} \in[0,1). 
\label{eq:whY}
\end{equation}

\end{lemma}

\begin{proof}

We use the innovations approach to filtering, as discussed in Rogers
and Williams \cite{rw00}, Chapter VI.8 or Bain and Crisan
\cite{bc09}, Chapter 3, for instance.

By Theorem VI.8.4 in \cite{rw00}, the innovations process
$\widehat{W}$, defined by (\ref{eq:innovation}), is a
$(\mathbb{P},\mathbb{\widehat{F}})$-Brownian motion. Using
(\ref{eq:innovation}) in the stock price SDE (\ref{eq:dX}) then yields
(\ref{eq:dSobs}).

It remains to prove (\ref{eq:whY}). For any bounded, measurable test
function $f$, write $f_{t}\equiv f(Y_{t})$, $t\in\mathbf{T}$, for
brevity. Define a process $(\mathcal{G}f_{t})_{t\in\mathbf{T}}$,
satisfying
$\mathbb{E}\left[\int_{0}^{t}|\mathcal{G}f_{s}|^{2}\ud
s\right]<\infty$ for all $t\in\mathbf{T}$, such that
\begin{equation*}
M^{(f)}_{t} := f_{t} - f_{0} - \int_{0}^{t}\mathcal{G}f_{s}\ud
s, \quad t\in\mathbf{T},  
\end{equation*}
is a $(\mathbb{P},\mathbb{F})$-martingale. With $h,W$ independent, we
have the (Kushner-Stratonovich) fundamental filtering equation (see
Theorem 3.30 in \cite{bc09}, for example)
\begin{equation}
\widehat{f}_{t} = \widehat{f}_{0} +
\int_{0}^{t}\widehat{\mathcal{G}f}_{s}\ud s +
\int_{0}^{t}\left(\widehat{f_{s}h_{s}} -
\widehat{f}_{s}\widehat{h}_{s}\right)\ud\widehat{W}_{s}, \quad
t\in\mathbf{T}.
\label{eq:ffeqn} 
\end{equation}
Take $f(y)=y$. Then the martingale $M^{(f)}=M^{(Y)}$, as defined in
(\ref{eq:M}), so that $\mathcal{G}f=\lambda(1-Y)$ and the filtering
equation (\ref{eq:ffeqn}) reads as
\begin{equation}
\widehat{Y}_{t} = y_{0} + \lambda\int_{0}^{t}(1-\widehat{Y}_{s})\ud s +
\int_{0}^{t}(\widehat{Y_{s}h_{s}} -
\widehat{Y}_{s}\widehat{h}_{s})\ud\widehat{W}_{s}, \quad t\in\mathbf{T},  
\label{eq:feqn}
\end{equation}
where we have used $\widehat{Y}_{0}=\mathbb{E}[Y_{0}]=y_{0}$.

Now, 
\begin{equation}
\widehat{Y_{t}h_{t}} =
\mathbb{E}\left[\left.Y_{t}\left(\frac{\mu_{0}}{\sigma} -\eta 
Y_{t}\right)\right\vert\widehat{\mathcal{F}}_{t}\right] =
\left(\frac{\mu_{0}}{\sigma}\right)\widehat{Y}_{t} - 
\eta\mathbb{E}[Y^{2}_{t}|\widehat{\mathcal{F}}_{t}] =
\left(\frac{\mu_{0}}{\sigma} - \eta\right)\widehat{Y}_{t}, \quad
t\in\mathbf{T},   
\label{eq:whYh}
\end{equation}
the last equality a consequence of $Y^{2}=Y$. 

On the other hand,
\begin{equation}
\widehat{Y}_{t}\widehat{h}_{t} =
\widehat{Y}_{t}\mathbb{E}\left[\left.\frac{\mu_{0}}{\sigma} -\eta 
Y_{t}\right\vert\widehat{\mathcal{F}}_{t}\right] =
\left(\frac{\mu_{0}}{\sigma}\right)\widehat{Y}_{t} -
\eta\left(\widehat{Y}_{t}\right)^{2}, \quad t\in\mathbf{T}.
\label{eq:whYwhh}
\end{equation}
Using (\ref{eq:whYh}) and (\ref{eq:whYwhh}) in (\ref{eq:feqn}) then
yields the integral form of (\ref{eq:whY}).

\end{proof}

\begin{remark}[Right-continuity of observation filtration]
\label{rem:rc}

Note that $\widehat{Y}$ in (\ref{eq:whY}) is an
$\widehat{\mathbb{F}}$-adapted diffusion in $[0,1]$ with an absorbing
state at $\widehat{Y}=1$. Note also that, since observations of the
stock price are sufficient to specify $\widehat{Y}$, the observation
filtration is also the $\mathbb{P}$-augmentation of the filtration
generated by the two-dimensional diffusion $(X,\widehat{Y})$. Then,
Karatzas and Shreve \cite[Proposition 2.7.7]{ks91} guarantees that
$\widehat{\mathbb{F}}$ is right-continuous, as it is the augmented
filtration generated by the Strong Markov Process $(X,\widehat{Y})$.
  
\end{remark}

\section{The full information ESO problem}
\label{sec:fiesop}

In this section we focus on the full information problem defined in
(\ref{eq:fiproblem}). Define the (continuous) \emph{reward process}
$R$ as the discounted payoff process:
\begin{equation}
R_{t} := \e^{-rt}(X_{t}-K)^{+}, \quad t\in\mathbf{T}.
\label{eq:reward}
\end{equation}
The reward process is assumed to satisfy 
\begin{equation}
\mathbb{E}\left[\sup_{t\in[0,T]}R_{t}\right] < \infty.  
\label{eq:wellposed}
\end{equation}
The discounted full information ESO value process is $\widetilde{V}$,
given by
\begin{equation}
\widetilde{V}_{t} := \e^{-rt}V_{t} =
\esssup_{t\in\mathcal{T}_{t,T}}\mathbb{E}[R_{\tau}|\mathcal{F}_{t}],
\quad t\in\mathbf{T}. 
\label{eq:fisnell}
\end{equation}

Classical optimal stopping theory for continuous time processes, as
described in Karatzas and Shreve \cite[Appendix D]{ks98},
characterises the solution to the problem (\ref{eq:fisnell}) as
follows. First, by \cite[Proposition D.2]{ks98}, $\widetilde{V}$ is a
$(\mathbb{P},\mathbb{F})$-super-martingale. Further, by
\cite[Proposition D.3 and Corollary D.4]{ks98}, there exists a
c\`adl\`ag modification $\widetilde{V}^{0}$ of $\widetilde{V}$, called
the \emph{Snell envelope} of $R$, that by \cite[Theorem D.7]{ks98}
satisfies $\widetilde{V}^{0}_{t}=\widetilde{V}_{t}$ almost surely, for
all $t\in[0,T]$, and is the smallest c\`adl\`ag
$(\mathbb{P},\mathbb{F})$-super-martingale that dominates (in the
sense of \cite[Definition D.5]{ks98}, so
$\mathbb{P}[\widetilde{V}^{0}_{t}\geq R_{t},\forall\, 0\leq t\leq
T]=1$) the reward $R$. Then, by \cite[Theorem D.9]{ks98}, a stopping
time $\tau^{*}\in\mathcal{T}$ is optimal for the problem
(\ref{eq:fisnell}) starting at time zero if and only if
$\widetilde{V}^{0}_{\tau^{*}}=R_{\tau^{*}}$ almost surely, and if and
only if the stopped super-martingale
$(\widetilde{V}^{0}_{\tau^{*}\wedge t})_{t\in[0,T]}$, is a
$(\mathbb{P},\mathbb{F})$-martingale. Finally, under
\eqref{eq:wellposed} and with a continuous reward process,
\cite[Theorem D.12]{ks98} gives that the smallest optimal stopping
time in $\mathcal{T}_{t,T}$ for the problem (\ref{eq:fisnell}) is
$\tau^{*}(t)$, the first time that the Snell envelope coincides with
the reward, so is given by
\begin{equation}
\tau^{*}(t) :=
\inf\{\tau\in[t,T): \widetilde{V}^{0}_{\tau} = R_{\tau}\}\wedge T,
\quad t\in[0,T].
\label{eq:fiost}
\end{equation}
Given this characterisation of the full information ESO problem via
the Snell envelope, from now on we identify the discounted ESO value
process with the Snell envelope, and adopt the standard notational
convention of not distinguishing between them, so
$\widetilde{V}\equiv\widetilde{V}^{0}$. The ESO value process is then
given by $V_{t}=\e^{rt}\widetilde{V}_{t},t\in[0,T]$, with the
understanding that $\widetilde{V}$ is the Snell envelope of the
reward. With this standard convention, the optimal stopping time in
\eqref{eq:fiost} is given by the first time the ESO value process hits
the payoff:
\begin{equation*}
\tau^{*}(t) =
\inf\{\tau\in[t,T): V_{\tau} = (X_{\tau}-K)^{+}\}\wedge T,
\quad t\in[0,T].
\end{equation*}

\subsection{Full information value function}
\label{subsec:fivf}

Introduce the value function
$v:[0,T]\times\mathbb R_{+}\times\{0,1\}\to\mathbb R_{+}$ for the full
information optimal stopping problem (\ref{eq:fiproblem}) as
\begin{equation}
v(t,x,i) := \sup_{\tau\in\mathcal{T}_{t,T}}\mathbb
E\left[\left.{\rm e}^{-r(\tau-t)}(X_{\tau}-K)^{+}\right\vert 
X_{t}=x,Y_{t}=i\right], \quad i=0,1, \quad t\in[0,T],
\label{eq:fivf}
\end{equation}
and write $v_{i}(\cdot,\cdot)\equiv v(\cdot,\cdot,i)$, $i=0,1$. Thus,
the value function in the full information scenario is a pair of
functions of time and current stock price, such that $v_{0}(t,x)$
(respectively, $v_{1}(t,x)$) represents the value of the ESO to the
insider at time $t\in[0,T]$ given $X_{t}=x$ and $Y_{t}=0$
(respectively, $Y_{t}=1$). In other words, the value process $V$ in
(\ref{eq:fiproblem}) has the representation
\begin{equation}
V_{t} = v(t,X_{t},Y_{t}) = (1-Y_{t})v_{0}(t,X_{t}) +
Y_{t}v_{1}(t,X_{t}), \quad t\in[0,T].  
\label{eq:fivp}
\end{equation}
Very general results on optimal stopping in a continuous-time Markov
setting (see for instance El Karoui, Lepeltier and Millet
\cite{elklm92}) imply that each $v_{i}(\cdot,\cdot)$, $i=0,1$, is a
continuous function of time and current stock price, and the process
$({\rm e}^{-rt}v(t,X_{t},Y_{t}))_{t\in[0,T]}$ is the Snell envelope of
the reward process $R$.

In what follows, we first establish, in Lemma \ref{lem:cmtdfi}, some
elementary properties of the full information value function, so as to
then characterise the nature of the continuation and stopping regions
in Corollary \ref{corr:fiexbds}. As we shall see, the two-drift model
leads to two ordered exercise thresholds
$x^{*}_{i}:[0.T]\to[K,\infty),\,i=0,1$, and we shall establish that
these thresholds are right-continuous on $[0,T)$. Later, using the
free boundary system (Proposition \ref{prop:fbpfivf}) and smooth
pasting property (Theorem \ref{thm:spfi}) satisfied by the value
function, as well as the Doob-Meyer decomposition of the
super-martingale characterising the discounted ESO value process
(Theorem \ref{thm:dmdecomp}) we shall obtain the limiting values
$x^{*}_{i}(T-)$ of the exercise boundaries, given in Proposition
\ref{prop:tvfull}, where we also show that the exercise boundaries are
continuous on $[0,T)$.

With respect to $\mathbb{F}$, the dynamics of the stock are given in
(\ref{eq:dS}). For $0\leq s\leq t\leq T$, define the accumulation
factor
\begin{equation}
H_{s,t} :=  \exp\left\{\left(\mu_{0} - \frac{1}{2}\sigma^{2}\right)(t
- s) - \sigma\eta\int_{s}^{t}Y_{u}\ud u + \sigma(W_{t} -
W_{s})\right\}, \quad 0\leq s\leq t\leq T.   
\label{eq:Hst}
\end{equation}
Then, given $X_{s}=x\in\mathbb R_{+}$, the stock price at $t\in[s,T]$
is $X_{t}\equiv X^{s,x}_{t}$, given by
\begin{equation*}
X_{t}\equiv X^{s,x}_{t} = xH_{s,t}, \quad 0\leq s\leq t\leq T.
\end{equation*}
When $s=0$, write $H_{t}\equiv H_{0,t}$ and $X^{x}_{t}\equiv
X^{0,x}_{t}$, so that 
\begin{equation*}
X^{x}_{t} = xH_{t}, \quad t\in[0,T].  
\end{equation*}
For use further below, also define the accumulation factor when the
stock is exclusively in state $i\in\{0,1\}$, by
\begin{equation}
H^{(i)}_{s,t} :=  \exp\left\{\left(\mu_{i} -
\frac{1}{2}\sigma^{2}\right)(t - s) + \sigma(W_{t} - W_{s})\right\},
\quad 0\leq s\leq t\leq T, \quad i=0,1.
\label{eq:Hi}
\end{equation}
and as before, for $s=0$ write $H^{(i)}_{t}\equiv
H^{(i)}_{0,t},\,i=0,1$ for $t\in[0,T]$.

Note, in particular, that if the stock starts at time zero at
$X_{0}=x$, and the change point occurs in $[0,T]$, then the stock
price at $t\in[\theta,T]$ (so at or beyond the change point), is
$X_{t}\equiv X^{x}_{t}$ given by
\begin{equation}
X_{t} = x\exp(\sigma\eta\theta)H^{(1)}_{t}, \quad 0\leq \theta \leq
t\leq T.
\label{eq:Xacp}
\end{equation}

With these definitions in place, the value function in (\ref{eq:fivf})
is expressed in the form
\begin{equation}
v_{i}(t,x) = \sup_{\tau\in\mathcal{T}_{t,T}}\mathbb{E}\left[
\left.\e^{-r(\tau-t)}(xH_{t,\tau} - K)^{+}\right\vert Y_{t}=i\right],
\quad (t,x)\in[0,T]\times\mathbb R_{+}, \quad i=0,1,      
\label{eq:fivfrep} 
\end{equation}
where $H_{t,\tau}$ is the process in \eqref{eq:Hst} over the interval
$[t,\tau]$:
\begin{equation}
H_{t,\tau} :=  \exp\left\{\left(\mu_{0} -
\frac{1}{2}\sigma^{2}\right)(\tau - t) -
\sigma\eta\int_{t}^{\tau}Y_{u}\ud u + \sigma(W_{\tau} - 
W_{t})\right\}, \quad \tau\in[t,T].
\label{eq:Httau}
\end{equation}
Now, the Brownian increment $W_{\tau}-W_{t}$ in the interval
$[t,\tau]$ is identical in Law to
$W_{\tau-t}-W_{0}=W_{\tau-t}$. Further, the integral over $Y$ in
\eqref{eq:Httau} may be re-written according to
$\int_{t}^{\tau}Y_{u}\ud u=\int_{0}^{\tau-t}Y_{t+s}\ud s$, and the
absence of memory property of the exponential distribution
($\mathbb{P}[\theta>t+s|\theta>t]=\mathbb{P}[\theta>s]$ for any
$s,t\geq 0$) means that
$\Law(Y_{t+s}|Y_{t}=i)=\Law(Y_{s}|Y_{0}=i)$. Therefore, in
\eqref{eq:fivfrep}, the integral of $Y$ over $[t,\tau]$ with
conditioning on the value of $Y_{t}$ may be replaced by one over
$[0,\tau-t]$ with conditioning on the value of $Y_{0}$.  In other
words, stationarity of Brownian increments and the memoryless property
of the exponential distribution imply that optimising over
$\mathcal{T}_{t,T}$ is equivalent to optimising over
$\mathcal{T}_{0,T-t}$, so the value function in \eqref{eq:fivfrep} may
be re-cast into the form
\begin{equation}
v_{i}(t,x) = \sup_{\tau\in\mathcal{T}_{0,T-t}}\mathbb
E\left[\left.\e^{-r\tau}(xH_{\tau} - K)^{+}\right\vert Y_{0}=i\right],
\quad (t,x)\in[0,T]\times\mathbb R_{+}, \quad i=0,1.      
\label{eq:fivfrep2}
\end{equation}
Thus, the ESO value with maturity $T$ and starting time $t\in[0,T]$ is
the same as the ESO value with maturity $T-t$ and initial time zero.
This re-casting of the ESO value will be helpful below in
demonstrating some properties of the value function, and is frequently
utilised in American option valuation problems (see for example the
proof of Proposition 31 in Detemple \cite[Chapter 4]{detemple06} for
the same re-casting in the (simpler) case of a stock with constant
drift).

The following lemma gives the elementary properties of the full
information value function.

\begin{lemma}[Convexity, monotonicity, time decay: full information]
\label{lem:cmtdfi}

The functions
$v(\cdot,\cdot,i)\equiv v_{i}:[0,T]\times\mathbb R_{+},i=0,1$ in
\eqref{eq:fivfrep2} or \eqref{eq:fivf} characterising the full
information ESO value function (and the ESO value process via
(\ref{eq:fivp})) have the following properties:

\begin{enumerate}

\item For $i=0,1$ and $t\in[0,T]$, the map $x\to v_{i}(t,x)$ is convex
and non-decreasing. 

\item For any fixed $(t,x)\in[0,T]\times\mathbb{R}_{+}$,
$v_{0}(t,x)\geq v_{1}(t,x)$.

\item For $i=0,1$ and $x\in\mathbb R_{+}$, the map $t\to v_{i}(t,x)$
is non-increasing.

\end{enumerate}  

\end{lemma}

\begin{proof}

\begin{enumerate}

\item Convexity and monotonicity of the map $x\to v_{i}(t,x)$ follow
from the representation \eqref{eq:fivfrep2}, along with convexity
and monotonicity properties of the payoff function $x\to(x-K)^{+}$
and the linearity of the map $x\to X^{x}_{\tau}=xH_{\tau}$. For
example, to show convexity, consider $0\leq x_{1}<x_{2}<\infty$ and
some $\gamma\in[0,1]$. For each $i\in\{0,1\}$ we then have, on using
\eqref{eq:fivfrep2}, that
\begin{eqnarray*}
&& \gamma v_{i}(t,x_{1}) + (1-\gamma)v_{i}(t,x_{2}) \\
& = & \sup_{\tau\in\mathcal{T}_{0,T-t}}\mathbb{E}\left[\left.
\e^{-r\tau}\left(\gamma(x_{1}H_{\tau}-K)^{+} +
(1-\gamma)(x_{2}H_{\tau}-K)^{+}\right)\right\vert Y_{0}=i\right] \\
& \geq & \sup_{\tau\in\mathcal{T}_{0,T-t}}\mathbb{E}\left[\left.  
\e^{-r\tau}\left((\gamma x_{2} +
(1-\gamma)x_{2})H_{\tau}-K\right)^{+}\right\vert Y_{0}=i\right] \\
& = & v_{i}(t,\gamma x_{1} + (1-\gamma)x_{2}),   
\end{eqnarray*}
where the inequality follows from convexity of the payoff
function. This establishes convexity of $x\to
v_{i}(t,x)$. Monotonicity is established in the same manner.

\item At maturity we have $v_{0}(T,x)=v_{1}(T,x)=(x-K)^{+}$ for all
$x\in\mathbb{R}_{+}$. For $t\in[0,T)$, using the representation
\eqref{eq:fivfrep} and the definition \eqref{eq:Hi} for $i=0$ we
have
\begin{eqnarray}
v_{0}(t,x) & = & \sup_{\tau\in\mathcal{T}_{t,T}}\mathbb{E}\left[
\left.\e^{-r(\tau-t)}(xH_{t,\tau} - K)^{+}\right\vert Y_{t}=0\right]
\nonumber \\
& = & \sup_{\tau\in\mathcal{T}_{t,T}}\mathbb{E}\left[
\left.\e^{-r(\tau -
t)}\left(xH^{(0)}_{t,\tau}\exp\left(-\sigma\eta\int_{t}^{\tau}Y_{u}\ud 
u\right) - K\right)^{+}\right\vert Y_{t}=0\right].
\label{eq:v0rep}
\end{eqnarray}
Now, if $Y_{t}=0$ (so $\theta>t$) then for any $\mathbb{F}$-stopping
time $\tau\in[t,T)$ we have
$\int_{t}^{\tau}Y_{u}\ud
u=(\tau-\theta)\mathbbm{1}_{\{\tau\geq\theta\}}\leq \tau-t$, which
implies that
\begin{equation*}
H_{t,\tau}\equiv
H^{(0)}_{t,\tau}\exp\left(-\sigma\eta\int_{t}^{\tau}Y_{u}\ud u\right)
\geq H^{(0)}_{t,\tau}\e^{-\sigma\eta(\tau-t)}=H^{(1)}_{t,\tau}.  
\end{equation*}
Using this in the representation \eqref{eq:v0rep} we have
\begin{equation*}
v_{0}(t,x) \geq \sup_{\tau\in\mathcal{T}_{t,T}}\mathbb{E}\left[\left.
\e^{-r(\tau-t)}(xH^{(1)}_{t,\tau}-K)^{+}\right\vert Y_{t}=0\right].
\end{equation*}
But $xH^{(1)}_{t,\tau}$ is also the value of the stock at time $\tau$
given $X_{t}=x$ and $Y_{t}=1$ (since the drift appearing in $H^{(1)}$
is $\mu_{1}$), so we have
\begin{equation*}
v_{0}(t,x) \geq \sup_{\tau\in\mathcal{T}_{t,T}}\mathbb{E}\left[
\left.\e^{-r(\tau - t)}(xH_{t,\tau}-K)^{+}\right\vert
Y_{t}=1\right]  = v_{1}(t,x), \quad t\in[0,T).  
\end{equation*}

\item This is the classical time decay property of American claims,
which follows from the representation \eqref{eq:fivfrep2} and the
fact that $\mathcal{T}_{0,T-t^{\prime}}\subseteq\mathcal{T}_{0,T-t}$
for $t^{\prime}\geq t$. That is, given the time-homogeneity of the
stock price model (that is, the absence of explicit time dependence in
the model parameters), the possible stopping strategies starting at
the later time $t^{\prime}$ are a subset of the available strategies
starting at an earlier time, leading immediately to
$v_{i}(t^{\prime},x)\leq v_{i}(t,x)$ for any fixed $x$ and
$t^{\prime}\geq t$. This time decay property is well-known to hold in
time-homogeneous models, as discussed by Ekstr\"om \cite{ekstrom04}
and Monoyios and Ng \cite{mman11}.
  
\end{enumerate}

\end{proof}

\subsection{Full information continuation and stopping
  regions}
\label{subsec:ficsr}

Define the continuation regions $\mathcal{C}_{i}$ and stopping regions
$\mathcal{S}_{i}$ when the one-jump process $Y$ is in state
$i\in\{0,1\}$ by
\begin{eqnarray*}
\mathcal{C}_{i} :=
\{(t,x)\in[0,T)\times\mathbb{R}_{+}:v_{i}(t,x)>(x-K)^{+}\}, \quad
i=0,1, \\ 
\mathcal{S}_{i} :=
\{(t,x)\in[0,T)\times\mathbb{R}_{+}:v_{i}(t,x)=(x-K)^{+}\}, \quad
i=0,1. 
\end{eqnarray*}
Since the functions $v_{i}(\cdot,\cdot) $ are continuous, the
continuation regions $\mathcal{C}_{i},\,i=0,1$ are open sets and their
respective complements $\mathcal{S}_{i},\,i=0,1$ are closed sets. At
maturity, by definition one cannot continue, so exercise takes place
if the terminal stock price exceeds the strike.

\begin{remark}[Minimal conditions for early exercise: full information]
\label{rem:murfi}

If the drift process $\mu(Y)$ of the stock in \eqref{eq:mufi}
satisfies $\mu(Y)\geq r$ almost surely, then the reward process is a
$(\mathbb{P},\mathbb{F})$-sub-martingale, so no early exercise is
optimal, and the American ESO value coincides with that of its
European counterpart. In particular, if $\mu_{0}\geq r$, then we
expect no early exercise when $Y=0$ (so before the change point).
 
\end{remark}

The properties in Lemma \ref{lem:cmtdfi} imply that for each $i=0,1$,
the boundary between $\mathcal{C}_{i},\mathcal{S}_{i}$ will take the
form of a non-increasing critical stock price function (or exercise
boundary) $x^{*}_{i}:[0,T)\to[K,\infty)$, with
$x^{*}_{0}(t)\geq x^{*}_{1}(t)\geq K$ for all $t\in[0,T)$. The optimal
exercise policy when $Y$ is in state $i\in\{0,1\}$ is to exercise the
ESO the first time the stock price crosses $x^{*}_{i}(\cdot)$ from
below, unless the change point occurs at a juncture when the exercise
boundaries are strictly ordered and the stock price satisfies
$x^{*}_{1}(\theta)\leq X_{\theta}<x^{*}_{0}(\theta)$, in which case
the change point causes the system to immediately switch from being in
$\mathcal{C}_{0}$ to $\mathcal{S}_{1}$, and the ESO is exercised
immediately after the change point. At the maturity time itself,
exercise takes place if the terminal stock price exceeds the strike,
so the exercise boundaries may be extended to maturity by defining
$x^{*}_{i}(T):=K,\,i=0,1$ (though as we shall see shortly in
Proposition \ref{prop:tvfull} there exists the possibility of a
discontinuity in the boundaries at maturity, with $x^{*}_{i}(T-)$
possibly not equal to $K$). We formalise these properties in the
corollary below.

\begin{corollary}
\label{corr:fiexbds}

For $i=0,1$, if $\mu_{i}<r$, then there exist two non-increasing
right-continuous functions $x^{*}_{i}:[0,T)\to[K,\infty)$, $i=0,1$,
satisfying
\begin{equation}
x^{*}_{1}(t) \leq x^{*}_{0}(t), \quad t\in[0,T),  
\label{eq:ordering}
\end{equation}
such that the continuation and stopping regions in state $i\in\{0,1\}$
are given by
\begin{eqnarray}
\mathcal{C}_{i} = \{(t,x)\in[0,T)\times\mathbb
R_{+}:x < x^{*}_{i}(t)\}, \quad i=0,1, \label{eq:Ci} \\
\mathcal{S}_{i} = \{(t,x)\in[0,T)\times\mathbb
R_{+}:x \geq x^{*}_{i}(t)\}, \quad i=0,1. \label{eq:Si}
\end{eqnarray}
The smallest optimal stopping time for the full information problem
\eqref{eq:fiproblem} starting at time zero is
$\tau^{*}(0)\equiv\tau^{*}$, given by
\begin{equation*}
\tau^{*} = \inf\left\{t\in[0,T):\mathbbm{1}_{\{Y_{t}=0\}}X_{t}\geq
x^{*}_{0}(t) + \mathbbm{1}_{\{Y_{t}=1\}}X_{t}\geq x^{*}_{1}(t)\right\}\wedge
T.
\end{equation*}
For $i=0,1$, if $\mu_{i}\geq r$, then the exercise thresholds satisfy
$x^{*}_{i}(t)=+\infty$ for $t\in[0,T)$, in accordance with Remark
\ref{rem:murfi}.

At maturity, regardless of the values of $\mu_{i},\,i=0,1$, we have
$x^{*}_{i}(T)=K,\,i=0,1$.

\end{corollary}

Before giving the proof of this corollary, we state in Proposition
\ref{prop:tvfull} below some further properties of the exercise
boundaries which it is natural to give here, and which we shall prove
later, after establishing free boundary PDEs and smooth pasting
properties for the value functions in Sections \ref{subsec:fifbs} and
\ref{subsec:fisp}, along with the Doob-Meyer decomposition of the
Snell envelope of the reward process (that is, the discounted full
information ESO process) in Section \ref{subsec:dmdfise}.

When $\mu_{i}<r$, $i=0,1$, so that bounded exercise thresholds exist
prior to maturity, it turns out that the exercise boundaries are
continuous over $[0,T)$, with a possible discontinuity at $T$, as we
show below in Proposition \ref{prop:tvfull}. This mirrors the
classical situation in the Black-Scholes model for an American call,
in which the critical stock price satisfies
$x^{*}_{\mathrm{BS}}(T-)=\max(K,(r/\delta)K)$ and
$x^{*}_{\mathrm{BS}}(T)=K$, where $\delta$ is the dividend yield (see
for example Detemple \cite[Chapter 4, Proposition
33]{detemple06}). The proposition below shows that these formulae
extend to the random dividend yield case, where the dividend yield
can switch from its initial value to another, and where we invoke
Remark \ref{rem:maptoclassical} to map our problem to a classical
no-arbitrage valuation of an American call. A similar remark will
pertain to the partial information problem as well, where the random
dividend yield will depend on a diffusion with values in $[0,1]$.

\begin{proposition}
\label{prop:tvfull}

Suppose, for $i=0,1$, that $\mu_{i}<r$. The optimal exercise
boundaries $x^{*}_{i}(\cdot),\,i=0,1$ for the full information ESO
problem are continuous over $[0,T)$, with limiting values as we
approach maturity given by
\begin{equation}
\lim_{t\uparrow T}x^{*}_{i}(t) \equiv x^{*}_{i}(T-) =
\max\left(K,\frac{r}{r-\mu_{i}}K\right), \quad i=0,1.   
\label{eq:xistm}
\end{equation}
At maturity itself, we have $x^{*}_{i}(T)=K$, for $i=0,1$.

\end{proposition}

The proof of this proposition will be given later in Section
\ref{subsec:dmdfise}, after we establish the free boundary PDE for the
full information value function in Proposition~\ref{prop:fbpfivf}, the
smooth pasting condition in Theorem \ref{thm:spfi}, as well as the
Doob-Meyer decomposition of the Snell envelope process in Theorem
\ref{thm:dmdecomp}, these results being utilised in the proof of
Proposition~\ref{prop:tvfull}.

We now turn to proving Corollary \ref{corr:fiexbds}.
  
\begin{proof}[Proof of Corollary \ref{corr:fiexbds}]

For $i=0,1$, take $\mu_{i}<r$, as the case $\mu_{i}\geq r$ is
covered by Remark \ref{rem:murfi}. First, if early exercise has not
occurred prior to maturity, then it will occur at maturity provided
the stock price is not below the strike, so we have terminal critical
stock prices $x^{*}_{i}(T)=K,\,i=0,1$.

Next, let us show that the continuation and stopping regions have the
threshold forms shown in in \eqref{eq:Ci} and \eqref{eq:Si},
respectively. Fix $i\in\{0,1\}$ and $t\in[0,T)$, and suppose that
$(t,x)\in[0,T)\times\mathbb{R}_{+}$ is such that
$(t,x)\in\mathcal{S}_{i}$, so we have $v_{i}(t,x)=x-K$. Now take
$\bar{x}>x$. We want to show that
$(t,\bar{x})\in\mathcal{S}_{i}$. Suppose, to the contrary, that
$(t,\bar{x})\notin\mathcal{S}_{i}$, so that
$v_{i}(t,\bar{x})>\bar{x}-K$. But we also have, with $\bar{\tau}$
denoting the time interval to the optimal exercise time for starting
state $(t,\bar{x},i)$ in the representation \eqref{eq:fivfrep2}, that
\begin{eqnarray*}
v_{i}(t,\bar{x}) & = &
\mathbb{E}\left[\left.\e^{-r\bar{\tau}}(\bar{x}H_{\bar{\tau}}-K)^{+}\right\vert
Y_{0}=i\right] \\
& = & \mathbb{E}\left[\left.\e^{-r\bar{\tau}}\left(xH_{\bar{\tau}} +
(\bar{x}-x)H_{\bar{\tau}}-K\right)^{+}\right\vert Y_{0}=i\right] \\
& \leq & \mathbb{E}\left[\left.\e^{-r\bar{\tau}}(xH_{\bar{\tau}}-K)^{+}\right\vert
Y_{0}=i\right] + \mathbb{E}[\left.
\e^{-r\bar{\tau}}(\bar{x}-x)H_{\bar{\tau}}\right\vert Y_{0}=i] \\
& \leq & v_{i}(t,x) +
(\bar{x}-x)\mathbb{E}[\left.\e^{-r\bar{\tau}}H_{\bar{\tau}}\right\vert
Y_{0}=i] \\
& < & v_{i}(t,x) + \bar{x}-x \\
& = & \bar{x} - K.
\end{eqnarray*}
Above, the first inequality follows from the inequality
$(a+b)^{+}\leq a^{+}+b^{+}$, the second inequality follows from the
sub-optimality of $\bar{\tau}$ for starting state $(t,x,i)$, and the
third inequality is due to the strict super-martingale property of
$(\e^{-rt}H_{t})_{t\in[0,T]}$ when $\mu_{i}<r$, which we now show.

If $Y_{0}=0$, then for $t\in[0,T]$ we have, with $\mathcal{E}(\cdot)$
denoting the stochastic exponential,
\begin{equation*}
\e^{-rt}H_{t} = \e^{-(r-\mu_{0})t}\mathcal{E}(\sigma
W)_{t}\exp\left(-\sigma\eta\int_{0}^{t}Y_{s}\ud s\right) \leq
\e^{-(r-\mu_{0})t}\mathcal{E}(\sigma W)_{t}, \quad t\in[0,T],
\end{equation*}
which for $\mu_{0}<r$ yields a strict super-martingale. If $Y_{0}=1$
the argument is yet simpler, as in that case we obtain
\begin{equation*}
\e^{-rt}H_{t} = \e^{-(r-\mu_{1})t}\mathcal{E}(\sigma W)_{t}, \quad t\in[0,T],
\end{equation*}
again yielding a strict super-martingale. We thus obtain
$v_{i}(t,\bar{x})<\bar{x}-K$, which contradicts
$v_{i}(t,\bar{x})>\bar{x}-K$. Hence, $(t,\bar{x})\in\mathcal{S}_{i}$,
which establishes \eqref{eq:Ci} and \eqref{eq:Si}.

Next, let us show that the exercise boundaries are non-increasing. Fix
$i\in\{0,1\}$ and $(t,x)\in(0,T)\times\mathbb{R}_{+}$ such that
$(t,x)\in\mathcal{C}_{i}$, so that $v_{i}(t,x)>(x-K)^{+}$ and
$x<x^{*}_{i}(t)$. Consider a time $t_{0}$ satisfying
$0\leq t_{0}<t<T$. By the time decay property in Lemma
\ref{lem:cmtdfi} we have $v_{i}(t_{0},x)\geq v_{i}(t,x)$, and
therefore,
\begin{equation*}
v_{i}(t_{0},x) - (x-K)^{+} \geq v_{i}(t,x) - (x-K)^{+} > 0,
\end{equation*}
so that we also have $(t_{0},x)\in\mathcal{C}_{i}$. In other words,
$x<x^{*}_{i}(t)\implies x<x^{*}_{i}(t_{0})$, which can only be true if
$x^{*}_{i}(\cdot)$ is non-increasing.

Let us now show the ordering of the boundaries as expressed in
\eqref{eq:ordering}. Suppose
$[0,T)\times\mathbb{R}_{+}\ni(t,x)\in\mathcal{C}_{1}$, so that
$x<x^{*}_{1}(t)$ and $v_{1}(t,x)>(x-K)^{+}$. We then have, using the
ordering of the value functions established in Lemma \ref{lem:cmtdfi},
that $v_{0}(t,x) \geq v_{1}(t,x) > (x-K)^{+}$, so that we also have
$(t,x)\in\mathcal{C}_{0}$ and hence $x<x^{*}_{0}(t)$, which implies
that $x^{*}_{0}(t)\geq x^{*}_{1}(t)$ over $[0,T)$.

Finally, let us show that the exercise boundaries are right-continuous
over $[0,T)$. Fix $i\in\{0,1\}$ and $t\in[0,T)$,
and consider a sequence $(t_{n})_{n\in\mathbb{N}}$ of times converging
from above to $t$, that is, $t_{n}\downarrow t$ as $n\to\infty$. Since
$x^{*}_{i}(\cdot)$ is non-increasing, we know that the right-hand
limit $x^{*}_{i}(t+)$ exists. Now, for each $n\in\mathbb{N}$,
$(t_{n},x^{*}_{i}(t_{n}))\in\mathcal{S}_{i}$, and because the stopping
region $\mathcal{S}_{i}$ is a closed set, we get that
$(t,x^{*}_{i}(t+))\in\mathcal{S}_{i}$. Then, recalling that
$\mathcal{S}_{i}$ has the up-connected representation \eqref{eq:Si},
we see that we have $x^{*}_{i}(t+)\geq x^{*}_{i}(t)$. But we also have
the reverse inequality $x^{*}_{i}(t+)\leq x^{*}_{i}(t)$ from the fact
that $x^{*}_{i}(\cdot)$ is non-increasing, so we obtain
$x^{*}_{i}(t+)=x^{*}_{i}(t)$, showing that $x^{*}_{i}(\cdot)$ is
right-continuous.

\end{proof}

\subsection{Full information free boundary system}
\label{subsec:fifbs}

Let us now proceed to the free boundary characterisation
  of the full information value function. Define differential
operators $\mathcal{L}_{i}$, $i=0,1$, acting on functions
$f\in C^{1,2}([0,T])\times\mathbb R_{+})$, by
\begin{equation*}
\mathcal{L}_{i}f(t,x) := \left(\frac{\partial}{\partial t} +
\mu_{i}x\frac{\partial}{\partial x} +
\frac{1}{2}\sigma^{2}x^{2}\frac{\partial^{2}}{\partial
x^{2}} - r\right)f(t,x),  \quad i=0,1.
\end{equation*}
The free boundary problem for the full information value function then
involves a pair of coupled PDEs as given in Proposition
\ref{prop:fbpfivf} below. The proof illustrates that a classical
approach, akin to the proof of Theorem 2.7.7 of Karatzas and Shreve
\cite{ks98} in the Black-Scholes model, can be extended in our random
drift scenario. This is in marked contrast to the much more involved
proof of the free boundary system satisfied by finite maturity
American put options in regime switching models given by Le and Wang
\cite[Proposition 1]{lw10}. To the best of our knowledge, our result
below constitutes the first time the classical method of proof is
extended to a finite horizon American option model with regime
switching (for example, no such regularity is established in
Buffington and Elliott \cite{be02}).

\begin{proposition}[Free boundary problem: full information] 
\label{prop:fbpfivf}

The full information value function $v(t,x,i)\equiv v_{i}(t,x)$,
$i=0,1$, defined in (\ref{eq:fivf}) is the unique solution in
$[0,T]\times\mathbb R_{+}\times\{0,1\}$ of the free boundary problem
\begin{eqnarray}
\mathcal{L}_{0}v_{0}(t,x) & = &
- \lambda\left(v_{1}(t,x) - v_{0}(t,x)\right), \quad 0\leq x <
x^{*}_{0}(t), \quad t\in[0,T), \label{eq:pde1} \\   
\mathcal{L}_{1}v_{1}(t,x) & = & 0, \quad 0\leq x < x^{*}_{1}(t), \quad
t\in[0,T), \label{eq:pde2} \\
v_{i}(t,x) & = & x-K, \quad x \geq x^{*}_{i}(t), \quad t\in[0,T), \quad
i=0,1, \label{eq:bci} \\ 
v_{i}(T,x) & = & (x-K)^{+}, \quad x\in\mathbb R_{+}, \quad
i=0,1, \label{eq:vibcT} \\
\lim_{x\downarrow 0}v_{i}(t,x) & = & 0, \quad t\in[0,T), \quad
i=0,1. \label{eq:vibc0} 
\end{eqnarray}

\end{proposition}

\begin{proof}

It is clear that $v_{i}(\cdot,\cdot),\,i=0,1$ satisfy the boundary
conditions \eqref{eq:bci}, \eqref{eq:vibcT} and \eqref{eq:vibc0}. It
remains to verify the PDEs \eqref{eq:pde1} and \eqref{eq:pde2}. To
this end, take a pair of points
$(t_{i},x_{i})\in\mathcal{C}_{i},\,i=0,1$ and a pair of rectangles
$\mathcal{R}_{i}:=(t^{\min}_{i},t^{\max}_{i})\times
(x^{\min}_{i},x^{\max}_{i}),\,i=0,1$,
with
$(t_{i},x_{i})\in\mathcal{R}_{i}\subset\mathcal{C}_{i},\,i=0,1$.
Let $\partial\mathcal{R}_{i},\,i=0,1$ denote the boundaries of these
rectangles, and denote by
$\partial_{0}\mathcal{R}_{i}:=
\partial\mathcal{R}_{i}\setminus[\{t^{\min}_{i}\}\times
(x^{\min}_{i},x^{\max}_{i})]$
the so-called parabolic boundaries of these rectangles. With this
set-up, consider the terminal-boundary value problem
\begin{eqnarray}
\mathcal{L}_{0}f_{0} & = & -\lambda(f_{1}-f_{0}),  \quad \mbox{in
$\mathcal{R}_{0}$}; \quad f_{0}=v_{0}, \quad \mbox{on
$\partial_{0}\mathcal{R}_{0}$}, \label{eq:tbvp0} \\
\mathcal{L}_{1}f_{1} & = & 0, \quad \mbox{in 
$\mathcal{R}_{1}$}; \quad f_{1}=v_{1}, \quad \mbox{on
$\partial_{0}\mathcal{R}_{1}$}. \label{eq:tbvp1}
\end{eqnarray}
Classical theory for parabolic PDEs (for example, Friedman
\cite[Chapter 3]{friedman64}) guarantees the existence of a unique
solution to \eqref{eq:tbvp0}--\eqref{eq:tbvp1} with all derivatives
appearing in $\mathcal{L}_{i},\,i=0,1$ being continuous. We wish to
show that $f_{i}$ and $v_{i}$ agree on $\mathcal{R}_{i},\,i=0,1$,
respectively.

With $(t_{i},x_{i})\in\mathcal{R}_{i},\,i=0,1$ given, define stopping
times $\tau_{i},\,i=0,1$ by
\begin{equation*}
\tau_{i} := \inf\{\rho\in[0,t^{\max}_{i}-t):
(t_{i}+\rho,x_{i}H_{\rho}) \in\partial_{0}\mathcal{R}_{i}\} \wedge
(t^{\max}_{i}-t), \quad i=0,1, 
\end{equation*}
and processes $N^{i},\,i=0,1$ by
\begin{equation*}
N^{i}_{\rho} := \e^{-r\rho}f_{i}(t_{i}+\rho,x_{i}H_{\rho}),
\quad 0\leq \rho\leq t^{\max}_{i}-t, \quad i=0,1.  
\end{equation*}
where $H_{\rho}\equiv H_{0,\rho}$ is the accumulation factor in
\eqref{eq:Hst} for the interval $[0,\rho]$. The stopped processes
$(N^{i}_{\rho\wedge\tau_{i}})_{0\leq\rho\leq
  t^{\max}_{i}-t_{i}},\,i=0,1$ are
$(\mathbb{P},\mathbb{F})$-martingales by virtue of the It\^o formula
and the system \eqref{eq:tbvp0}--\eqref{eq:tbvp1} satisfied by
$f_{i},\,i=0,1$, and therefore
\begin{equation}
f_{i}(t_{i},x_{i}) = N^{i}_{t_{i}} =  \mathbb{E}[N^{i}_{\tau_{i}}] =
\mathbb{E}[\e^{-r\tau_{i}}v_{i}(t_{i}+\tau_{i},x_{i}H_{\tau_{i}})],
\quad i=0,1,
\label{eq:fiexrep}
\end{equation}
where we have used the boundary conditions in
\eqref{eq:tbvp0}--\eqref{eq:tbvp1}  to obtain the last equality for
each $i=0,1$. 

But $\mathcal{R}_{i}\subset\mathcal{C}_{i},\,i=0,1$ implies that
$(t_{i}+\tau_{i},x_{i}H_{\tau_{i}})\in\mathcal{C}_{i},\,i=0,1$, which
implies that $\tau_{i},\,i=0,1$ must be less than or equal to the smallest
optimal stopping time for starting state $(t_{i},x_{i}),\,i=0,1$, that
is
\begin{equation*}
\tau_{i} \leq \tau^{*}_{i}(t_{i},x_{i}) := \inf\{\rho\in[0,T-t_{i}):
v_{i}(t_{i}+\rho,x_{i}H_{\rho}) = (x_{i}H_{\rho} - K)^{+}\}\wedge
(T-t_{i}), \quad i=0,1.  
\end{equation*}
Now, the stopped processes
\begin{equation*}
\e^{-r(\rho\wedge\tau^{*}_{i}(t_{i},x_{i}))}v_{i}\left(t_{i} +
(\rho\wedge\tau^{*}_{i}(t_{i},x_{i})),
x_{i}H_{\rho\wedge\tau^{*}_{i}(t_{i},x_{i})}\right), \quad 0\leq
\rho\leq T-t_{i}, \quad i=0,1,
\end{equation*}
are martingales, so this and the optional sampling theorem yield that
\begin{equation}
\mathbb{E}\left[\e^{-r\tau_{i}}v_{i}(t_{i}+\tau_{i},x_{i}H_{\tau_{i}})\right]
= v_{i}(t_{i},x_{i}),
\label{eq:viexrep}
\end{equation}
Then, \eqref{eq:fiexrep} and \eqref{eq:viexrep} show that, for each
$i=0,1$, $f_{i}$ and $v_{i}$ agree on $\mathcal{R}_{i}$ (and hence
also on $\mathcal{C}_{i}$ since
$\mathcal{R}_{i}\subset\mathcal{C}_{i}$ and
$(t_{i},x_{i})\in\mathcal{R}_{i}$ were arbitrary). Thus,
$v_{i},\,i=0,1$ satisfy the PDEs \eqref{eq:pde1} and \eqref{eq:pde2}.

Finally, to show uniqueness, let $g_{i},\,i=0,1$ defined on the
closure of $\mathcal{C}_{i},\,i=0,1$ respectively, be solutions to the
system \eqref{eq:pde1}--\eqref{eq:vibc0}. For starting states
$(0,x_{i},i),\,i=0,1$ such that $x_{i}<x^{*}_{i}(0),\,i=0,1$, define
\begin{equation*}
L^{i}_{t} := \e^{-rt}g_{i}(t,x_{i}H_{t}), \quad t\in[0,T], \quad i=0,1,  
\end{equation*}
as well as the smallest optimal stopping times for
$v_{i}(0,x_{i}),\,i=0,1$, given by
\begin{eqnarray*}
\tau^{*}_{0}(x_{0}) & := & \inf\{t\in[0,T): x_{0}H^{(0)}_{t}\geq
x^{*}_{0}(t)\}\wedge \inf\{t\in[0,T): x_{0}H^{(1)}_{t}\e^{\sigma\eta\theta}\geq
x^{*}_{1}(t)\} \wedge T, \\
\tau^{*}_{1}(x_{1}) & := & \inf\{t\in[0,T): x_{1}H^{(1)}_{t}\geq
x^{*}_{1}(t)\}\wedge T. \nonumber
\end{eqnarray*}
In the first equation above, the early exercise times on the
right-hand side correspond to exercise before the change point (for
$x_{0}H^{(0)}_{t}\geq x^{*}_{0}(t)$) and after the change point (for
$x_{0}H^{(1)}_{t}\e^{\sigma\eta\theta}\geq x^{*}_{1}(t)$), where we
have used the form \eqref{eq:Xacp} of the stock price after the change
point.

The It\^o formula yields that each
$(L^{i}_{t\wedge\tau^{*}_{i}(x_{i})})_{t\in[0,T]}$ is a
martingale. Then, optional sampling along with the fact that
$\tau^{*}_{i}(x_{i}),\,i=0,1$ attain the respective suprema in
\eqref{eq:fivfrep2} starting at time zero, yields that
\begin{eqnarray*}
g_{i}(0,x_{i}) = L^{i}_{0} & = & \mathbb{E}[L^{i}_{\tau^{*}_{i}(x_{i})}] \\
& = & \mathbb{E}\left[
\e^{-r\tau^{*}_{i}(x_{i})}g_{i}(\tau^{*}_{i}(x_{i}),x_{i}H_{\tau^{*}_{i}(x_{i})})\right] \\ 
& = & \mathbb{E}\left[\e^{-r\tau^{*}_{i}(x_{i})}(x_{i}H_{\tau^{*}_{i}(x_{i})} -
K)^{+}\right] \\
& = & v_{i}(0,x_{i}), \quad i=0,1, 
\end{eqnarray*}
so that the solution is unique.

\end{proof}

\subsection{Full information smooth fit condition}
\label{subsec:fisp}

Proposition \ref{prop:fbpfivf} shows that for $i=0,1$, each
$v_{i}(\cdot,\cdot)$ is $C^{1,2}([0,T)\times\mathbb R_{+})$ in the
corresponding continuation region $\mathcal{C}_{i}$. In the stopping
region we know that $v_{i}(t,x)=x-K$, which is also smooth. At issue
then is the smoothness of $v_{i}(\cdot,\cdot)$ across the exercise
boundaries $x^{*}_{i}(\cdot)$. This is settled by the smooth pasting
property in Theorem \ref{thm:spfi} below. This property has been
established for an American put in a model with multiple
regime-switching by Le and Wang \cite[Lemma 8]{lw10}, though the
method of proof is complicated, relying on extending an iterative
procedure first developed by Bayraktar \cite{bayraktar09}, and relies
on the boundedness of the put payoff as well. Our proof is more
direct, exploiting our specific one-switch model, and showing how
classical techniques developed for the Black-Scholes model (see for
example the proof of Lemma 2.7.8 in Karatzas and Shreve \cite{ks98}),
which proceed by analysing properties of the smallest optimal stopping
time from a given starting state, can be extended to the random drift
scenario.

\begin{theorem}[Smooth pasting: full information value function] 
\label{thm:spfi}

The functions $v_{i}(\cdot,\cdot)$, $i=0,1$, satisfy the smooth pasting
property at the optimal exercise thresholds $x^{*}_{i}(\cdot)$:
\begin{equation*}
\frac{\partial v_{i}}{\partial x}(t,x^{*}_{i}(t)) = 1, \quad t\in[0,T),
\quad i=0,1.  
\end{equation*}

\end{theorem}

\begin{proof}

It entails no loss of generality in this proof if we use the
starting time $t=0$, so for simplicity of presentation we do so, and
write $v_{i}(x)\equiv v_{i}(0,x),\,i=0,1,\,x\in\mathbb{R}_{+}$, and
$x^{*}_{i}\equiv x^{*}_{i}(0)$ for brevity.

For $x\in\mathbb{R}_{+}$ and for each $i\in\{0,1\}$, the map
$x\to v_{i}(x)$ is convex and non-decreasing, so we have
$0\leq v^{\prime}_{i}(x)\leq 1$ in the continuation region at time
zero, $\mathcal{C}^{0}_{i}:=\{x\in\mathbb{R}_{+}:x<x^{*}_{i}\}$, and
thus $v^{\prime}_{i}(x^{*}_{i}-)\leq 1$. We also have
$v^{\prime}_{i}(x)=1$ in the corresponding stopping region
$\mathcal{S}^{0}_{i}:=\{x\in\mathbb{R}_{+}:x\geq x^{*}_{i}\}$ and thus
$v^{\prime}_{i}(x^{*}_{i}+)=1$. Hence, the proof will be complete if
we can show that $v^{\prime}_{i}(x^{*}_{i}-)\geq 1$.

First consider the case $i=1$, that is, the stock price evolution
begins in the low-drift regime, so the change point happens at the
initial time. The stock drift is thus equal to $\mu_{1}$ throughout
$[0,T]$ and the relevant value function is $v_{1}(\cdot)$. Denote by
$\tau_{1}(x)$ the smallest optimal stopping time given an initial
stock price $x\in\mathbb{R}_{+}$, given by the first time the stock
breaches the boundary $x^{*}_{1}(\cdot)$:
\begin{equation*}
\tau_{1}(x) :=  \inf\{t\in[0,T):xH^{(1)}_{t}\geq x^{*}_{1}(t)\}\wedge
T,
\end{equation*}
where $H^{(1)}$ is the process in \eqref{eq:Hi} for $s=0$ and $i=1$,
giving the multiplicative random factor by which the stock
price appreciates, so that, given $Y_{0}=1$ and $X_{0}=x$, the stock
price at $t\in[0,T]$ is $X_{t}\equiv X^{x}_{t}$, given by
\begin{equation*}
X_{t} = xH^{(1)}_{t} =
x\exp\left[\left(\mu_{1}-\frac{1}{2}\sigma^{2}\right)t + \sigma
  W_{t}\right], \quad t\in[0,T].  
\end{equation*}
Set $x=x^{*}_{1}\geq K$ (the last inequality due to the fact that
exercise below the strike is never optimal), fixed for the remainder
of the proof for the case $i=1$, and define
\begin{equation*}
\tau_{1}(x-\epsilon) :=  \inf\{t\in[0,T):(x-\epsilon)H^{(1)}_{t}\geq
x^{*}_{1}(t)\}\wedge T,
\end{equation*}
for $\epsilon\geq 0$, so that $\tau_{1}(x)\equiv 0$ and
$\tau_{1}(x-\epsilon)$ is non-decreasing in $\epsilon$. Because
$x^{*}_{1}(\cdot)$ is non-increasing, we have
\begin{equation}
\tau_{1}(x-\epsilon) \leq \inf\{t\in[0,T):(x-\epsilon)H^{(1)}_{t}\geq
x\}\wedge T.
\label{eq:tau1xe}
\end{equation}
The Law of the Iterated Logarithm for the Brownian motion $W$
(Karatzas and Shreve \cite[Theorem 2.9.23]{ks91}) implies that
$\mathbb{P}[\sup_{0\leq t\leq a}H^{(1)}_{t}>1]=1$ for every $a>0$, so
there will exist a sufficiently small $\epsilon>0$ such that
$\sup_{0\leq t\leq a}(x-\epsilon)H^{(1)}_{t}\geq x$ almost surely for
every $a>0$. Thus, the right-hand-side of \eqref{eq:tau1xe} tends to
zero as $\epsilon\downarrow 0$, and therefore
\begin{equation}
\tau_{1}(x-\epsilon) \downarrow 0 \quad \mbox{as} \quad
\epsilon\downarrow 0, \quad \mbox{almost surely}.  
\label{eq:limtau1}
\end{equation}
Using the fact that $\tau_{1}(x-\epsilon)$ will be sub-optimal for the
starting state $(X_{0},Y_{0})=(x,1)$ we have
\begin{eqnarray}
&& v_{1}(x) - v_{1}(x-\epsilon) \label{eq:v1xv1xe} \\
& \geq & \mathbb{E}\left[
\e^{-r\tau_{1}(x-\epsilon)}\left((xH^{(1)}_{\tau_{1}(x-\epsilon)}-K)^{+}
-   ((x-\epsilon)H^{(1)}_{\tau_{1}(x-\epsilon)}-K)^{+}\right)\right]
\nonumber \\
& \geq & \mathbb{E}\left[
\e^{-r\tau_{1}(x-\epsilon)}\left((xH^{(1)}_{\tau_{1}(x-\epsilon)}-K)^{+}
- ((x-\epsilon)H^{(1)}_{\tau_{1}(x-\epsilon)}-K)^{+}\right)
\mathbbm{1}_{\{(x-\epsilon)H^{(1)}_{\tau_{1}(x-\epsilon)}\geq
K\}}\right] \nonumber \\
& = & \epsilon\mathbb{E}\left[
\e^{-r\tau_{1}(x-\epsilon)}H^{(1)}_{\tau_{1}(x-\epsilon)}
\mathbbm{1}_{\{(x-\epsilon)H^{(1)}_{\tau_{1}(x-\epsilon)}\geq
K\}}\right]. \nonumber     
\end{eqnarray}
We now take the limit as $\epsilon\downarrow 0$. Using
\eqref{eq:limtau1} we almost surely have
$\lim_{\epsilon\downarrow 0} H^{(1)}_{\tau_{1}(x-\epsilon)}=1$ and,
since it is never optimal to exercise below the strike,
$\lim_{\epsilon\downarrow 0}
\mathbbm{1}_{\{(x-\epsilon)H^{(1)}_{\tau_{1}(x-\epsilon)}\geq K\}}
=1$. Using these properties, along with the uniform integrability of
$(H^{(1)}_{t})_{t\in[0,T]}$, in \eqref{eq:v1xv1xe}, we compute
\begin{equation*}
v^{\prime}_{1}(x-) = \lim_{\epsilon\downarrow
0}\frac{1}{\epsilon}(v_{1}(x) - v_{1}(x-\epsilon)) \geq 1,
\end{equation*}
which completes the proof in the case $i=1$.

Now consider the case $i=0$, so that the stock begins at time zero in
the high-drift state with drift $\mu_{0}$. The early exercise
scenarios bifurcate into two possibilities, either (i) before the
change point or (ii) at or after the change point. Recall that, given
$Y_{0}=0$ and $X_{0}=x$, the stock price at $t\in[0,\theta]$ (so
up to the change point) is
$X_{t}\equiv X^{x}_{t}$, given by
\begin{equation*}
X_{t} = xH^{(0)}_{t} =
x\exp\left[\left(\mu_{0}-\frac{1}{2}\sigma^{2}\right) + \sigma
  W_{t}\right], \quad 0\leq t\leq\theta,  
\end{equation*}
while at or after the change point the stock price is given by
\begin{equation*}
X_{t} = xH^{(1)}_{t}\exp(\sigma\eta\theta), \quad 0\leq\theta\leq t,
\end{equation*}
and observe that for $t=\theta$ the stock price is
$xH^{(1)}_{\theta}\e^{\sigma\eta\theta} = xH^{(0)}_{\theta}$. The
smallest optimal stopping time starting from $(0,X_{0},Y_{0})=(0,x,0)$
is then $\tau_{0}(x)$, given by
\begin{equation}
\tau_{0}(x) := \inf\{t\in[0,T):xH^{(0)}_{t}\geq x^{*}_{0}(t)\}\wedge
\inf\{t\in[0,T):xH^{(1)}_{t}\e^{\sigma\eta\theta}\geq
x^{*}_{1}(t)\}\wedge T.
\label{eq:tau0x}
\end{equation}
The first time on the right-hand-side of \eqref{eq:tau0x}
corresponds to early exercise before the change point if the stock
breaches $x^{*}_{0}(\cdot)$, while the second time corresponds to
early exercise at or after the change point if the stock breaches
$x^{*}_{1}(\cdot)$. The latter scenario includes the possibility of
early exercise at the change point itself, in which case the stock
price on exercise is
$xH^{(1)}_{\theta}\e^{\sigma\eta\theta} =
xH^{(0)}_{\theta}\in[x^{*}_{1}(\theta),x^{*}_{0}(\theta))$.

As we did for the case $i=1$, set $x=x^{*}_{0}\geq K$, fixed for the
remainder of the proof, and define
\begin{equation*}
\tau_{0}(x-\epsilon) :=  \inf\{t\in[0,T):(x-\epsilon)H^{(0)}_{t}\geq
x^{*}_{0}(t)\} \wedge
\inf\{t\in[0,T):(x-\epsilon)H^{(1)}_{t}\e^{\sigma\eta\theta}\geq
x^{*}_{1}(t)\} \wedge T,
\end{equation*}
for $\epsilon\geq 0$, so that $\tau_{0}(x)\equiv 0$ and
$\tau_{0}(x-\epsilon)$ is non-decreasing in $\epsilon$. Now,
regardless of whether exercise occurs before the change point or not,
because the exercise boundaries are non-increasing and because
$x^{*}_{1}(t)\leq x^{*}_{0}(t)$ for all $t\in[0,T)$, we always have
\begin{equation}
\tau_{0}(x-\epsilon) \leq \inf\{t\in[0,T):(x-\epsilon)H^{(0)}_{t}\geq
x\}\wedge T,
\label{eq:tau0xe}
\end{equation}
which is the analogue of \eqref{eq:tau1xe} for the case $i=0$. With
\eqref{eq:tau0xe} in place, the rest of the proof follows the same
arguments as in the $i=1$ case, so we obtain $v^{\prime}_{0}(x-)\geq
1$, and the proof of smooth fit is complete.

\end{proof}

\subsection{Doob-Meyer decomposition of full information
  Snell envelope}
\label{subsec:dmdfise}

With the free boundary PDE and smooth pasting condition established
for the full information value function, we can now turn to the proof
of Proposition \ref{prop:tvfull}, characterising the continuity over
$[0,T)$ and left limits $x^{*}_{i}(T-).\,i=0,1$ of the exercise
boundaries as we approach maturity. The key to rigorously establishing
this result turns out to be the Doob-Meyer decomposition of the
supermartingale that is the full information Snell envelope, in other
words, the discounted full information ESO value process. This in turn
leads to the decompositions below for the discounted processes
$(\e^{-rt}v_{i}(t,X_{t}))_{t\in[0,T]},\,i=0,1$, where we recall the
representation \eqref{eq:fivp} for the ESO value process $V$ in terms
of the processes $(v_{i}(t,X_{t}))_{t\in[0,T]},\,i=0,1$.

\begin{theorem}[Doob-Meyer decomposition of full information Snell
envelope] 
\label{thm:dmdecomp}
  
The processes $(\e^{-rt}v_{i}(t,X_{t}))_{t\in[0,T]},\,i=0,1$, admit
the decomposition
\begin{equation}
\e^{-rt}v_{i}(t,X_{t}) = v_{i}(0,X_{0}) + M^{i}_{t} - A^{i}_{t}, \quad
t\in[0,T], \quad  i=0,1,
\label{eq:dmdecomp}
\end{equation}
where
\begin{equation*}
M^{i}_{t} :=  \sigma\int_{0}^{t}\e^{-rs}X_{s}\frac{\partial v_{i}}{\partial
  x}(s,X_{s})\ud W_{s}, \quad t\in[0,T], \quad  i=0,1,
\end{equation*}
are $(\mathbb{P},\mathbb{F})$-martingales, and
\begin{equation*}
A^{i}_{t} :=
\int_{0}^{t}\e^{-rs}\left((r-\mu_{i})X_{s}-rK\right)
\mathbbm{1}_{\{X_{s}\geq  x^{*}_{i}(s)\}}\ud s, \quad t\in[0,T], \quad
i=0,1,
\end{equation*}
are non-decreasing finite variation processes.

Consequently, the exercise boundaries $x^{*}_{i}(\cdot),\,i=0,1$
satisfy
\begin{equation}
(r-\mu_{i})x^{*}_{i}(t) - rK \geq 0, \quad \mbox{for
  Lebesgue-almost-every $t\in[0,T)$}, \quad i=0,1,  
\label{eq:xislb}
\end{equation}
and in particular we have the terminal left-limit lower bounds
\begin{equation}
x^{*}_{i}(T-) \geq \left(\frac{r}{r-\mu_{i}}\right)K, \quad i=0,1.  
\label{eq:xilb}
\end{equation}

\end{theorem}

\begin{proof}
  
We have identified the full information discounted ESO value process
$(\e^{-rt}V_{t})_{t\in[0,T]}$ with the Snell envelope of the reward
process, the smallest c\`adl\`ag
$(\mathbb{P},\mathbb{F})$-supermartingale which dominates the reward
process. We recall the representation \eqref{eq:fivp} of the value
process $V$ in terms of the value function processes
$(v_{i}(t,X_{t}))_{t\in[0,T]},\,i=0,1$, and also recall that the
process $Y$ is equal to either $0$ (before the change point) or $1$
(from the change point onwards). The smooth fit condition in Theorem
\ref{thm:spfi}, along with the free boundary PDE system in Proposition
\ref{prop:fbpfivf}, guarantee that the first partial derivatives
$\partial v_{i}(\cdot,\cdot)/\partial x,\,i=0,1$, are continuous, even
across their respective exercise boundaries $x^{*}_{i}(\cdot)$. We
know also from Proposition \ref{prop:fbpfivf} that the second partial
derivatives $\partial^{2}v_{i}(\cdot,\cdot)/\partial x^{2},\,i=0,1$,
are continuous in their respective continuation regions
$\mathcal{C}_{i},\,i=0,1$, and equal to zero in their respective
stopping regions $\mathcal{S}_{i},\,i=0,1$. Though these second
derivatives might not be continuous across their respective exercise
boundaries, we may nevertheless apply the generalised It\^o formula
for convex functions (for instance, Karatzas and Shreve \cite [Theorem
3.7.1]{ks91}) to the (discounted) ESO value process. In differential
form, we have
\begin{eqnarray*}
\ud(\e^{rt}V_{t}) & = & \e^{-rt}\left\{(1-Y_{t})\left(\ud v_{0}(t,X_{t}) -
rv_{0}(t,X_{t})\ud t + \lambda(v_{1}(t,X_{t})-v_{0}(t,X_{t}))\ud
t\right)\right. \\
& + & \left.Y_{t}\left(\ud v_{1}(t,X_{t}) - rv_{1}(t,X_{t})\ud
t\right)\right\},     
\end{eqnarray*}
where, of course, the term involving $\lambda$ is due to the
possibility of the change point occurring in the next instant. Then,
using the generalised It\^o rule on the functions $v_{i}(\cdot,\cdot)$
and integrating over $[0,t]$ for $t\in[0,T]$, we obtain
\begin{eqnarray}
\e^{-rt}V_{t} - V_{0} & = & 
(1-Y_{t})\left(\sigma\int_{0}^{t}\e^{-rs}X_{s}\frac{\partial
v_{0}}{\partial x}(s,X_{s})\ud W_{s}\right. \nonumber \\
& - & \left.
\int_{0}^{t}\e^{-rs}\left((r-\mu_{0})X_{s}-rK\right)\mathbbm{1}_{\{X_{s}\geq
x^{*}_{0}(s)\}}\ud s\right) \nonumber \\
& + & Y_{t}\left(\sigma\int_{0}^{t}\e^{-rs}X_{s}\frac{\partial
v_{1}}{\partial x}(s,X_{s})\ud W_{s}\right. \nonumber \\
& - & \left.
\int_{0}^{t}\e^{-rs}\left((r-\mu_{1})X_{s}-rK\right)\mathbbm{1}_{\{X_{s}\geq
x^{*}_{1}(s)\}}\ud s\right), \quad t\in[0,T].
\label{eq:dmd}      
\end{eqnarray}
In applying the generalised It\^o rule to obtain \eqref{eq:dmd}, we
have used the aforementioned properties of the functions
$v_{i}(\cdot,\cdot),\,i=0,1$ (that is, the PDEs satisfied by these
functions in the respective continuation regions, along with their
analytic forms in the respective stopping regions), with the second
derivative of a convex function considered as a measure (see for
example Karatzas and Shreve \cite[equation (3.6.47)]{ks91}).

Now, in \eqref{eq:dmd}, the stochastic integral terms are
$(\mathbb{P},\mathbb{F})$-martingales, since the discount factor and
partial derivative terms are bounded and the stock price process is
square-integrable: $\mathbb{E}[X^{2}_{t}]<\infty$ for any
$t\in[0,T]$. Then, recalling once again the representation
\eqref{eq:fivp} for the value process $V$, we have that in both
\eqref{eq:fivp} and \eqref{eq:dmd} above, one either has $Y_{t}=0$ or
$Y_{t}=1$ on a mutually exclusive basis, so only one of the
martingales in \eqref{eq:dmd} contributes at any particular time. The
same also applies to the finite variation terms on the right-hand-side
of \eqref{eq:dmd}, which is thus the (unique) Doob-Meyer decomposition
of the supermartingale $(\e^{-rt}V_{t})_{t\in[0,T]}$ into a martingale
minus a non-decreasing process. This establishes the decompositions in
\eqref{eq:dmdecomp}, and also the non-decreasing property of the
finite variation processes in \eqref{eq:dmd}, and thus in
\eqref{eq:dmdecomp}. Since $\mathbb{P}[X_{t}\geq x^{*}_{i}(t)]>0$ for
$i=0,1$ and for Lebesgue-almost every $t\in[0,T)$, the non-decreasing
property implies that the exercise boundaries must satisfy
\eqref{eq:xislb} and, in particular, \eqref{eq:xilb} must hold.

\end{proof}

We can now establish Proposition \ref{prop:tvfull}.

\begin{proof}[Proof of Proposition \ref{prop:tvfull}]

It is clear that at maturity itself, exercise will not occur below the
strike, so we must have $x^{*}_{i}(T)=K,\,i=0,1$.

We have established in Corollary \ref{corr:fiexbds} that the
exercise thresholds $x^{*}_{i}(\cdot),\,i=0,1$ are non-increasing
and right-continuous over $[0,T)$, with lower bounds
$x^{*}_{i}(T-),\,i=0,1$ given in \eqref{eq:xilb}. With
$\mu_{i}<r,\,i=0,1$, we first refine this lower bound to be the
right-hand-side of \eqref{eq:xistm}, then we show that in fact we
have the equality \eqref{eq:xistm}. For $\mu_{i}<r,\,i=0,1$, we can
distinguish two cases:

\begin{itemize}
    
\item for $0\leq \mu_{i}<r$, we have $x^{*}_{i}(T-)\geq
(r/(r-\mu_{i}))K\geq K$;

\item for $\mu_{i}<0\leq r$, because it is never optimal to exercise
below the strike, we have $x^{*}_{i}(T-)\geq K >(r/(r-\mu_{i}))K$.

\end{itemize}

We thus have, in all cases, the refined lower bound
\begin{equation*}
x^{*}_{i}(T-) \geq
\max\left(K,\left(\frac{r}{r-\mu_{i}}\right)K\right), \quad i=0,1.
\end{equation*}
We now show that in fact we have the equality
\eqref{eq:xistm}. Suppose, to the contrary, that we have
$x^{*}_{i}(T-)>\max\left(K,\left(r/(r-\mu_{i})\right)K\right),\,i=0,1$. For
each $i=0,1$, consider a value
$x_{i}\in\left(
\max\left(K,\left(r/(r-\mu_{i})\right)K\right),x^{*}_{i}(T-)\right)$. Then,
for $0\leq t<T$, we have $(t,x_{i})\in\mathcal{C}_{i},\,i=0,1$, so
that $v_{i}(t,x)>(x_{i}-K)^{+}=x_{i}-K$. Using temporal continuity of
$v_{i}(\cdot,\cdot)$, we thus obtain
$v_{i}(T,x)=\lim_{t\uparrow T}v_{i}(t,x)> x_{i}-K$. But, on the other
hand, we know that at maturity we have
$v_{i}(T,x)=(x_{i}-K)^{+}=x_{i}-K$, so we have a contradiction. Thus,
\eqref{eq:xistm} holds.

Finally, let us show that the exercise thresholds
$x^{*}_{i}(\cdot),\,i=0,1$ are left-continuous over $[0,T)$, thus
establishing the claimed continuity. To prove left-continuity we shall
suppose $x^{*}_{i}(t_{i}-)>x^{*}_{i}(t_{i}),\i=0,1$ for some
$t_{i}\in(0,T)$ and obtain a contradiction. Under this assumption,
take
$x_{i}:=\frac{1}{2}(x^{*}_{i}(t_{i}-)+x^{*}_{i}(t_{i}))>x^{*}_{i}(t_{i})\geq
K,\,i=0,1$ (of course, not the same $x_{i}$ as in the previous
paragraph). Observe that $(t_{i},x_{i})\in\mathcal{S}_{i},\,i=0,1$ but
that $(t,x_{i})\in\mathcal{C}_{i},\,i=0,1$ for $t\in(0,t_{i})$. For
each $i=0,1$, let $t\in(0,t_{i})$ and $x\in(x_{i},x^{*}_{i}(t))$ be
given, so that (as for $x_{i}$) we have
$(t_{i},x)\in\mathcal{S}_{i},\,i=0,1$ but
$(t,x)\in\mathcal{C}_{i},\,i=0,1$ for $t\in(0,t_{i})$.

Now use the fact that $v_{i}(\cdot,\cdot)$ solves a given PDE in
$\mathcal{C}_{i}$, as follows: for $v_{0}(\cdot,\cdot)$, use
\eqref{eq:pde1} along with the ordering of the value functions and
time decay (properties (2) and (3) in Lemma \ref{lem:cmtdfi}), while
for $v_{1}(\cdot,\cdot)$, use \eqref{eq:pde2} and time decay, to
conclude that
\begin{equation*}
\frac{1}{2}\sigma^{2}x^{2}\frac{\partial^{2}v_{i}}{\partial
x^{2}}(t.x) \geq rv_{i}(t,x) - \mu_{i}x\frac{\partial
v_{i}}{\partial x}(t,x), \quad (t,x)\in\mathcal{C}_{i}, \quad i=0,1.
\end{equation*}

Now consider separately the cases (i) $\mu_{i}<0\leq r$ and (ii)
$0\leq \mu_{i}< r$. In case (i) we have $- \mu_{i}x\frac{\partial
  v_{i}}{\partial x}(t,x)>0$; using this and $v_{i}(t,x)>x-K$ in
$\mathcal{C}_{i}$, we conclude that $\frac{\partial^{2}v_{i}}{\partial
  x^{2}}(t.x) \geq \epsilon>0$, for some $\epsilon>0$. In case (ii),
usiing that $x\to v_{i}(\cdot,x)$ is non-decreasing and convex, so that
$0\leq \frac{\partial v_{i}}{\partial x}(t,x)\leq 1$, and once again
using $v_{i}(t,x)>x-K$, we get  
\begin{equation*}
\frac{1}{2}\sigma^{2}x^{2}\frac{\partial^{2}v_{i}}{\partial
  x^{2}}(t.x) \geq rv_{i}(t,x) - \mu_{i}x > r(x-K) - \mu_{i}x =
(r-\mu_{i})x - rK, \quad (t,x)\in\mathcal{C}_{i}, \quad i=0,1.
\end{equation*}
But $x>x^{*}_{i}(t_{i})$ implies that (with $0\leq \mu_{i}< r$),
$(r-\mu_{i})x - rK>(r-\mu_{i})x^{*}_{i}(t_{i}) - rK\geq 0$, on using
\eqref{eq:xislb}, and so once again we conclude that
$\frac{\partial^{2}v_{i}}{\partial x^{2}}(t.x) \geq \epsilon>0$, for
some $\epsilon>0$.

Thus, in either case we have
\begin{equation*}
\frac{\partial^{2}v_{i}}{\partial
  x^{2}}(t.x) \geq \epsilon>0, \quad \forall \, t\in (0,t_{i}), \quad
x\in (x_{i},x^{*}_{i}(t)), \quad i=0,1.  
\end{equation*}
Then, with $\varphi(\xi):=(\xi-K)^{+}=\xi-K$ (in the region of
interest) and $x\in (x_{i},x^{*}_{i}(t_{i}-))$ (so that
$(t,x)\in\mathcal{C}_{i}$ for $t\in(0,t_{i})$ but
$(t_{i},x)\in\mathcal{S}_{i}$), we compute
\begin{equation*}
v_{i}(t,x) - \varphi(x) =
\int_{x^{*}_{i}(t)}^{x}\int_{x^{*}_{i}(t)}^{u}\left(\frac{\partial^{2}v_{i}}{\partial
x^{2}}(t,\xi) - \varphi^{\prime\prime}(\xi)\right)\ud\xi\ud u \geq
\frac{1}{2}\epsilon(x-x^{*}_{i}(t))^{2}, \quad i=0,1,
\end{equation*}
where we have used the value-matching and smooth pasting relations
$v_{i}(t,x^{*}_{i}(t))=\varphi(x^{*}_{i}(t))$ and
$\frac{\partial v_{i}}{\partial
  x}(t,x^{*}_{i}(t))=\varphi^{\prime}(x^{*}_{i}(t))$. Finally, letting
$t\uparrow t_{i}$ and using the continuity of $v_{i}(\cdot,\cdot)$, we
get $v_{i}(t_{i},x) \geq x_{i}-K +
\frac{1}{2}\epsilon(x-x^{*}_{i}(t_{i}-))^{2} > x_{i}-K$, which implies
that $(t_{i},x)\in\mathcal{C}_{i},\,i=0,1$. But this contradicts our
earlier assertion that $(t_{i},x)\in\mathcal{S}_{i},\,i=0,1$, and the
proof is complete.

\end{proof}

\section{The partial information ESO problem}
\label{sec:piesop}

We now turn to the partial information problem (\ref{eq:piproblem}),
over $\widehat{\mathbb F}$-stopping times, with model dynamics given
by Lemma \ref{lem:ofd}. In particular, the stock price drift is
$\mu(\widehat{Y})$, defined by
\begin{equation*}
\mu(\widehat{Y}_{t}) := \mu_{0} - \sigma\eta\widehat{Y}_{t}, \quad t\in[0,T],  
\end{equation*}
which we see is the partial information analogue of the full
information drift in \eqref{eq:mufi}.

The partial information value function
$u:[0,T]\times\mathbb R_{+}\times[0,1]\to\mathbb R_{+}$ is defined by
\begin{equation}
u(t,x,y) := \sup_{\tau\in\widehat{\mathcal{T}}_{t,T}}\mathbb
E\left[\left.{\rm e}^{-r(\tau-t)}(X_{\tau}-K)^{+}\right\vert
  X_{t}=x,\widehat{Y}_{t}=y\right], \quad t\in[0,T],
\label{eq:pivf}
\end{equation}
subject to the $(\mathbb{P},\widehat{\mathbb{F}})$-dynamics of the
two-dimensional diffusion $(X,\widehat{Y})$ as given in
\eqref{eq:dSobs} and \eqref{eq:whY}, and the ESO value process $U$ in
\eqref{eq:piproblem} is given as
\begin{equation*}
U_{t} = u(t,X_{t},\widehat{Y}_{t}), \quad t\in[0,T].  
\end{equation*}
For $0\leq s\leq t\leq T$, write
$(X_{t},\widehat{Y}_{t})\equiv(X^{s,x,y}_{t},\widehat{Y}^{s,y}_{t})$
for the value of this diffusion given $(X_{s},\widehat{Y}_{s})=(x,y)$.
Define
\begin{equation*}
G^{s,y}_{t} :=  \exp\left\{\left(\mu_{0} - \frac{1}{2}\sigma^{2}\right)(t
- s) - \sigma\eta\int_{s}^{t}\widehat{Y}^{s,y}_{u}\ud u +
\sigma(\widehat{W}_{t} - \widehat{W}_{s})\right\}, \quad 0\leq s\leq
t\leq T, 
\end{equation*}
so we have
\begin{equation}
X^{s,x,y}_{t} = xG^{s,y}_{t}, \quad 0\leq s\leq t\leq T.
\label{eq:Xsolnpi} 
\end{equation}
When $s=0$, write
$(X^{x,y}_{t},\widehat{Y}^{y}_{t})\equiv(X^{0,x,y}_{t},\widehat{Y}^{0,y}_{t})$
and $G^{y}_{t}\equiv G^{0,y}_{t}$ for $t\in[0,T]$, so that
\begin{equation*}
X^{x,y}_{t} = xG^{y}_{t}, \quad t\in[0,T].  
\end{equation*}
The partial information value function in \eqref{eq:pivf} is thus
\begin{equation*}
u(t,x,y) = \sup_{\tau\in\widehat{\mathcal{T}}_{t,T}}\mathbb
E\left[\e^{-r(\tau-t)}(xG^{t,y}_{\tau} - K)^{+}\right], \quad
(t,x,y)\in[0,T]\times\mathbb R_{+}\times[0,1]. 
\end{equation*}
Using the time-homogeneity of the diffusion $(X,\widehat{Y})$,
optimising over $\widehat{\mathcal{T}}_{t,T}$ is equivalent to
optimising over $\widehat{\mathcal{T}}_{0,T-t}$, so the value function
can be re-cast into the form
\begin{equation}
u(t,x,y) = \sup_{\tau\in\widehat{\mathcal{T}}_{0,T-t}}\mathbb
E\left[\e^{-r\tau}(xG^{y}_{\tau} - K)^{+}\right].
\label{eq:pivfrep} 
\end{equation}
From this representation, elementary properties of the ESO partial
information value function can be derived, largely in a similar manner
to the proof of Lemma \ref{lem:cmtdfi} in the full information case
(but proving monotonicity in $y$ is more involved, as we shall
see).

\begin{remark}[Minimal conditions for early exercise: partial information]
\label{rem:murpi}

Similarly to the full information case, if the drift process
$\mu(\widehat{Y})$ of the stock satisfies $\mu(\widehat{Y})\geq r$
almost surely, then the reward process is a
$(\mathbb{P},\widehat{\mathbb{F}})$-sub-martingale, so no early
exercise is optimal, and the American ESO value coincides with that of
its European counterpart.
 
\end{remark}

\begin{lemma}[Convexity, monotonicity, time decay: partial information]
\label{lem:cmtdpi}

The function $u:[0,T]\times\mathbb R_{+}\times[0,1]$ in
(\ref{eq:pivf}) characterising the partial information ESO value
function has the following properties: 

\begin{enumerate}

\item For $(t,y)\in[0,T]\times[0,1]$, the map $x\to u(t,x,y)$ is
  convex and non-decreasing.

\item For $(t,x)\in[0,T]\times\mathbb R_{+}$, the map $y\to u(t,x,y)$
  is non-increasing.

\item For $(x,y)\in\mathbb R_{+}\times[0,1]$, the map $t\to u(t,x,y)$
  is non-increasing.

\end{enumerate}  

\end{lemma}

\begin{proof}

The proofs of the first and third properties are virtually identical
to the proofs of the corresponding properties for the full information
case in Lemma \ref{lem:cmtdfi}: that is, convexity and monotonicity of
$x\to u(t,x,y)$ follow directly from the corresponding properties of
the payoff map $x\to(x-K)^{+}$, while the time decay property that
$t\to u(t,x,y)$ is non-increasing follows directly from the fact that
the exercise opportunities at an earlier time contain all the exercise
opportunities available at a later time, given the time-homogeneity of
the diffusion $(X,\widehat{Y})$. That is, in \eqref{eq:pivf} we have
$\widehat{\mathcal{T}}_{t,T}\supseteq\widehat{\mathcal{T}}_{t^{\prime},T}$
for $t^{\prime}\geq t$ (equivalently, in \eqref{eq:pivfrep}, we have
$\widehat{\mathcal{T}}_{0,T-t}\supseteq\widehat{\mathcal{T}}_{0,T-t^{\prime}}$).

Let us focus therefore on the second claim, that the map
$y\to u(t,x,y)$ is non-increasing. In \eqref{eq:pivfrep}, the
quantity $G^{y}_{\tau}$ is the value at
$\tau\in\widehat{\mathcal{T}}_{0,T-t}$ of the process $G^{y}$ given by
\begin{equation} 
G^{y}_{t} := \exp\left(\left(\mu_{0} - \frac{1}{2}\sigma^{2}\right)t
+ \sigma \widehat{W}_{t} - \sigma\eta\int_{0}^{t}\widehat{Y}^{y}_{s}\ud
s\right), \quad t\in[0,T].
\label{eq:Gy}  
\end{equation}
From \eqref{eq:Gy} and \eqref{eq:pivfrep}, the desired monotonicity of
the map $y\to u(t,x,y)$ will follow if we can show that the process
$\widehat{Y}^{y}\equiv\widehat{Y}(y)$, seen as a function of the
initial value $y$, that is, as a \emph{stochastic flow}, is
non-decreasing with respect to $y$:
\begin{equation}
\frac{\partial\widehat{Y}_{t}}{\partial y}(y) \geq 0, \quad
\mbox{almost surely}, \quad t\in[0,T].  
\label{eq:yflowi}
\end{equation}
The meaning of \eqref{eq:yflowi} is that for almost all
$\omega\in\Omega$, we consider the process $\widehat{Y}$ with initial
value $y\in[0,1)$ as a function of $y$, so we have
$Y_{t}(y)\equiv Y_{t}(y,\omega)$, and the theory of stochastic flows
(for example Kunita \cite[Chapter 4]{k97}) guarantees that we may
choose versions of $\widehat{Y}(y)$ which, for each $t\in[0,T]$ and
almost all $\omega\in\Omega$, are diffeomorphisms in $y$ from
$[0,1)\to[0,1]$. In other words, the map $y\to\widehat{Y}(\omega,y)$
is smooth, and one can compute the derivative of
$\widehat{Y}(\omega,y)$ with respect to $y$ for almost all
$\omega\in\Omega$. We do this Proposition \ref{prop:flow} below, to
give \eqref{eq:yflowi}, and this completes the proof.

\end{proof}

\subsection{The filtered change point stochastic flow}
\label{subsec:flow}

Consider the solution to the SDE \eqref{eq:whY} for $\widehat{Y}$ for
some initial condition $\widehat{Y}_{0}=y \in[0,1)$. Write
$\widehat{Y}(y)=(\widehat{Y}_{t}(y))_{t\in[0,T]}$ for this
process. Using the theory of stochastic flows (see for instance Kunita
\cite{k97}, Chapter 4), we may choose versions of $\widehat{Y}(y)$
which, for each $t\in[0,T]$ and almost all $\omega\in\Omega$, are
diffeomorphisms in $y$ from $[0,1)\to[0,1]$. In other words, the map
$y\to\widehat{Y}(y)$ is smooth.  (See El Karoui et al.~\cite{elkjbsh98}
and Monoyios and Ng \cite{mman11} for other applications of these
ideas to American claims and ESOs, respectively.) 

We wish to show the property (\ref{eq:yflowi}). To achieve this, we
shall look at the flow of the so-called likelihood ratio $\Phi$,
defined for $\hat{Y} \in [0, 1)$ by
\begin{equation}
\Phi_{t} := \frac{\widehat{Y}_{t}}{1-\widehat{Y}_{t}}, \quad
t\in[0,T].
\label{eq:Phi}
\end{equation}
To examine the flow of $\Phi$, it turns out to be helpful to define
the measure $\mathbb P^{*}\sim\mathbb P$ on
$\widehat{\mathcal{F}}_{T}$ by
\begin{equation}
\Gamma_{t} := \left.\frac{\ud\mathbb P^{*}}{\ud\mathbb
P}\right\vert_{\widehat{\mathcal{F}}_{t}} 
= \mathcal{E}(\eta\widehat{Y}\cdot\widehat{W})_{t}, \quad t\in[0,T],  
\label{eq:Pstar}
\end{equation}
where $\mathcal{E}(\cdot)$ denotes the stochastic exponential, and
$(\widehat{Y}\cdot\widehat{W})\equiv\int_{0}^{\cdot}\widehat{Y}_{s}\ud
\widehat{W}_{s}$ denotes the stochastic integral. Since $\widehat{Y}$
is bounded, the Novikov condition is satisfied and $\mathbb P^{*}$ is
indeed a probability measure equivalent to $\mathbb P$.

By Girsanov's Theorem the process
\begin{equation*}
W^{*}_{t} := \widehat{W}_{t} - \eta\int_{0}^{t}\widehat{Y}_{s}\ud s,
\quad t\in[0,T],  
\end{equation*}
is a $(\mathbb{P}^{*},\widehat{\mathbb{F}})$ Brownian motion. Using
this along with the It\^o formula, the dynamics of $(X,\Phi)$ with
respect to $(\mathbb P^{*},\widehat{\mathbb F})$ are given by
\begin{eqnarray}
\ud X_{t} & = & \mu_{0}X_{t}\ud t + \sigma X_{t}\ud W^{*}_{t},
\label{eq:sdeXstar} \\
\ud\Phi_{t} & = & \lambda(1+\Phi_{t})\ud t -\eta\Phi_{t}\ud W^{*}_{t}.  
\label{eq:sdephi}
\end{eqnarray}
Equations (\ref{eq:sdeXstar}) and (\ref{eq:sdephi}) exhibit an
interesting feature in that $X$ and $\Phi$ become decoupled under
$\mathbb{P}^{*}$. Similar measure changes have been employed by
D{\'e}camps et al.~\cite{dmv05,dmv09}, Klein \cite{klein09} and
Ekstr{\"o}m and Lu \cite{el11} for related optimal stopping problems
involving an investment timing decision or an optimal liquidation
decision when a drift parameter is assumed to take on one of two
values, but the agent is unsure which value pertains in reality. This
corresponds to $\lambda\downarrow 0$ in our set-up, and both $X$ and
$\Phi$ become geometric Brownian motions with respect to
$(\mathbb{P}^{*},\widehat{\mathbb{F}})$, yielding an easier problem,
in that $\Phi$ becomes a deterministic function of $X$. This property,
when combined with the linear payoff function in these papers, allows
for a reduction in dimension under some circumstances in those
works. In our problem, $\Phi$ depends on the entire history of the
Brownian paths, as exhibited in equation \eqref{eq:Phirep} below, and
hence on the history of the stock price, given that we are in the
observation filtration with driving Brownian motion
$\widehat{W}$. This, combined with the non-linear call payoff makes
the aforementioned dimension reduction impossible, and the numerical
solution of the partial information ESO problem is made more complex.

With $\Phi_{0}=\phi$, here is the result which quantifies the
derivative of $\Phi(\phi)$ and hence of $\widehat{Y}(y)$ with respect
to their respective initial conditions, a property which was used in
the proof of Lemma \ref{lem:cmtdpi}.

\begin{proposition}
\label{prop:flow}

Define $\Phi$ by (\ref{eq:Phi}), and define the exponential
$(\mathbb P^{*},\widehat{\mathbb F})$-martingale $\Lambda$ by
\begin{equation*}
\Lambda_{t} := \mathcal{E}(-\eta W^{*})_{t}, \quad t\in [0,T].  
\end{equation*}
Let $\Phi(\phi)$ denote the solution of the SDE (\ref{eq:sdephi}) with
initial condition $\Phi_{0}=\phi\in\mathbb R_{+}$. Then $\Phi(\phi)$ has
the representation
\begin{equation}
\Phi_{t}(\phi) = \e^{\lambda t}\Lambda_{t}\left(\phi +
\lambda\int_{0}^{t}\frac{\e^{-\lambda s}}{\Lambda_{s}}\ud s\right), \quad
t\in[0,T],  
\label{eq:Phirep}
\end{equation}
so that
\begin{equation}
\frac{\partial \Phi_{t}}{\partial \phi}(\phi) = \e^{\lambda
t}\Lambda_{t}, \quad t\in[0,T]. 
\label{eq:Phideriv}
\end{equation}
Consequently, if $\widehat{Y}(y)$ denotes the solution to (\ref{eq:whY})
with initial condition $\widehat{Y}_{0}=y\neq 1$, then
\begin{equation}
\frac{\partial \widehat{Y}_{t}}{\partial y}(y) =
\e^{\lambda t}\Lambda_{t}\left(\frac{1 - \widehat{Y}_{t}(y)}{1 -
y}\right)^{2} \geq 0, \quad t\in[0,T].  
\label{eq:Yderiv}
\end{equation}
  
\end{proposition}

\begin{proof}

It is straightforward to show that $\Phi(\phi)$ as given in
\eqref{eq:Phirep} solves the SDE \eqref{eq:sdephi} with initial
condition $\Phi_{0}=\phi$, and the formula (\ref{eq:Phideriv}) follows
immediately. Then, using
\begin{equation*}
\widehat{Y}_{t}(y) = \frac{\Phi_{t}(\phi)}{1 + \Phi_{t}(\phi)}, \quad
y = \frac{\phi}{1 + \phi}, \quad t\in[0,T],   
\end{equation*}
an exercise in differentiation yields (\ref{eq:Yderiv}).
  
\end{proof}

Observe the second term on the right-hand-side of \eqref{eq:Phirep}
depends on the whole history of $(\Lambda_{s})_{s\in[0,t]}$ over the
time interval $[0,t]$, so that $\Phi$ (and hence $\widehat{Y}$) are
path-dependent. As we are working in the observation filtration, these
processes depend on the history of the stock price itself. This can be
made explicit in some circumstances, as we show for $\Phi$ in equation
\eqref{eq:PhirepX} of Section \ref{subsec:chv}, where the integral
term is written in terms of the stock price path. This path-dependence
is a consequence of the filtering algorithm, and in particular that we
are continuously computing an updated version at each time of the
conditional expectation of a process given observations of the stock
up to that time. It is not uncommon for this updating to generate
path-dependence. This is the ``learning'' aspect of the filtering
algorithm. For some special parameter values, the path-dependence can
sometimes disappear. In this example, for $\lambda=0$ we lose the
history-dependent term in \eqref{eq:Phirep}, reducing to the uncertain
two-value drift model alluded to after \eqref{eq:sdephi}.

\begin{remark}[Completing the proof of Lemma \ref{lem:cmtdpi}]
\label{rem:flow}

Equation \eqref{eq:Yderiv} as derived in the above proof is a
$\mathbb{P}^{*}$-almost sure relation, and so also holds under
$\mathbb{P}$ since these measures are equivalent. This is enough to
complete the proof of Lemma \ref{lem:cmtdpi} as claimed earlier.

\end{remark}

\subsection{Partial information free boundary problem}
\label{eq:fbppieso}

The properties in Lemma \ref{lem:cmtdpi} imply that there exists a
function $x^{*}:[0,T]\times[0,1]\to[K,\infty)$, the optimal exercise
boundary, which is decreasing in time and also in $y$, such that it is
optimal to exercise the ESO as soon as the stock price exceeds the
threshold $x^{*}(t,y)$. Thus, the optimal exercise boundary in the
finite horizon ESO problem under partial information is a surface, and
the continuation and stopping regions
$\widehat{\mathcal{C}},\widehat{\mathcal{S}}$ for the partial
information problem are given by
\begin{eqnarray*}
\widehat{\mathcal C} & := & \{(t,x,y)\in[0,T]\times\mathbb
R_{+}\times[0,1]:u(t,x,y) > (x-K)^{+}\} \\
& = & \{(t,x,y)\in[0,T]\times\mathbb R_{+}\times[0,1]: x<
x^{*}(t,y)\}, \\
\widehat{\mathcal S} & := & \{(t,x,y)\in[0,T]\times\mathbb
R_{+}\times[0,1]:u(t,x,y) = (x-K)^{+}\} \\
& = & \{(t,x,y)\in[0,T]\times\mathbb R_{+}\times[0,1]: x\geq
x^{*}(t,y)\}.   
\end{eqnarray*}

The following lemma gives the left-limiting terminal value
$x^{*}(T-,y)$ of the exercise surface. As in the full information
case, this requires for its proof the free boundary characterisation
of the value function along with a smooth pasting property and also
the Doob-Meyer decomposition of the (partial information) Snell
envelope, so the proof of the lemma will be given in Section
\ref{subsec:dmdpise}, once the required preparation is in place.

\begin{lemma}
\label{lem:tvpartial}

The partial information exercise surface $x^{*}(\cdot,\cdot)$ has
left-limiting value as we approach maturity, given by
\begin{equation}
x^{*}(T-,y) = \max\left(K,\left(\frac{r}{r-(\mu_{0}-\sigma\eta
y)}\right)K\right), \quad y\in[0,1], \quad \mu_{0}-\sigma\eta y <r. 
\label{eq:xstm}
\end{equation}

\end{lemma}

Observe that, since the drift of the stock under the observation filtration
is $\mu(\widehat{Y}):=\mu_{0}-\sigma\eta\widehat{Y}$, the limiting
value in \eqref{eq:xstm} is
\begin{equation*}
x^{*}(T-,y) = \max\left(K,\left(\frac{r}{r-\mu(y)}\right)K\right),  
\end{equation*}
where $\mu(y)$ is the $\widehat{\mathbb{F}}$-drift of the stock when
the filtered change point is equal to $y\in[0,1]$. The last condition
in \eqref{eq:xstm} therefore corresponds to the region of the state
space where the filtered stock drift is less than the interest rate,
and Lemma \ref{lem:tvpartial} is in a similar spirit to the full
information result in Proposition \ref{prop:tvfull}, where we replace
the distinct values $i=0,1$ of the change point process by the
continuum of values in $[0,1]$ for filtered change point process.

Also, by Remark \ref{rem:maptoclassical}, if we invoke a fictitious
``dividend yield'' $\delta(\cdot):=r-\mu(\cdot)$, then we have
$x^{*}(T-y)=\max(K,(r/\delta(y))K)$, so the classical result for the
exercise boundary value at $(T-)$ for no-arbitrage call valuation
extends to the scenario a with a random dividend yield
$\delta(\widehat{Y})$, the same pattern we saw in the full information
problem with random drift $\mu(Y)$.  

We now turn to the free boundary characterisation of the partial
information value function. Let $\mathcal{L}_{X,\widehat{Y}}$ denote
the generator under $\mathbb{P}$ of the two-dimensional process
$(X,\widehat{Y})$ with respect to the observation filtration
$\widehat{\mathbb{F}}$, with dynamics given by (\ref{eq:dSobs}) and
(\ref{eq:whY}). Thus, $\mathcal{L}_{X,\widehat{Y}}$ is defined by
\begin{equation*}
\mathcal{L}_{X,\widehat{Y}}f(t,x,y) := (\mu_{0}-\sigma\eta y)xf_{x}
+ \frac{1}{2}\sigma^{2}x^{2}f_{xx} + \lambda(1-y)f_{y} +
\frac{1}{2}\eta^{2}y^{2}(1-y)^{2}f_{yy} - \sigma\eta xy(1-y)f_{xy},
\end{equation*}
acting on any sufficiently smooth function $f:[0,T]\times\mathbb
R_{+}\times[0,1]$. Define the operator $\mathcal{L}$ by
\begin{equation*}
\mathcal{L} := \frac{\partial}{\partial t} +
\mathcal{L}_{X,\widehat{Y}} - r.
\end{equation*}
The partial information free boundary problem for the ESO is then as
follows.

\begin{proposition}[Free boundary problem: partial information] 
\label{prop:fbppivf}

The partial information ESO value function $u(\cdot,\cdot,\cdot)$
defined in (\ref{eq:pivf}) is the unique solution in
$[0,T]\times\mathbb R_{+}\times[0,1]$ of the free boundary problem
\begin{eqnarray}
\mathcal{L}u(t,x,y) & = & 0, \quad 0\leq x < x^{*}(t,y), \quad
t\in[0,T), \quad y\in[0,1], \label{eq:pdepic}\\  
u(t,x,y) & = & x - K, \quad x \geq x^{*}(t,y), \quad
t\in[0,T), \quad y\in[0,1], \label{eq:bcs}\\
u(T,x,y) & = & (x-K)^{+}, \quad x\in\mathbb R_{+}, \quad
y\in[0,1], \label{eq:bcT} \\
\lim_{x\downarrow 0}u(t,x,y) & = & 0, \quad t\in[0,T), \quad
y\in[0,1]. \label{eq:bcpi0}
\end{eqnarray}

\end{proposition}

\begin{proof}

It is clear that $u$ satisfies the boundary conditions
(\ref{eq:bcs}), (\ref{eq:bcT}) and (\ref{eq:bcpi0}). To verify
(\ref{eq:pdepic}), take a point $(t,x,y)\in\widehat{\mathcal{C}}$ (so
that $x<x^{*}(t,y)$) and a rectangular cuboid
$\mathcal{R}=(t_{\min},t_{\max}) \times (x_{\min},x_{\max}) \times
(y_{\min},y_{\max})$, with
$(t,x,y)\in\mathcal{R}\subset\widehat{\mathcal{C}}$. Let
$\partial\mathcal{R}$ denote the boundary of this region, and let
$\partial_{0}\mathcal{R}:=\partial\mathcal{R}\setminus\left(\{t_{\min}\}
\times (x_{\min},x_{\max}) \times (y_{\min},y_{\max})\right)$ denote
the so-called parabolic boundary of $\mathcal{R}$. Consider the
terminal-boundary value problem
\begin{equation}
\mathcal{L}f=0, \quad \mbox{in} \quad \mathcal{R}, \quad f=u, \quad
\mbox{on} \quad \partial_{0}\mathcal{R}.  
\label{eq:bvp1}
\end{equation}
Classical theory for parabolic PDEs (for instance, Friedman
\cite[Chapter 3]{friedman64}) guarantees the existence of a unique
solution to \eqref{eq:bvp1} with all derivatives appearing in
$\mathcal{L}$ being continuous. We wish to show that $f$ and $u$ agree
on $\mathcal{R}$.

With $(t,x,y)\in\mathcal{R}$ given, define the stopping time
$\tau\in\widehat{\mathcal{T}}_{0,t_{\max}-t}$ by
\begin{equation*}
\tau := \inf\{\rho\in[0,t_{\max}-t):
(t+\rho,xG^{y}_{\rho},\widehat{Y}^{y}_{\rho})
\in\partial_{0}\mathcal{R}\} \wedge (t_{\max}-t),  
\end{equation*}
where the process $G^{y}$ is defined in (\ref{eq:Gy}), 
and define the process $N$ by
\begin{equation*}
N_{\rho} := \e^{-r\rho}f(t+\rho,xG^{y}_{\rho},\widehat{Y}^{y}_{\rho}),
\quad 0\leq \rho\leq t_{\max}-t.  
\end{equation*}
The stopped process $(N_{\rho\wedge\tau})_{0\leq\rho\leq t_{\max}-t}$
is a $(\mathbb{P},\widehat{\mathbb{F}})$-martingale by virtue of the
It\^o formula and the system \eqref{eq:bvp1} satisfied by $f$, and
therefore
\begin{equation}
f(t,x,y) = N_{t} = \mathbb{E}[N_{\tau}] =
\mathbb{E}[\e^{-r\tau}u(t+\tau,xG^{y}_{\tau},\widehat{Y}^{y}_{\tau})], 
\label{eq:fexrep}
\end{equation}
where we have used the boundary condition in (\ref{eq:bvp1}) to obtain
the last equality. 

Since $\mathcal{R}\subset\widehat{\mathcal{C}}$,
$(t+\tau,xG^{y}_{\tau},\widehat{Y}^{y}_{\tau})\in\widehat{\mathcal{C}}$,
so $\tau$ must satisfy
\begin{equation*}
\tau\leq\tau^{*}(t,x,y) := \inf\{\rho\in[0,T-t):
u(t+\rho,xG^{y}_{\rho},\widehat{Y}^{y}_{\rho}) = (xG^{y}_{\rho} -
K)^{+}\}\wedge (T-t).
\end{equation*}
In other words, $\tau$ must be less than or equal to the smallest
optimal stopping time $\tau^{*}(t,x,y)$ for the starting state
$(t,x,y)$. Now, the stopped process
\begin{equation*}
\e^{-r(\rho\wedge\tau^{*}(t,x,y))}u\left(t +
(\rho\wedge\tau^{*}(t,x,y)), xG^{y}_{\rho\wedge\tau^{*}(t,x,y)},
\widehat{Y}^{y}_{\rho\wedge\tau^{*}(t,x,y)}\right), \quad 0\leq \rho\leq T-t,
\end{equation*}
is a martingale, so this and the optional sampling theorem yield that
\begin{equation}
\mathbb{E}\left[\e^{-r\tau}u(t+\tau,xG^{y}_{\tau},\widehat{Y}^{y}_{\tau})\right]
= u(t,x,y).
\label{eq:uexrep}
\end{equation}
Then \eqref{eq:fexrep} and \eqref{eq:uexrep} show that $f$ and $u$
agree on $\mathcal{R}$ (and hence also on $\widehat{\mathcal{C}}$
since $\mathcal{R}\subset\widehat{\mathcal{C}}$ and
$(t,x,y)\in\mathcal{R}$ were arbitrary). Thus, $u$ satisfies
\eqref{eq:pdepic}.

Finally, to show uniqueness, let $g$ defined on the closure of
$\widehat{\mathcal C}$ be a solution to the system
(\ref{eq:pdepic})--(\ref{eq:bcpi0}). For starting state $(0,x,y)$ such
that $x<x^{*}(0,y)$ define
\begin{equation*}
L_{t} := \e^{-rt}g(t,xG^{y}_{t},\widehat{Y}^{y}_{t}), \quad t\in[0,T],  
\end{equation*}
as well as the optimal stopping time for $u(0,x,y)$, given by
\begin{equation*}
\tau^{*}(x,y) := \inf\{t\in[0,T): xG^{y}_{t}\geq
x^{*}(t,\widehat{Y}^{y}_{t})\}\wedge T.
\end{equation*}
The It\^o formula yields that
$(L_{t\wedge\tau^{*}(x,y)})_{t\in[0,T]}$ is a martingale. Then,
optional sampling along with the fact that $\tau^{*}(x,y)$ attains the
supremum in (\ref{eq:pivfrep}) starting at time zero, yields that
\begin{eqnarray*}
g(0,x,y) = L_{0} & = & \mathbb E[L_{\tau^{*}(x,y)}] \\
& = & \mathbb
E\left[\e^{-r\tau^{*}(x,y)}g(\tau^{*}(x,y),xG^{y}_{\tau^{*}(x,y)},
\widehat{Y}^{y}_{\tau^{*}(x,y)})\right] \\
& = & \mathbb E\left[\e^{-r\tau^{*}(x,y)}(xG^{y}_{\tau^{*}(x,y)} -
      K)^{+}\right] \\
& = & u(0,x,y), 
\end{eqnarray*}
so that the solution is unique.

\end{proof}

\subsection{Partial information smooth fit condition}
\label{subsec:sfcpi}

We establish, in Theorem \ref{thm:sppi} below, a smooth pasting
property for the partial information value function. This is a natural
property one might expect to hold, but to the best of our knowledge
has not been established before in a diffusion model model such as our
partial information model. In stochastic volatility models, Touzi
\cite{touzi99} has used variational inequality techniques to show the
smooth pasting property. It may be that this method could be adapted
to our setting.

We shall employ a method more akin to the classical proof of smooth
fit in American option problems, in a similar spirit to Karatzas and
Shreve \cite[Lemma 2.7.8]{ks98} (for the case of the Black-Scholes
put) or Monoyios and Ng \cite[Theorem 3.4]{mman11} (in a model with
inside information). The proof of Theorem \ref{thm:sppi} is simplified
by using the measure $\mathbb{P}^{*}\sim\mathbb{P}$ defined in
\eqref{eq:Pstar}. Because the proof involves analysing the first time
the stock almost surely breaches a surface, and as we are working in
the observation filtration, any early exercise crossing point must
ultimately depend only on the stock price path, so moving to a measure
where $X$ has constant drift (equal to $\mu_{0}$ under
$\mathbb{P}^{*}$, recall the SDE \eqref{eq:sdeXstar}) simplifies
matters.

Put explicitly, any optimal early exercise time will be the first time
$t\in[0,T)$ that we have $X_{t}\geq x^{*}(t,\widehat{Y}_{t})$. In this
relation, the process $\widehat{Y}$ depends on the history of the
stock price, through the history-dependence of the process
$\Phi\equiv\widehat{Y}/(1-\widehat{Y})$ in \eqref{eq:Phirep} (see also
equation \eqref{eq:PhirepX} in Section \ref{subsec:chv}, where we make
explicit the dependence of $\Phi$ on the history of the stock price),
so the early exercise crossing point is indeed dependent only on the
stock price (albeit in a path-dependent manner) and this makes our
method of proof work. This in turn can ultimately be traced to the
fact that, under the observation filtration, both the stock $X$ and
the filtered change point process $\widehat{Y}$ are driven by the
\emph{same} one-dimensional Brownian motion. Put yet another way, the
full information incomplete model with an observed but unhedgeable
change point has been rendered into a complete model with two
diffusion processes driven by one Brownian motion. This is a not
uncommon feature in filtering models. The price one pays for this
induced market completeness is that the second factor $\widehat{Y}$
depends on the entire history of the stock price, also a not uncommon
feature of models with filtering -- this is the ``learning'' aspect of
filtering coming to the fore.

\begin{theorem}[Smooth pasting: partial information value function]
\label{thm:sppi}

The partial information value function defined in (\ref{eq:pivf})
satisfies the smooth pasting property
\begin{equation*}
\frac{\partial u}{\partial x}(t,x^{*}(t,y),y) =  1, \quad
t\in[0,T), \quad y\in[0,1],
\end{equation*}
at the optimal exercise threshold $x^{*}(t,y)$. 

\end{theorem}

\begin{proof}

In this proof it entails no loss of generality if we set $r=0$ and
$t=0$, but this considerably simplifies notation, so let us proceed in
this way. Write $u(x,y)\equiv u(0,x,y)$ and $x^{*}(y)\equiv
x^{*}(0,y)$ for brevity.

The map $x\to u(x,y)$ is convex and non-decreasing, so we have
$u_{x}(x,y)\leq 1$ in the continuation region
$\widehat{\mathcal{C}}=\{(x,y)\in\mathbb
R_{+}\times[0,1]:x<x^{*}(y)\}$,
and thus $u_{x}(x^{*}(y)-,y)\leq 1$. We also have $u_{x}(x,y)=1$ in
the stopping region
$\widehat{\mathcal{S}}=\{(x,y)\in\mathbb R_{+}\times[0,1]:x\geq
x^{*}(y)\}$,
and thus $u_{x}(x^{*}(y)+,y)=1$. Hence, the proof will be complete if
we can show that $u_{x}(x^{*}(y)-,y)\geq 1$.
Recall the measure $\mathbb{P}^{*}$ defined in \eqref{eq:Pstar}, and the
$(\mathbb{P}^{*},\widehat{\mathbb{F}})$-dynamics of the stock in
\eqref{eq:sdeXstar}. Given $X_{0}=x$, the stock price at time $t\in[0,T]$ is
\begin{equation*} 
X_{t} = xG_{t} := x\exp\left(\left(\mu_{0} - \frac{1}{2}\sigma^{2}\right)t
+ \sigma {W}^{*}_{t}  \right), \quad t\in[0,T].
\end{equation*}
For $(x,y)\in\mathbb R_{+}\times(0,1)$, denote by $\tau(x,y)$ the
optimal $\widehat{\mathbb{F}}$-stopping time for $u(x,y)$, given by
the first time the stock breaches the exercise surface at the
prevailing value of $\widehat{Y}$. Working under $\mathbb{P}^{*}$, we
thus have
\begin{equation*}
\tau(x,y) = \inf\{t\in[0,T):xG_{t}\geq
x^{*}(t,\widehat{Y}^{y}_{t})\}\wedge T, 
\end{equation*}
where $\widehat{Y}^{y}$ denotes the filtered change point process with
initial condition $Y_{0}=y$.

Set $x=x^{*}(y)\geq K$, which will be fixed for the remainder of the
proof, and define
\begin{equation*}
\tau(x-\epsilon,y) :=  \inf\{t\in[0,T):(x-\epsilon)G_{t}\geq
x\}\wedge T,
\end{equation*}
for $\epsilon\geq 0$, and the dependence on $y$ on the right-hand-side
is of course suppressed in $x\equiv x^{*}(y)$. We have
$\tau(x,y)\equiv 0$ and that $\tau(x-\epsilon,y)$ is non-decreasing in
$\epsilon$. Moreover, because the exercise surface is non-increasing
in time and in $y$, we have
\begin{equation}
\tau(x-\epsilon,y) \leq \inf\{t\in[0,T):(x-\epsilon)G_{t}\geq
x\}\wedge T.
\label{eq:tauxey}
\end{equation}
The Law of the Iterated Logarithm for the Brownian motion $W^{*}$
(Karatzas and Shreve \cite[Theorem 2.9.23]{ks91}) implies that
\begin{equation*}
\sup_{0\leq t\leq a}G_{t} > 1, \quad \mbox{$\mathbb{P}^{*}$-a.s.}  
\end{equation*}
for every $a>0$, so there will exist a sufficiently small $\epsilon>0$
such that
\begin{equation*}
\sup_{0\leq t\leq a}(x-\epsilon)G^{y}_{t} \geq x, \quad
\mbox{$\mathbb{P}^{*}$-a.s.}
\end{equation*}
for every $a>0$. Thus, the right-hand-side of \eqref{eq:tauxey} tends
to zero as $\epsilon\downarrow 0$, and therefore
$\tau(x-\epsilon,y) \downarrow 0$ as $\epsilon\downarrow 0$,
$\mathbb{P}^{*}$-almost surely and, since
$\mathbb{P}^{*}\sim\mathbb{P}$, this is also true $\mathbb{P}$-almost
surely:
\begin{equation}
\tau(x-\epsilon,y) \downarrow 0 \quad \mbox{as} \quad
\epsilon\downarrow 0, \quad \mbox{$\mathbb{P}$-almost surely}.
\label{eq:limtau}
\end{equation}
Using the fact that $\tau(x-\epsilon,y)$ will be sub-optimal for the
starting state $(X_{0},\widehat{Y}_{0})=(x,y)$, we have
\begin{eqnarray}
&& u(x,y) - u(x-\epsilon,y) \label{eq:uxuxe} \\
& \geq & \mathbb{E}\left[\left((xG_{\tau(x-\epsilon,y)} - K)^{+} -
((x-\epsilon)G_{\tau(x-\epsilon,y)} - K)^{+}\right)\right] \nonumber \\
& \geq & \mathbb{E}\left[\left((xG_{\tau(x-\epsilon,y)} - K)^{+} -
((x-\epsilon)G_{\tau(x-\epsilon,y)} -K)^{+}\right)
\mathbbm{1}_{\{(x-\epsilon)G_{\tau(x-\epsilon,y)}\geq K\}}\right]
\nonumber \\
& = & \epsilon\mathbb{E}\left[G_{\tau(x-\epsilon,y)}
\mathbbm{1}_{\{(x-\epsilon)G_{\tau(x-\epsilon,y)}\geq
K\}}\right]. \nonumber 
\end{eqnarray}
We now take the limit as $\epsilon\downarrow 0$. Using
\eqref{eq:limtau} we almost surely have
$\lim_{\epsilon\downarrow 0}G_{\tau(x-\epsilon,y)}=1$ and, since it is
never optimal to exercise below the strike,
$\lim_{\epsilon\downarrow 0}
\mathbbm{1}_{\{(x-\epsilon)G_{\tau(x-\epsilon,y)}\geq K\}}=1$. Using
these properties, along with the uniform integrability of
$(G_{t})_{t\in[0,T]}$, in \eqref{eq:uxuxe}, we compute
\begin{equation*}
u_{x}(x-,y) = \lim_{\epsilon\downarrow 0}\frac{1}{\epsilon}(u(x,y) -
u(x-\epsilon,y)) \geq 1,  
\end{equation*}
which completes the proof.

\end{proof}

\subsection{Doob-Meyer decomposition of partial information
  Snell envelope}
\label{subsec:dmdpise}

As was done in the full information case, with the free boundary PDE
and smooth pasting condition established for the partial information
value function, we can now derive a Doob-Meyer decomposition for the
partial information Snell envelope of the reward process, and this
allows us to prove Lemma \ref{lem:tvpartial} on the left-limiting
value of the partial information exercise surface as we approach
maturity.

Recall that the partial information Snell envelope is the c\'adl\'ag
supermartingale identified with the discounted ESO value process
$(\e^{-rt}U_{t})_{t\in[0,T]}$, with $U_{t}=u(t,X_{t},\widehat{Y}_{t})$. 

\begin{lemma}[Doob-Meyer decomposition of partial information
Snell envelope]
\label{lem:dmdpise}

The process $(\e^{-rt}u(t,X_{t},\widehat{Y}_{t}))_{t\in[0,T]}$ admits
the decomposition
\begin{equation}
\e^{-rt}u(t,X_{t},\widehat{Y}_{t}) = u(0,X_{0},\widehat{Y}_{0}) +
M_{t} - A_{t}, \quad t\in[0,T],
\label{eq:dmdecomppi}
\end{equation}
where
\begin{equation*}
M_{t} := \int_{0}^{t}\e^{-rs}\left((\sigma
X_{s}u_{x}(s,X_{s},\widehat{Y}_{s}) -
\eta\widehat{Y}_{s}(1-\widehat{Y}_{s})u_{y}(s,X_{s},\widehat{Y}_{s}))\right)\ud
\widehat{W}_{s}, \quad t\in[0,T], 
\end{equation*}
is a $(\mathbb{P},\widehat{\mathbb{F}})$-martingale, and
\begin{equation*}
A_{t} :=
\int_{0}^{t}\e^{-rs}\left((r-\mu_{0}+\sigma\eta\widehat{Y}_{s})X_{s}-rK\right)
\mathbbm{1}_{\{X_{s}\geq  x^{*}(s,\widehat{Y}_{s})\}}\ud s, \quad t\in[0,T],
\end{equation*}
is a non-decreasing finite variation process.

\end{lemma}

\begin{proof}

The proof is similar to the corresponding proof of Theorem
\ref{thm:dmdecomp} in the full information scenario, so we shall be
more brief here. Using the generalised It\^o formula for convex
functions, the PDE \eqref{eq:pdepic} satisfied by
$u(\cdot,\cdot,\cdot)$ in the continuation region
$\widehat{\mathcal{C}}$ and the fact that $u(t,x,y)=x-K$ in the
stopping region, we obtain the decomposition
\eqref{eq:dmdecomppi}. The square integrability of the stock price and
bounded nature of the derivatives $u_{x},u_{y}$ in $M$ imply that $M$
is indeed a martingale. Since the Snell envelope is a super-martingale
with a unique Doob-Meyer decomposition into a martingale minus a
non-decreasing process of finite variation, we conclude that $A$ is a
non-decreasing process.
  
\end{proof}

Some observations on the parameter values for which we obtain a
bounded exercise surface are in order. With
$\mu(\widehat{Y})\equiv \mu_{0}-\sigma\eta\widehat{Y}$ the partial
information stock price drift, the non-decreasing property of the
process $A$ in Lemma \ref{lem:dmdpise} means that we have
$((r-\mu(\widehat{Y}_{t}))X_{t}-rK) \mathbbm{1}_{\{X_{t}\geq
x^{*}(t,\widehat{Y}_{t})\}}\geq 0$ almost surely, for all
$t\in[0,T]$, and hence we also have
$(r-\mu(\widehat{Y}_{t}))x^{*}(t,\widehat{Y}_{t})-rK\geq 0$. Now,
suppose we have $\mu(\widehat{Y_{t}})\geq r$ almost surely for all
$t\in[0,T]$. We then compute that
$x^{*}(t,\widehat{Y}_{t})\leq -(r/(\mu(\widehat{Y}_{t})-r))K$, which
is impossible, since the exercise surface cannot lie below the
strike. We conclude that, when the stock drift exceeds the interest
rate, the finite variation process in the Doob-Meyer decomposition
will be zero, and the ESO value process is a martingale. This is of
course exactly in line with Remark \ref{rem:murpi}, that early
exercise will not occur if the stock drift dominates the interest
rate, in which case the ESO value process is a martingale and equal to
the European version of the ESO.

We are now ready to prove Lemma \ref{lem:tvpartial}.

\begin{proof}[Proof of Lemma \ref{lem:tvpartial}]

From the non-decreasing property of the process $A$ in Lemma
\ref{lem:dmdpise} we have $((r-\mu(\widehat{Y}_{t}))X_{t}-rK)
\mathbbm{1}_{\{X_{t}\geq x^{*}(t,\widehat{Y}_{t})\}}\geq 0$ almost
surely, for all $t\in[0,T]$, and hence we also have
$(r-\mu(\widehat{Y}_{t}))x^{*}(t,\widehat{Y}_{t})-rK\geq 0$.

Suppose that $\mu(\widehat{Y}_{t})<r$. In this case, we conclude that
$x^{*}(t,\widehat{Y}_{t})\geq (r/(r-\mu(\widehat{Y}_{t})))K$. From the
fact that the exercise surface is non-increasing in time, we conclude
that we have the terminal left-limit lower bound
\begin{equation*}
x^{*}_{i}(T-,y) \geq \left(\frac{r}{r-\mu_{0}+\sigma\eta y}\right)K,
\end{equation*}
for all values of $y\in[0,1]$ satisfying $\mu_{0}-\sigma\eta
y<r$. There are now two cases to consider separately, which lead to a
refinement of this lower bound:

\begin{itemize}
\item for $0\leq \mu_{0}-\sigma\eta y <r$, we obtain $x^{*}(T-,y)\geq
\left(r/(r-\mu_{0}+\sigma\eta y)\right)K\geq K$;

\item for $\mu_{0}-\sigma\eta y \leq 0 <r$, because it is never
optimal to exercise below the strike, we have
$x^{*}(T-,y)\geq K >(r/(r-\mu_{0}+\sigma\eta y))K$.
  
\end{itemize}

We thus have, in all cases, the refined lower bound
\begin{equation*}
x^{*}(T-,y) \geq
\max\left(K,\left(\frac{r}{r-\mu_{0}+\sigma\eta y}\right)K\right),
\quad \mu_{0}-\sigma\eta y <r.
\end{equation*}
We now show that in fact we have equality here, thus establishing
\eqref{eq:xstm}. Suppose, to the contrary, that we have
$x^{*}(T-,y)>\max\left(K,\left(r/(r-\mu_{0}+\sigma\eta
y)\right)K\right)$. Fixing $y\in[0,1]$, consider a value
$x\in\left( \max\left(K,\left(r/(r-\mu_{0}+\sigma\eta
y)\right)K\right),x^{*}(T-,y)\right)$. Then, for $0\leq t<T$, we
have $(t,x,y)\in\widehat{\mathcal{C}}$, so that
$u(t,x,y)>(x-K)^{+}=x-K$. Using temporal continuity of
$u(\cdot,\cdot,\cdot)$, we thus obtain
$u(T,x,y)=\lim_{t\uparrow T}u(t,x,y)> x-K$. But, on the other hand, we
know that at maturity we have $u(T,x,y)=(x-K)^{+}=x-K$, so we have a
contradiction. Thus, \eqref{eq:xstm} holds.

\end{proof}

\subsection{A comment on a change of state variable}
\label{subsec:chv}

In this section, we illustrate the inherent complexity of the partial
information case, due to its path-dependent structure.  Consider the
partial information problem (\ref{eq:piproblem}). We shall change
measure to $\mathbb{P}^{*}$ defined in (\ref{eq:Pstar}), and this
naturally leads to a change of state variable from $(X,\widehat{Y})$
to $(X,\Phi)$, with $\Phi$ defined in (\ref{eq:Phi}). This leads to
the following lemma.

\begin{lemma}
\label{lem:change}

Let $\Phi$ be the likelihood ratio process defined in
(\ref{eq:Phi}). The partial information ESO value process $U$ in
(\ref{eq:piproblem}) satisfies
\begin{equation}
\e^{-(r+\lambda)t}(1+\Phi_{t})U_{t} =
\esssup_{\tau\in\widehat{\mathcal{T}}_{t,T}}\mathbb
E^{*}\left[\e^{-(r+\lambda)\tau}(1+\Phi_{\tau})
(X_{\tau}-K)^{+}\vert\widehat{\mathcal{F}}_{t}\right], \quad t\in[0,T],
\label{eq:Urep}
\end{equation}
where $\mathbb E^{*}[\cdot]$ denotes expectation with respect to
$\mathbb P^{*}$ in (\ref{eq:Pstar}), and the $(\mathbb
P^{*},\widehat{\mathbb F})$-dynamics of $X,\Phi$ are given in
(\ref{eq:sdeXstar}) and (\ref{eq:sdephi}). 

\end{lemma}

\begin{proof}

Let $Z$ denote the change of measure martingale defined by
\begin{equation}
Z_{t} := \frac{1}{\Gamma}_{t} = \left.\frac{\ud\mathbb P}{\ud\mathbb
P^{*}}\right\vert_{\widehat{\mathcal{F}}_{t}} 
= \mathcal{E}(-\eta\widehat{Y}\cdot W^{*})_{t}, \quad t\in[0,T], 
\label{eq:zed}
\end{equation}
satisfying
\begin{equation}
\ud Z_{t} = -\eta\widehat{Y}_{t}Z_{t}\ud W^{*}_{t}, \quad Z_{0}=1.  
\label{eq:Zsde}
\end{equation}
The It\^o formula along with the dynamics of $\Phi$ in
(\ref{eq:sdephi}) yields that $Z$ is given in terms of $\Phi$ as
\begin{equation}
Z_{t} = \e^{-\lambda t}\left(\frac{1 + \Phi_{t}}{1 + \Phi_{0}}\right),
\quad t\in[0,T],
\label{eq:Zrep}
\end{equation}
because the right-hand-side of (\ref{eq:Zrep}) satisfies the SDE
(\ref{eq:Zsde}). Then an application of the Bayes formula to the
definition of $U$ in (\ref{eq:piproblem}) yields the result.

\end{proof}

The point of \eqref{eq:Urep} is that the state variables in the
objective function have decoupled dynamics under $\mathbb{P}^{*}$
(recall \eqref{eq:sdeXstar} and \eqref{eq:sdephi}). However, the
problematic feature of the history dependence of $\Phi$ remains, as
exhibited in \eqref{eq:Phirep}, inheriting this feature from the
filtered change-point process $\widehat{Y}$. Indeed, using the solution
of the stock price SDE \eqref{eq:sdeXstar}, the representation
\eqref{eq:Phirep} may be converted to one involving the stock price
and its history, as follows.

With $X_{0}=x$, from \eqref{eq:sdeXstar} we have
$X_{t}=x\exp\left(\mu_{0}-\frac{1}{2}\sigma^{2}\right)t
+ \sigma W^{*}_{t},\,t\geq 0$, so that
\begin{equation*}
\exp(\sigma W^{*}_{t}) = \left(\frac{X_{t}}{x}\right)\exp\left(\mu_{0} -
    \frac{1}{2}\sigma^{2}\right)t, \quad t\geq 0.  
\end{equation*}
Using this relation to compute the process
$\Lambda=\mathcal{E}(-\eta W^{*}) $ we get
\begin{equation}
\Lambda_{t} = \exp\left(-\eta W^{*}_{t} - \frac{1}{2}\eta^{2}t\right)
= \left(\frac{X_{t}}{x}\right)^{-\eta/\sigma}\exp\left(\eta\nu_{0} -
\frac{1}{2}\eta^{2}\right)t, \quad t\geq 0,
\label{eq:LambdaX}
\end{equation}
where
\begin{equation*}
\nu_{0} := \frac{\mu_{0}}{\sigma} - \frac{1}{2}\sigma.  
\end{equation*}
Then,  with $\Phi_{0}=\phi$, substituting \eqref{eq:LambdaX} into
\eqref{eq:Phirep}, we obtain
\begin{equation}
\Phi_{t}(\phi) = \phi\e^{\kappa t}\left(\frac{X_{t}}{x}\right)^{-\eta/\sigma} + 
\lambda\int_{0}^{t}\e^{\kappa(t-s)}\left(\frac{X_{t}}{X_{s}}\right)^{-\eta/\sigma}\ud
s, \quad t\in[0,T],
\label{eq:PhirepX}
\end{equation}
where $\kappa$ is a constant given by
\begin{equation*}
\kappa := \lambda + \eta\nu_{0} - \frac{1}{2}\eta^{2}.
\end{equation*}
The second term on the right-hand-side of \eqref{eq:PhirepX} is the
awkward history-dependent term which makes numerical solution of the
partial information ESO problem difficult. For $\lambda=0$, we see
that $\Phi$ becomes a deterministic function of the current stock
price, and this limit corresponds to a simpler model in which an
unknown drift is assumed to take one of two values, but the agent is
unsure which value pertains in reality, and so filtering is used to
estimate the drift. A number of papers have used such a model and
exploited the absence of path-dependence to reduce the dimension of
the problem (see D{\'e}camps et al.~\cite{dmv05,dmv09}, Klein
\cite{klein09}, and Ekstr{\"o}m and co-authors
\cite{el11,el13,ev19}). This simplification is not available to us, so
the partial information problem is potentially more challenging to
solve numerically.

\section{On the effect of a vesting period on ESO exercise}
\label{sec:vesting}

ESOs often include a contractual feature called a vesting period, a
period of time during which option exercise is not permitted.  In this
section, we briefly describe the effect of a vesting period on the
exercise of ESOs in the full and partial information models. In
Section \ref{sec:numex} we shall also demonstrate the impact of
vesting on ESO value.

Suppose there is a vesting period $[0,t_{v})$, so that the ESO can
only be exercised in the time interval $[t_{v},T]$. Then we seek
optimal stopping times, with respect to the appropriate filtration,
lying in the exercise interval $[t_{v},T]$. Thus, for $0\leq t<t_{v}$,
the discounted full information ESO value process is
\begin{equation}
\e^{-rt}\check{V}_{t} =
\esssup_{\tau\in\mathcal{T}_{t_{v},T}}\mathbb{E}[R_{\tau}|\mathcal{F}_{t}],
\quad t\in[0,t_{v}),
\label{eq:prevesting}
\end{equation}
while for $t\in[t_{v},T]$, the vesting period is over, and we have
reverted back to our original problem without a vesting period with
value process $(V_{t})_{t\in[t_{v},T]}$, given by
\begin{equation}
\e^{-rt}V_{t} =
\esssup_{\tau\in\mathcal{T}_{t,T}}\mathbb{E}[R_{\tau}|\mathcal{F}_{t}],
\quad t\in[t_{v},T].  
\label{eq:postvesting}
\end{equation}
Note that $\check{V}_{t}\leq V_{t}$ for $t<t_{v}$ (the value with
vesting is clearly dominated by the one without vesting, due to the
extra exercise opportunities).

Similarly, for $t\in[0,t_{v})$, the discounted partial information
value process is
\begin{equation*}
\e^{-rt}\check{U}_{t} =
\esssup_{\tau\in\widehat{\mathcal{T}}_{t_{v},T}}\mathbb{E}
[R_{\tau}|\widehat{\mathcal{F}}_{t}], \quad t\in[0,t_{v}),  
\end{equation*}
satisfying $\check{U}_{t}\leq U_{t}$ for $t<t_{v}$ while for
$t\in[t_{v},T]$ we are back to our original problem without a vesting
period:
\begin{equation*}
\e^{-rt}U_{t} =
\esssup_{\tau\in\widehat{\mathcal{T}}_{t,T}}\mathbb{E}
[R_{\tau}|\widehat{\mathcal{F}}_{t}], \quad t\in[t_{v},T].  
\end{equation*}

The key overall idea is well expressed by Leung and Sircar \cite[Section
5.1.1]{ls092}, as follows: ``When a vesting period of $t_{v}$ years is
imposed, the employee cannot exercise the ESO during $[0,t_{v})$, but
the post-vesting exercising strategy will be unaffected.''

In what follows, we examine the situation where we have
$y_{0}=0$, $\mu_{0}>r$, $\mu_{1}<r$. Thus, for the full information
problem, no exercise will occur before the strictly positive change
point $\theta\sim\mathrm{Exp}(\lambda)$, as the reward process
$(R_{t})_{0,\theta}$ over the time interval up to the change point is
a sub-martingale.

\subsection{The full information case}
\label{subsec:tfic}

First, consider the case that the change point occurs after the
vesting period has elapsed, that is, $\theta\geq t_{v}$. For
$t<t_{v}$, no exercise can occur, and at $t=t_{v}$ we revert back to
our original problem, the vesting period having elapsed. The
post-vesting exercise strategy will then be as in the no-vesting case.

Next, consider the case $\theta<t_{v}$, that is, the change point
occurs during the vesting period. For $t\in[0,t_{v})$ there is no
exercise as we are still in the vesting period. At $t=t_{v}$, we are now in the
low-drift state, so the stock is a GBM with drift $\mu_{1}<r$. There
will now be an exercise boundary $(x^*_{1}(t))_{t_{v}\leq t\leq T}$. If
$X_{t_{v}}\geq x^*_{1}(t_{v})$, then we are in the exercise region as
soon as the vesting period has elapsed, and immediate exercise occurs
at $t=t_{v}$. If, on the other hand, $X_{t_{v}}< x^*_{1}(t_{v})$, then
there is no immediate exercise at $t_{v}$, and exercise occurs the
first time that the stock breaches the boundary from below, at time
$\bar{\tau}=\inf\{t\in[t_{v},T):X_{t}\geq x^*_{1}(t)\}\wedge T$.

Thus, the overall conclusion is: the exercise boundary is infinite
over $[0,t_{v})$, regardless of when the change point occurs. If the
change point has occurred by time $t_{v}$, then immediate exercise
occurs at time $t_{v}$ if the prevailing stock price at $t_{v}$ is
higher than or equal to the exercise boundary $x^*_{1}(t_{v})$ at that
point. If the change point has not occurred by time $t_{v}$, we are
back to our original problem over the interval $[t_{v},T]$.

\subsection{The partial information case}
\label{subsec:tpic}

Regardless of when the change point occurs, if we are in the vesting
period $[0,t_{v})$, no exercise can occur, so the partially informed
agent's exercise surface is infinite.

At $t=t_{v}$ we revert back to our original problem, the vesting
period having elapsed. Again, this is regardless of whether the change
point has occurred or not (the partially informed agent is not aware
of the change point having occurred or not, and is therefore filtering
it from stock price observations). We now have an optimal exercise
surface $(x^{*}(t,y))_{t_{v}\leq t\leq T,0\leq y\leq 1}$, and exercise
occurs the first time that the stock breaches the exercise surface
evaluated at the prevaling value of $\widehat{Y}$, that is, at
$\tau^{*}=\inf\{t\in[t_{v},T):X_{t}\geq
x^{*}(t,\widehat{Y}_{t})\}\wedge T$.

In other words, the post-vesting exercise strategy will then be as in
the no-vesting case, with the pre-vesting boundary set to infinity.

\section{Numerical scheme and convergence tests}
\label{sec:nrsim}

In this section, we describe numerical schemes for the PDEs in the
full and partial information case, and present numerical studies to
illustrate the convergence and computational complexity.  We present
our novel algorithm for the two-dimensional, degenerate free boundary
value problem in the partial information case in some detail and
analyse its convergence properties, while we only state the simple
scheme for the full information case.  Note that alternative numerical
methods could be employed, for example, a binomial scheme
(non-recombining for the partial information case) or a
Longstaff-Schwartz Monte Carlo approach. However, the finite
difference schemes we propose are far superior in terms of speed and
accuracy.

\subsection{The partial information case}
\label{subsec:picn}

We begin by noting that the partial information ESO value function
$u(\cdot,\cdot,\cdot)$ satisfying (\ref{eq:pdepic})-(\ref{eq:bcpi0})
is also the unique solution in $[0,T]\times\mathbb R_{+}\times[0,1]$
of the equivalent linear complementarity problem
\begin{eqnarray}
\label{lcp-incomp}
\min\left(-\mathcal{L}u(t,x,y), u - (x-K)^{+}  \right) & = & 0, \; 
t\in[0,T), \; \quad x\in \mathbb R_{+}, \; y\in[0,1], \label{eq:pdepic2}\\  
u(T,x,y) & = & (x-K)^{+}, \quad x\in\mathbb R_{+}, \quad
y\in[0,1], \label{eq:bcT2}
\end{eqnarray}
where we repeat for convenience that
\begin{equation}
\label{rep_def:L}
\mathcal{L} = \frac{\partial}{\partial t} +
\mathcal{L}_{X,\widehat{Y}} - r
\end{equation}
with
\begin{equation*}
\mathcal{L}_{X,\widehat{Y}}f(t,x,y) = (\mu_{0}-\sigma\eta y)xf_{x}
+ \frac{1}{2}\sigma^{2}x^{2}f_{xx} + \lambda(1-y)f_{y} +
\frac{1}{2}\eta^{2}y^{2}(1-y)^{2}f_{yy} - \sigma\eta xy(1-y)f_{xy}
\end{equation*}
for any sufficiently smooth function $f:[0,T]\times\mathbb
R_{+}\times[0,1]$.

The degeneracy of the equation requires the notion of viscosity
solutions for a rigorous analysis. A general framework of so-called
monotone schemes for the approximation of viscosity solutions to
nonlinear PDEs was first introduced and analysed in Barles and
Souganidis \cite{barles1991convergence}.  It is well-documented in the
literature that the monotone approximation of degenerate diffusion
problems in multiple dimensions generally requires complicated,
so-called `wide stencil' schemes (see, for example, Debrabant and
Jacobsen \cite{debrabant2013semi}, Ma and Forsyth
\cite{ma2016unconditionally}).  The analysis in Reisinger
\cite{reisinger2018non} demonstrates clearly that the construction
becomes more difficult when the correlation approaches $\pm 1$, the
above case being such a singular limit of perfect negative correlation
between the driver of $X$ and $Y$.  Moreover, all schemes known to us
which are monotone for general, possibly degenerate multidimensional
equations, have convergence order no larger than 1 in the mesh size
and time step.

Initial numerical experiments with standard, non-monotone finite
difference schemes for the above PDE, in particular the 7-point and
9-point stencils for the diffusion term, exhibited severe
instabilities for small mesh sizes.

In the following construction, we take advantage of a problem-specific
coordinate transformation which allows us to define a monotone, second
order accurate approximation to the second order terms. This will be
supplemented with either monotone and first order, or non-monotone and
second order, backward differentiation formulae (BDF) for the first
order derivative terms.

The second order version of the method is not theoretically guaranteed
to converge to the viscosity solution in the degenerate case, however,
recent results in Bokanowski and Debrabant \cite{bokanowski2018high}
and Bokanowski et al.~\cite{bokanowski2018stability} show stability of
BDF schemes in more regular cases and we will demonstrate excellent
empirical properties of the scheme below.

\subsubsection{Mesh construction and diffusion approximation}
\label{subsubsec:mcda}

We begin by simultaneously constructing a computational domain
$[K^2/x_{\max},x_{\max}]\times [y_{\min},1-y_{\min}] \subset
\mathbb{R}_+ \times (0,1)$ and a non-uniform tensor-product mesh on
that domain, where $x_{\max}$ and $y_{\min}$ will be chosen so as to
make the impact that imposing approximate data at the boundary has on
the quantities of interest negligible.

We first fix $x_{\max}$ and a positive integer $N$ to define the
$x$-coordinates of the mesh nodes by
\begin{eqnarray}
\label{x_mesh}
x_i = K \exp(\sigma (i-N/2) h), \quad 0\le i \le N,
\end{eqnarray}
so that $x_{N/2}=K$ for even $N$ and $h$ is chosen such that
$x_N= x_{\max}$. This non-uniform mesh is motivated by the
observation that the log transform $X\rightarrow \log X/\sigma$ leads
to a standard Brownian motion with stochastic drift, i.e.\ satisfying
the SDE
\begin{eqnarray}
\label{dlogX}
\ud \left( \frac{1}{\sigma} \log X_t \right) = \ud\widehat{W}_t +
\left( \frac{1}{\sigma} \left(\mu_0 - \sigma \eta
\widehat{Y}_t\right) - \frac{1}{2} \sigma \right) \, \ud t, 
\end{eqnarray}
and turns the differential operator $\mathcal{L}_{X,\widehat{Y}}$ into
one with constant coefficients in $x$.

By a similar application of It{\^o}'s formula, one can further derive
that, for $\hat{Y} \ne 0$ or $1$,
\begin{eqnarray}
\label{dlogY}
\ud \left( \frac{1}{\eta} \log\left(\frac{\widehat{Y}_t}{1-\widehat{Y}_t}\right) \right) = -\ud\widehat{W}_t + \left(\frac{1}{2} \eta (2\widehat{Y}_t-1) +
 \lambda \frac{1}{\eta \widehat{Y}_t} \right) \, \ud t.
\end{eqnarray}
Inverting the map on the left-hand side, we define a mesh for the
$y$-coordinate by
\begin{eqnarray}
\label{eq:y_i}
y_j = \frac{\exp(\eta (j-L/2) h)}{1+\exp(\eta (j-L/2) h)}, \quad 0\le j \le L,
\end{eqnarray}
where $L$ is chosen such that $y_0=y_{\min}$ (and hence
$y_L = 1-y_{\min}$), a sufficiently small value, and centered at
$y_{L/2}=1/2$ for even $L$.

The purpose of these transformations is to fix the principal component
of the diffusion matrix to $(-1,1)$ and facilitate the construction of
a monotone, second order, narrow (i.e., using only neighbouring mesh
points) scheme.  More concretely, combining the identities above, we
obtain by simple Taylor expansion for smooth $f$,
\begin{eqnarray}
\label{2d-stencil}
(D^2 f)(t,x_i,y_j) &:=&
\frac{f(t,x_{i-1},y_{j+1})-2 f(t,x_{i},y_{j}) + f(t,x_{i+1},y_{j-1})}{h^2} \\
\nonumber
&=&\hspace{0 cm} \frac{1}{2}\sigma^{2}x_i^{2}f_{xx} +
\frac{1}{2}\eta^{2}y_j^{2}(1-y_j)^{2}f_{yy} - \sigma \eta x_i y_j(1-y_j)f_{xy} \\
\nonumber
&&\hspace{0 cm} + \frac{1}{2} \sigma^2 x_i f_{x} + \frac{1}{2} \eta^2 y_j (1-y_j) (1-2 y_j) f_{y} + O(h^2),
\end{eqnarray}
where the derivatives on the right-hand side are evaluated at $(t,x_i,y_j)$.

The important feature of (\ref{2d-stencil}) is that the second-order
part of the operator is approximated up to order two in $h$ by a
one-dimensional finite difference in a diagonal direction, plus some
first order terms.

\subsubsection{Drift approximation}
\label{subsubsec:da}

We define the drift coefficients in \eqref{dlogX} and \eqref{dlogY} by
\begin{eqnarray*}
\mu_x(t,x,y) := \frac{1}{\sigma}(\mu_{0}-\sigma\eta y) - \frac{1}{2}
\sigma, \qquad \mu_y(t,x,y) := \frac{\lambda}{\eta} \frac{1}{y}  - \frac{1}{2}\eta
(1-2 y),  
\end{eqnarray*}
(with the subscripts on $\mu_x$ and $\mu_y$ not denoting partial
derivatives).  These are precisely the the drifts of $X$ and
$\widehat{Y}$ minus the ``correction terms'' from \eqref{2d-stencil}
which have to be subtracted from $D^2$ for a consistent discretisation
of the second order terms in the PDE.

We approximate the first derivative in $x$, with coefficient $\mu_x$,
by an ``upwinding'' approximation
\begin{eqnarray*}
(\mu_x D_x f)(t,x_i,y_j) &=& \left(\mu_x(t,x_i,y_j)\right)^+ (D_x^+
f)(t,x_i,y_j) + \left(\mu_x(t,x_i,y_j)\right)^- (D_x^-
f)(t,x_i,y_j), 
\end{eqnarray*}
where $(\cdot)^\pm$ denotes the positive and negative part,
respectively, and $D_x^\pm$ is either the one-sided first order BDF1
approximation defined by
\begin{eqnarray*}
(\overline{D}_x^\pm f)(t,x_i,y_j) := \mp \frac{f(t,x_i,y_j)-
  f(t,x_{i\pm 1},y_j)}{h} = \sigma x f_x(t,x_i,y_j) + O(h), 
\end{eqnarray*}
or the one-sided second order BDF2 approximation
\begin{eqnarray*}
(\widehat{D}_x^\pm f)(t,x_i,y_j) := \mp \frac{3 f(t,x_i,y_j)- 4
  f(t,x_{i\pm 1},y_j)+ f(t,x_{i\pm 2},y_j)}{2 h} =  \sigma x
  f_x(t,x_i,y_j) + O(h^2). 
\end{eqnarray*}
Two approximations to the first $y$-derivative are defined analogously.

\subsubsection{Timestepping and overall scheme}
\label{subsubsec:taos}

Combining the approximations above, for all points $(t,x_i,y_j)$ where
$f$ is smooth we have
\begin{eqnarray*}
\overline{L} f := D^2 f + \mu_x \overline{D}_x f + \mu_y \overline{D}_y f = 
\mathcal{L}_{X,\widehat{Y}} f +
O(h), \\
\widehat{L} f := D^2 f + \mu_x \widehat{D}_x f + \mu_y \widehat{D}_y f = 
\mathcal{L}_{X,\widehat{Y}} f +
O(h^2).
\end{eqnarray*}

For the time discretisation, we follow Forsyth and Vetzal
\cite{forsyth2002quadratic} and Reisinger and Whitley
\cite{reisinger2014impact} to define a non-uniform time mesh of $M+1$
points $t_m = T - (\sqrt{T}- m k)^2$, $m = 0,...,M$, for
$k=\sqrt{T}/M$.  This transformation is motivated by the square-root
behaviour of both the exercise boundary and the value function at the
strike close to maturity. The limited regularity prevents second order
convergence of uniform timestepping schemes (see Forsyth and Vetzal
\cite{forsyth2002quadratic}).

Taking into account this time transformation, we introduce either the
BDF1 scheme (implicit Euler scheme)
\begin{eqnarray*}
\frac{f(t_{m+1}, x_i,y_j)-f(t_{m}, x_i,y_j)}{k} + 2 m k (\overline{L} f -r f)(t_{m}, x_i,y_j) =\\
\left(\frac{\partial}{\partial t} + \mathcal{L}_{X,\widehat{Y}} - r\right) f(t_{m}, x_i,y_j) + O(k) + O(h),
\end{eqnarray*}
where $\overline{L}$ uses the BDF1 scheme for the drift also, or the
BDF2 scheme
\begin{eqnarray*}
\frac{- f(t_{m+2}, x_i,y_j) + 4 f(t_{m+1}, x_i,y_j) - 3 f(t_{m},
  x_i,y_j)}{2 k} + 2 m k (\widehat{L} f -r f)(t_{m}, x_i,y_j) = \\
\left(\frac{\partial}{\partial t} + \mathcal{L}_{X,\widehat{Y}} -
  r\right) f(t_{m}, x_i,y_j) + O(k^2) + O(h^2), 
\end{eqnarray*}
where $\widehat{L}$ uses the BDF2 scheme for the drift.  The finite
difference approximations are therefore consistent with $\mathcal{L}$
in \eqref{rep_def:L} of order 1 and 2, respectively.

We can hence define a scheme for the numerical approximation
$U^m = (U_{i,j}^m)_{i,j}$ to the ESO value function $u$ in the partial
information case in the interior of the mesh by
\begin{eqnarray}
\label{LCP}
\min\left(
\frac{U^{m}_{i,j}-U^{m+1}_{i,j}}{k} - 
2 m k
\left(
(\overline{L}- r I)  U^m\right)_{i,j},
U_{i,j}^m - \max(x_i-K,0)
\right) = 0,&& \\
\nonumber
0\le m<M, \ 0<i<N, \ 0<j<L,&&
\end{eqnarray}
in the case of BDF1, and similarly in the case of BDF2.

From the construction of $\overline{L}$, the left-hand side of
\eqref{LCP} is increasing in $U^{m}_{i,j}$, and decreasing in
$U^{m'}_{i',j'}$ for all $(m',i',j')\neq (m,i,j)$, and therefore
satisfies the definition of monotonicity in Barles and Souganidis
\cite{barles1991convergence}.  The monotonicity is violated for the
BDF2 scheme due to the alternating signs in the approximations to the
first time and space derivatives.  It is shown in Bokanowski and
Debrabant \cite{bokanowski2018high} that such schemes still have good
stability properties for American options under
Black-Scholes. Although this analysis is not applicable here due to
the degeneracy of the diffusion operator, we observe no stability
issues in the numerical tests.  We emphasise that the judicious choice
of mesh and discretisation of the second derivative terms is crucial
for the stability of the scheme, due to again the degeneracy.

Summarising, we obtain the following properties of the schemes.

\begin{proposition}
\label{prop:bdfscheme}

The BDF1 scheme \eqref{LCP} is monotone and consistent with
\eqref{lcp-incomp} in the interior
$(-K^2/x_{\max}^2,x_{\max})\times (y_{\min},1-y_{\min})\times (0,T)$,
of first order in both $h$ and $k$.  The BDF2 scheme is non-monotone
and consistent of second order in both $h$ and $k$.

\end{proposition}

\subsubsection{Boundary and terminal conditions}
\label{subsubsec:batc}

We have four spatial boundaries with different characteristics as a
result of the degeneracy of the drift and diffusion coefficients at
some of the boundaries. The appropriate approximation of the boundary
conditions is therefore essential for convergence to the correct
solution of the initial boundary value problem. We discuss the
boundaries in some detail in turn.

For $x=0$, we set
\begin{eqnarray*}
U_{0,j}^m = 0, \qquad 0\le m<M, 0\le j\le L.
\end{eqnarray*}

For $x=x_{\max}$, we set
\begin{eqnarray*}
U_{N,j}^m = \max(x_N-K, C(t_m,x_N,y_j)), \qquad 0\le m<M, 0<j<L,
\end{eqnarray*}
where $C(t,x,y)$ is the Black-Scholes price of a European call option
at time $t$ and for underlying asset price $X_0=x$, with constant
interest rate $r$ and dividend yield $r-(\mu_0-\eta\sigma y)$,
volatility $\sigma$, strike $K$ and maturity $T$.  For those $y$ where
we can choose $x_{\max}$ such that $x^\star(T,y) \le
x_{\max}$, 
the assumed boundary value coincides with the value function exactly.
Generally, if $x^\star(T,y) > x_{\max}$ for some $y$, but with
$x_{\max}$ several standard deviations away from $K$, the
approximation error in the region of interest will be small.

For $y\rightarrow 0$, we have
\begin{equation*}
\mathcal{L}_{X,\widehat{Y}}f \rightarrow \mu_{0} x f_{x} +
\frac{1}{2}\sigma^{2}x^{2} f_{xx} + \lambda f_{y},
\end{equation*}
which we approximate at $(t_{m}, x_i,y_0)=(t_{m}, x_i,y_{\min})$ for
$0<i<N$ by
\begin{equation*}
\left(\frac{\mu_{0}}{\sigma} - \sigma\eta y_{\min} -
\frac{\sigma}{2}\right) D_x f + \frac{1}{2} D_x^2 f + 
\frac{\lambda}{\eta} \frac{1}{y_{\min}}D_y^+ f,
\end{equation*}
where
$D_x^2 f(t_{m}, x_i,y_{\min}) = (f(t_{m}, x_{i+1},y_{\min})-2 f(t_{m},
x_i,y_{\min}) + f(t_{m}, x_{i-1},y_{\min}))/h^2$.  As the coefficient
of the first $y$-derivative is positive, a right-sided difference
(i.e., using only points in the interior of the domain) is appropriate
and preserves monotonicity of the scheme.

For $y\rightarrow 1$, we have
\begin{equation*}
\mathcal{L}_{X,\widehat{Y}}f \rightarrow \mu_{0} x f_{x} +
\frac{1}{2}\sigma^{2}x^{2} f_{xx},
\end{equation*}
which we approximate at $(t_{m}, x_i,y_L)=(t_{m}, x_i,1-y_{\min})$ for
$0<i<N$ by
\begin{equation*}
\left(\frac{\mu_{0}}{\sigma} - \sigma\eta y_{\min} -
\frac{\sigma}{2}\right) D_x f + \frac{1}{2} D_x^2 f, 
\end{equation*}
using only boundary points.

As $y_{\min} \rightarrow 0$, the above approximations are consistent
with the equation at $y=0$ and $y=1$, respectively.  For fixed
$y_{\min}$, to compute the solution at time $t_m$ at a spatial point
$(x_i,y) \in \{x_i\} \times [0,y_{\min})$, i.e.\ outside the
computational domain, we extrapolate linearly from $y_0=y_{\min}$ by
$U_{i,0}^m + (y-y_0) (U_{i,1}^m-U_{i,0}^m)/(y_1-y_0)$. This is of
second order accurate in $y_{\min}$ as the solution is smooth in this
region. In particular, this is how the value in the regime $Y=0$ is
computed.

Lastly, the numerical terminal condition at $t=T$ is
\begin{eqnarray*}
U_{i,j}^M = \max(x_i-K,0), \qquad 0\le i \le N, 0\le j \le L.
\end{eqnarray*}

\subsubsection{Penalisation and Newton iteration}
\label{subsybsec:pani}

We now consider the penalty approximation
\begin{eqnarray}
\label{penalty}
\qquad
\frac{V^{m+1}_{i,j}-V^m_{i,j}}{k} + 
2 m k
\left(
(\overline{L}- r I)  V^m\right)_{i,j}
  + \rho \max\left(\max(x_i-K,0)-V_{i,j}^m,0 \right) = 0
\end{eqnarray}
for a penalty parameter $\rho > 0$, in the case of BDF1, and similarly
in the case of BDF2.

Defining $P$ as the $(N+1)\times (L+1)$ vector with
$P_{i,j} = \max(x_i-K,0)$ and $D(V)$ as the $((N+1)\times (L+1))^2$
diagonal matrix with $D_{(i,j),(i,j)}(V)=1$ if $V_{i,j}<P_{i,j}$ and 0
otherwise, this can be re-written as
\begin{equation*}
\left((1 + r mk) I-2 k (m k) \overline{L}\right) +\rho k D(V^m)) V^m =
k V^{m+1} + D(V^m) P.  
\end{equation*}

The solution of this type of equation by semi-smooth Newton iterations
is discussed in \cite{forsyth2002quadratic}.  In the case of the BDF1
scheme, $-\overline{L}$ is an M-matrix and hence
$(1 + r mk) I-2 k (m k) \overline{L}$ is a strictly diagonally
dominant M-matrix. This guarantees on the one hand convergence of the
solution of the penalised solution $V=V(\rho)$ of \eqref{penalty} to
$U$ from \eqref{LCP} as $\rho\rightarrow \infty$, and also convergence
of the Newton iteration in finitely many steps.  In practice, we can
choose the penalty parameter very large (e.g., $10^{10}$) to make the
difference between $V$ and $U$ negligible, without a negative impact
on other properties of the scheme.

We end by stating without detailed proof the convergence result for
the first order scheme.

\begin{proposition}
\label{prop:convfos}
  
The solution $V$ of the penalised BDF1 scheme \eqref{penalty}
converges to the solution $u$ of \eqref{lcp-incomp} uniformly on
compact subsets of $(0,T)\times (0,\infty)\times (0,1)$ as $k$, $h$,
$y_{\min}$ $\rightarrow 0$ and $x_{\max}$, $\rho$ $\rightarrow
\infty$.

\end{proposition}

We report the number of required Newton iterations, alongside the
empirically observed convergence order, below.

\subsection{The full information case}

We begin by observing that the full information ESO value function
$v(t,x,i)\equiv v_{i}(t,x)$, $i=0,1$, satisfying
(\ref{eq:pde1})-(\ref{eq:vibc0}), is also the unique solution in
$[0,T]\times\mathbb R_{+}\times\{0,1\}$ of the equivalent linear
complementarity problem (LCP)
\begin{eqnarray*}
\min\left(-\mathcal{L}_{0}v_{0}(t,x) + \lambda\left(v_{0}(t,x) - v_{1}(t,x)\right),
v_0- (x-K)^{+} \right) = 0, && \quad x\in \mathbb R_{+}, \; t\in[0,T), \\   
\min\left(-\mathcal{L}_{1}v_{1}(t,x), v_1- (x-K)^{+} \right) = 0, &&
\quad x\in \mathbb R_{+},  \; t\in[0,T), \\  
v_{i}(T,x) =  (x-K)^{+},  &&\quad x\in\mathbb R_{+},  \; i=0,1,
\end{eqnarray*}
where we repeat for convenience
\begin{equation*}
\mathcal{L}_{i}f(t,x) = \left(\frac{\partial}{\partial t} +
\mu_{i}x\frac{\partial}{\partial x} +
\frac{1}{2}\sigma^{2}x^{2}\frac{\partial^{2}}{\partial
x^{2}} - r\right)f(t,x),  \quad i=0,1.
\end{equation*}

We approximate this LCP by
\begin{eqnarray*}
\min\left(
\frac{V^{0,m}_{i}-V^{0,m+1}_{i}}{k} - (L  V^{0,m})_{i} + \lambda
(V^{0,m}_{i}-V^{1,m}_{i}), 
V_{i}^{0,m} - \max(x_i-K,0)
\right) = 0,&& \\
\min\left(
\frac{V^{1,m}_{i}-V^{1,m+1}_{i}}{k} - (L  V^{1,m})_{i},
V_{i}^{1,m} - \max(x_i-K,0)
\right) = 0,&& \\
\nonumber
0\le m<M, \ 0<i<N,&&
\end{eqnarray*}
where $x_i$ is as in \eqref{x_mesh} and
\begin{equation*}
(L  V^{j,m})_{i} = 
\left(\mu_{j} -\frac{1}{2}\sigma^2\right)  \frac{V^{j,m}_{i+1}-
V^{j,m}_{i-1}}{2 h} + 
\frac{1}{2} \frac{V^{j,m}_{i+1}- 2 V^{j,m}_{i}+ V^{j,m}_{i-1}}{h^2} 
- r V^{j,m}_{i}.
\end{equation*}

Consistency and monotonicity, and hence convergence, follow directly
in this case.  The scheme is of first order in $k$ and of second order
in $h$. The computational complexity is smaller than in the
two-dimensional case though and we therefore do not propose a
second-order version.  Penalisation is now applied separately to the
two components, and a Newton iteration can be applied in the natural
way to the system of equations.

\subsection{Numerical tests}
\label{subsec:numtests}

We discuss here some tests for the numerical performance of the
partial information algorithm.  The full information case is
straightforward and we do not report our test results here. In this
section, we test in detail the convergence of the finite difference
scheme with respect to the discretisation parameters.  The financial
parameters chosen are $\sigma=0.3$, $\lambda=0.1$, $\mu_0=0.08$,
$\mu_1=-0.05$, $r=0.025$, $T=10$, $K=100$.  The truncation parameters
were $y_{\min}=0.02$, $x_{\max}=8 K$, and the mesh parameters $h$ and
$k$ varied as detailed below.

We list in Table \ref{tab:numer} various quantities of interest for
different mesh refinements, for both the BDF1 and BDF2 scheme, where
$N$ and $L$ are (as above) the number of mesh intervals in the $x$ and
$y$ directions, and $M$ the number of timesteps.  The numbers for $N$
and $M$ are arrived at by the rule
$N=2 \lceil N_0 \sqrt{2}^{n} \rceil$, $n\ge 0$, with $N_0=8$, and
$M= \lceil M_0 \sqrt{2}^{n} \rceil$, $n\ge 0$, with $M_0=16$. This is
motivated by the identical convergence order in $h$ and $k$ for each
of the schemes.  Then, $L$ is determined as explained below
\eqref{eq:y_i} and also proportional to $N$ and $M$.  We ensure
moreover that $N$ is even for the mesh construction above. Here, $N_0$
and $M_0$ are chosen empirically so that the errors from the time and
space discretisation are similar. The fact that we arrived at
$N=L\approx M$ for these particular model parameters is coincidental.

\begin{figure}
\includegraphics[width=0.7\textwidth]{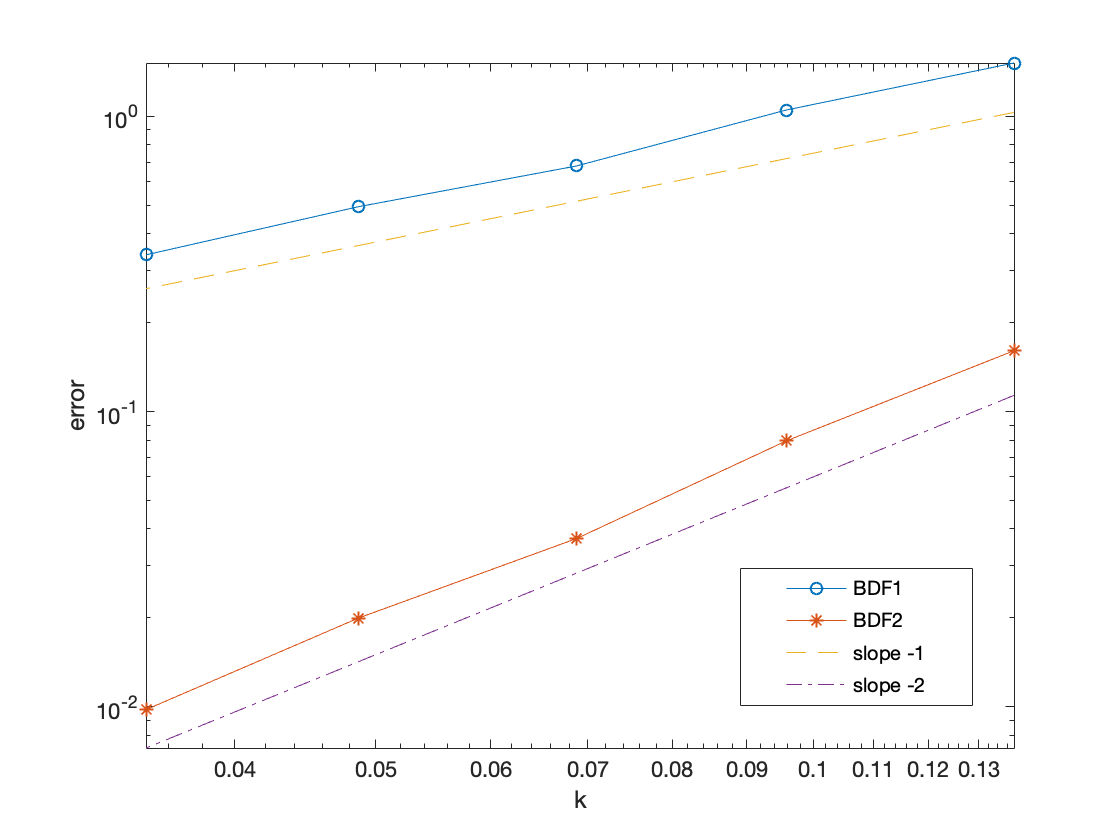}
\caption{Estimated pointwise errors for decreasing timesteps as in
  Table \ref{tab:numer}. The comparison with lines of slope -1 and -2
  in the loglog plot demonstrates first and second order convergence
  of the BDF1 and BDF2 scheme, respectively.}
\label{fig:numer}
\end{figure}

\begin{table}[h]
\begin{tabular}{|cc|cccc|cc|}
\hline
&& BDF2 &&&&BDF1& \\
\hline
$N = L$ & $M$ & error & order & av.\ iter. & CPU (s)  & error & order \\
\hline
24   &   23 &  $1.61 \cdot 10^{-1}$    &    -                &         2.4   &      0.38     &    $1.52 \cdot 10^{0}$            &         -                            \\
34   &   33 &   $7.95 \cdot 10^{-2}$  &    2.03           &           2.5   &      1.2   &        $1.05 \cdot 10^{0}$           &         1.06                                  \\
46   &   46 &    $3.71 \cdot 10^{-2}$  &    2.20           &             2.6   &    2.9      &       $6.79 \cdot 10^{-1}$           &         1.26                                 \\
66   &   65 &   $2.00 \cdot 10^{-2}$    &    1.79           &             2.6   &    7.6     &      $4.94 \cdot 10^{-1}$            &        0.92                               \\
92   &   91 &  $9.77 \cdot 10^{-3}$    &     2.06          &             2.7     &  34    &     $3.39 \cdot 10^{-1}$          &         1.08                          \\       
\hline 
\end{tabular}
\caption{\label{tab:numer} For a sequence of meshes, given are:
estimated pointwise errors of the BDF1 and BDF2 schemes; the
resulting convergence orders; the average number of Newton
iterations; and the run time.}
\end{table}

The numerical solution is evaluated at $(t,x,y)=(0,K,1/2)$ and then
the error (third and seventh column) estimated by extrapolation from
the solutions for subsequent mesh refinements; the order (fourth and
eigth column) is then estimated from the errors for consecutive
meshes.  The numbers clearly demonstrate first order and second order
convergence for the BDF1 and BDF2 scheme, respectively.  This
behaviour is further illustrated in Figure \ref{fig:numer}. The error
on the finest level is smaller than $0.01$ absolutely, or 1 basis
point given a strike of 100.

We also report in Table \ref{tab:numer} the number of Newton
iterations needed to solve the nonlinear system, averaged over all
time points. For non-uniform meshes, the number is typically higher
close to maturity due to the singular behaviour of the exercise
boundary, but this effect is alleviated by the local refinement.

The total number of unknowns increases by a factor of
$\sqrt{2}^3 \approx 2.8$ upon refinement, and this is a lower bound
for the asymptotic increase in computational complexity. In practice,
the cost of solving each linear system within the Newton iteration,
involving a sparse block-tridiagonal matrix, using the default sparse
equation solver in Matlab, increases superlinearly. For optimised
performance a multigrid solver as in Reisinger and Rotaetxe Arto
\cite{reisinger2017boundary} could be used.  Both the iteration count
and computational time are very similar between the two schemes, and
we only report the BDF2 ones.

\section{Numerical results: ESO exercise \& valuation}
\label{sec:numex}

This section demonstrates numerically the exercise policies of the
agents in Section \ref{subsec:ex}.  In Section \ref{subsec:postex}, we
undertake a study of post-exercise stock returns which supports the
approach taken in the empirical literature on private information. We
consider the impact of the information differential on ESO valuation
in Section \ref{subsec:val}.

\subsection{Difference in exercise policies due to information
  differential} \label{subsec:ex} 

\begin{figure}[!htbp]	\centering
{\includegraphics[width=\textwidth] {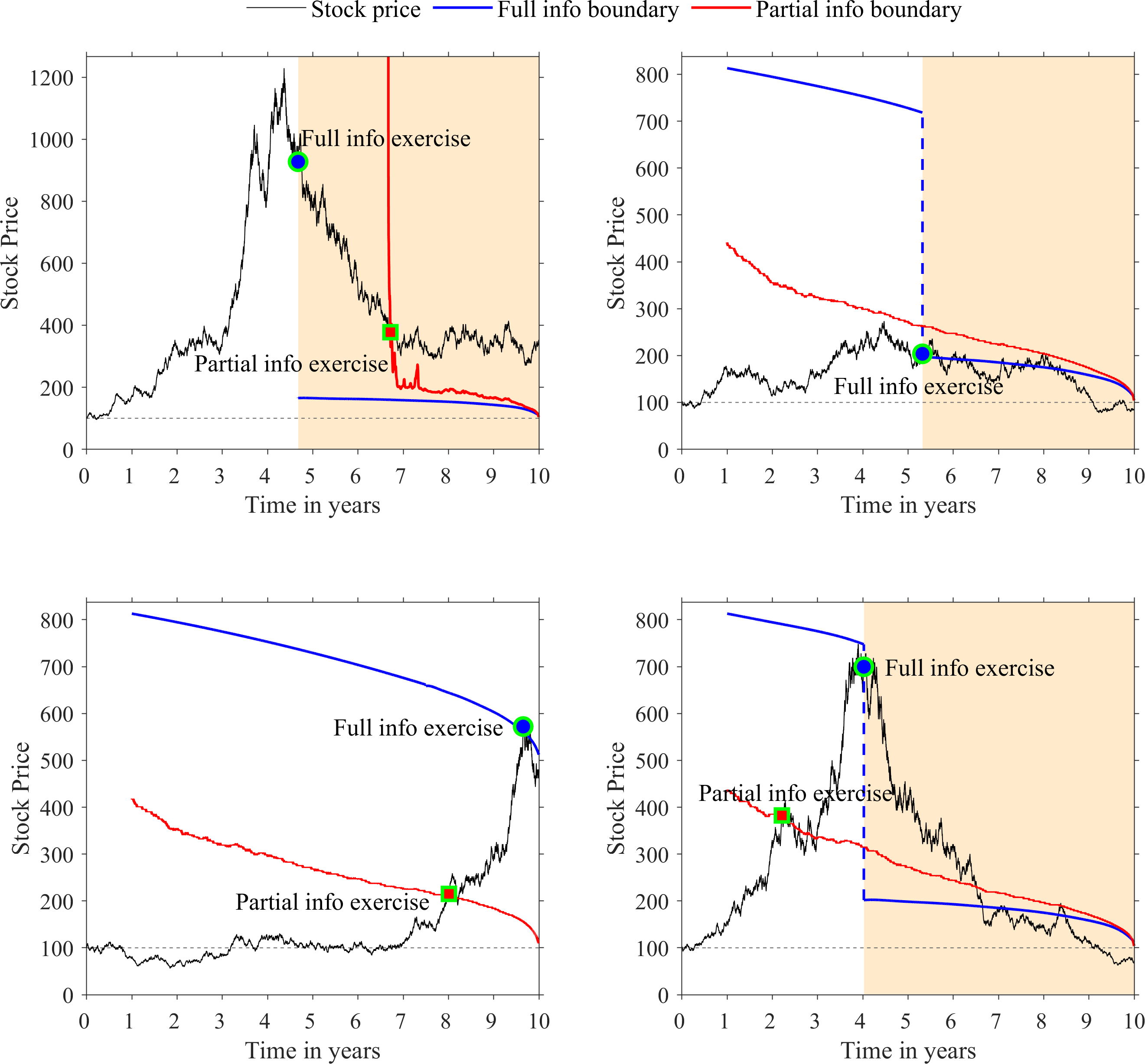}}
\caption{\small{ {\bf Monte Carlo simulations of the stock price,
      thresholds and exercise decisions of the agent's with full and
      partial information.}  In each panel we display the stock price,
    the exercise boundary for the full information case, and the
    exercise boundary for the partial information model, with
    $\hat{Y}_{0} = \mathbb{E} [Y_{0}] = y_0 = 0$.  Exercise decisions
    of the (full information agent, partial information agent with
    $y_0=0$) are marked with (circles, squares).  The option maturity
    is ten years with a one-year vesting period $t_v = 1$, and granted
    at-the-money with $X_0=K=100$.  In each panel, the shaded
    background indicates the switch in drift regime to
    $\mu_1 < \mu_0$.  In all panels, we take parameter values for the
    transition intensity $\lambda = 10\%$, volatility $\sigma=30\%$,
    and the riskfree rate is $r=2.5\%$.  In the top left panel,
    expected returns are given by $\mu_0 = 8\%, \mu_1 = -5\%$ so that
    $\mu_0 > r > \mu_1$ holds. In all other panels, expected returns
    in the two regimes are $\mu_0 = 2\%$, $\mu_1 = -2\%$, so that
    $r > \mu_0 > \mu_1$.}} 
\label{fig:montecarlo}
\end{figure}

We are primarily interested in the difference between the exercise
policies for the agents, due to the information differential they
have. To illustrate exercise patterns for both agents, we numerically
solve for the thresholds of both types of agents, and simulate the
stock price to demonstrate exercise behaviour.  A set of outputs with
various parameter values are plotted in Figure {\ref{fig:montecarlo}}.
In each panel we display the stock price, the exercise boundary for
the agent with full information, $x^*_{i}(t) ; i=0,1; t \in [0,T]$,
and the partially informed agent's exercise boundary,
$x^*(t, .); t \in [0,T]$ with
$\hat{Y}_{0} = \mathbb{E} [Y_{0}] = y_0 = 0$. We set the switch
intensity to be $\lambda = 10\%$ which implies a probability of 63 \%
of $\mu_0$ switching to $\mu_1$ during the option's life. Given the
``vast majority of options are granted at-the-money'' with maturities
of ten years (Carpenter et al.~\cite{csw2015}) we consider an ESO
granted at-the-money with $X_0=K=100$ and maturity $T=10$ years. We
include a vesting period of one year, $t_v = 1$. The shaded area in
each panel denotes the time after the changepoint has occurred,
ie. the drift has switched from $\mu_0$ to $\mu_1$. Exercise decisions
are recorded on each plot for both the partial information agent (with
a square) and the fully informed agent (with a circle).

In the top-left panel, we observe, since $ \mu_0 > r$,
$x^*_{0}(\cdot) = \infty$ and no exercise occurs before the change
point. The agent with full information exercises on the change point.
The threshold of the partially informed agent, $x^*(\cdot,\cdot) $,
rapidly drops from infinity following the change point, as the
filtering puts higher weight on the switch having occurred. The agent
with partial information exercises as the stock price reaches the
threshold. However, the fully informed agent has obtained a far larger
option payoff in this scenario.

The remaining three panels consider the case $ r > \mu_0 > \mu_1$.
The upper-right panel demonstrates a scenario where the stock price is
not performing as well as in the left panel, and the agent with
partial information never exercises. The agent with full information
exercises on the change point, although the stock price does go
slightly higher after that. The agent with full information has
obtained a higher option payoff than the agent with partial
information, as the latter never exercises and the option is
out-of-the-money at maturity.
 
In the lower-left panel, where no change point occurs before option
maturity, consistent with Proposition \ref{prop:tvfull},
$x^*_{0} (T-) = \max( K, \frac{r}{r-\mu_{0}}K ) = 500$.  In this
panel, the stock does very well. The stock price first reaches the
boundary of the partially informed agent and finally, the much higher
boundary of the agent with full information.  Under this scenario, the
fully informed agent has benefited from the additional information
(the knowledge that the switch has not occurred) and has secured a
much higher payoff than the agent with partial information.

Finally, the lower-right panel demonstrates a scenario where the agent
with full information exercises in direct response to the switch and
benefits from the additional information. In this panel, the partial
information agent has already exercised as the stock price crosses
their boundary. The agent with full information continues to wait as
he knows the switch has not occurred. He then benefits with a larger
exercise payoff by exercising exactly at the change point.

In all panels, we observe that the boundaries respect the mathematical
results of Sections \ref{sec:fiesop} and \ref{sec:piesop}.  The full
information boundaries are in accordance with Corollary
\ref{corr:fiexbds} since we can observe the ordering
$x^*_{0}(t) \geq x^*_{1}(t) \geq K$ for the three panels where
$ r > \mu_0 > \mu_1$, and, when $\mu_0 > r$, we see
$x^*_{0}(t) = \infty$.  For any $\mu_i$, we have $ x^*_{i}(T) = K$,
and $x^*_{i} (T-) = \max(K, \frac{r}{r-\mu_{i}}K)$ for $\mu_i < r$
from Proposition \ref{prop:tvfull} is also satisfied.  In the top left
panel with $\mu_0 > r$, consistent with Remark \ref{rem:murpi}, we
have no early exercise for the agent with partial information.  The
exercise boundary for the agent with partial information,
$x^{*}(t,.)$, is indeed decreasing in $t$, in accordance with Lemma
\ref{lem:cmtdpi}, and the boundaries respect Lemma
\ref{lem:tvpartial}.

In Figure \ref{fig:surface}, we illustrate the complete exercise
surfaces generated by the model for the agents with full and partial
information.  We plot the full information thresholds,
$x^*_0(t), x^*_{1}(t); t \in [0,T]$ and the partial information
surface, $x^*(t, y); t \in [0,T], y \in [0,1]$.  The behaviour with
the full and partial information thresholds with respect to time is
consistent with that displayed in Figure \ref{fig:montecarlo}. For
example, consistent with Proposition \ref{prop:tvfull}, we have for
the full information boundaries,
$x^*_{0}(10-) = 500, x^*_{1}(10-)=100$.  Turning to the behaviour of
the thresholds with respect to varying $\hat{Y}$, the exercise surface
for the agent with partial information, $x^{*}(t,y)$ is indeed
decreasing in $y$, in accordance with Lemma \ref{lem:cmtdpi}.

\begin{figure}[!htbp]
\centering
{\includegraphics[width=0.8\textwidth] {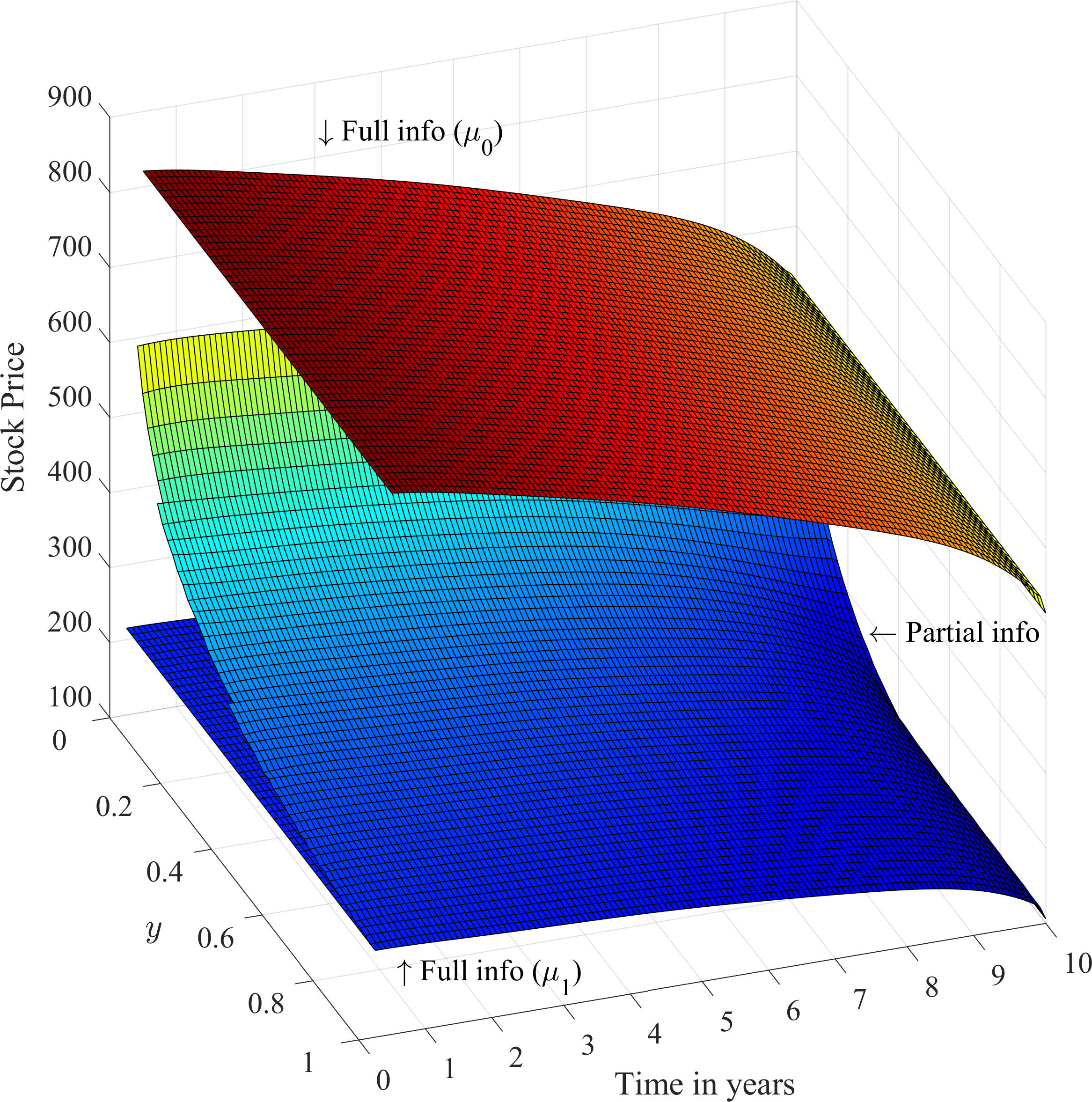}}  
\caption{ \small{ {\bf Exercise surfaces under full and partial
information against time and $y \in [0,1]$ the spatial dependence
arising from the filtered process $\hat{Y}$.}  The uppermost and
lowermost surfaces are those of the agent with full information: the
uppermost surface $x^*_0 (.)$ in regime 0 with $\mu_0$, and the
lowermost surface $x^*_1 (.)$ in regime 1 with $\mu_1$. These do not
depend upon $y$ so each surface for the full information agent is
constant in the $y$ direction, and has been plotted for comparison
with the surface of the agent with partial information.  The exercise
surface $x^{*}(t,y); t \in [0,T], y \in [0,1]$ for the agent with
partial information lies between the two surfaces from the full
information problem.  The option maturity is ten years and granted
at-the-money with $X_0=K=100$.  Expected returns in the two regimes
are $\mu_0 = 2\%$, $\mu_1 = -2\%$, transition intensity $\lambda =
10\%$, volatility $\sigma=30\%$, and the riskfree rate is $r=2.5\%$.
}}
\label{fig:surface}
\end{figure}

\subsection{An application to post-exercise returns}
\label{subsec:postex} 

\begin{figure}[!htbp]
\centering
{\includegraphics[width=0.8\textwidth] {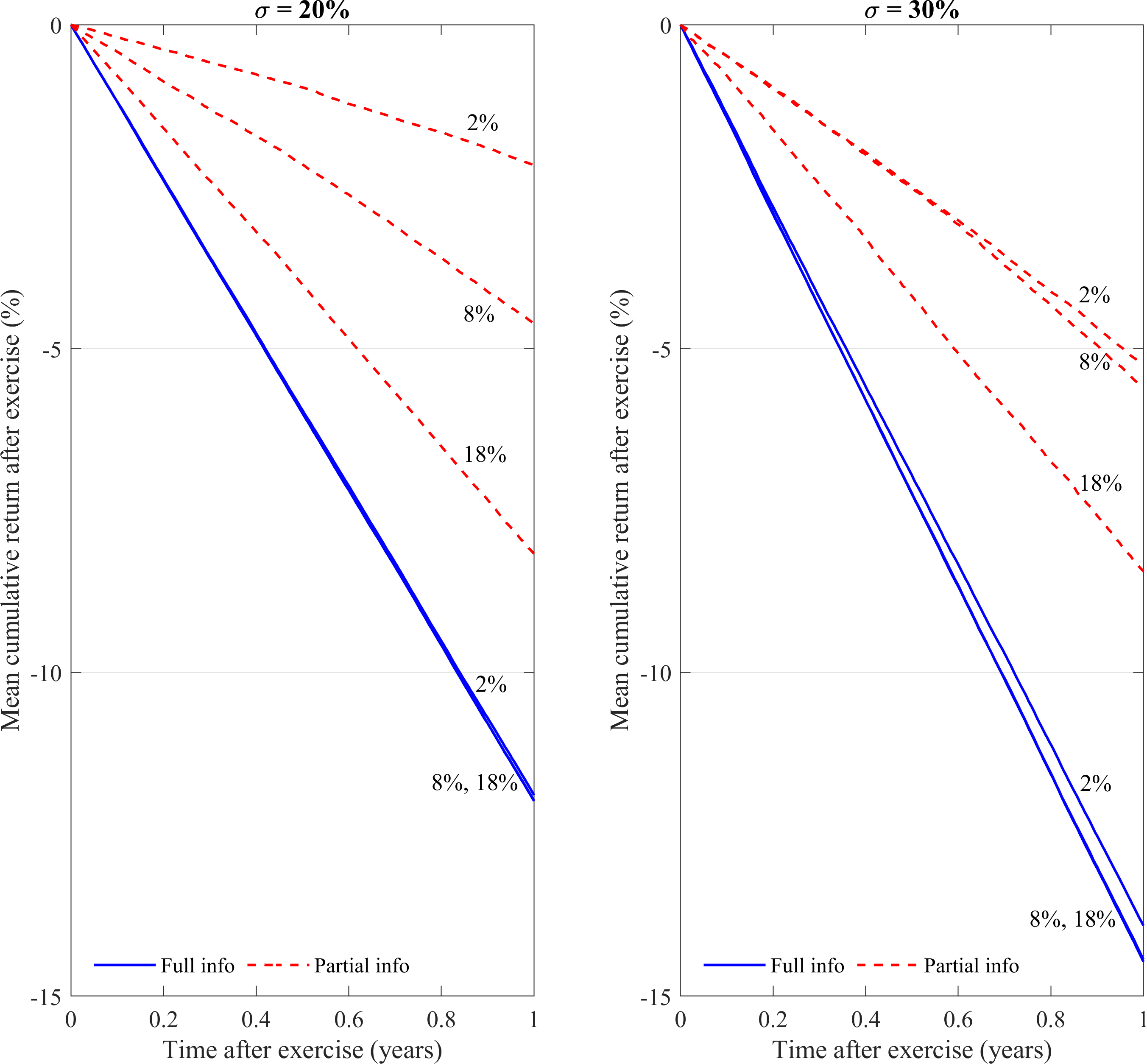}}  
\caption{ \small{ {\bf Mean cumulative post-exercise returns
with full and partial information over one year.}  In the left panel,
volatility is $\sigma=20\%$, in the right panel, volatility is
$\sigma=30\%$.  The expected return $\mu_1 = -10\%$ is fixed, and
expected return $\mu_0 = 2, 8, 18\%$.  The transition intensity is
$\lambda = 10\%$, and the riskfree rate is $r=2.5\%$.  The option
maturity is ten years and granted at-the-money with $X_0=K=100$.
Simulations use 1 million price paths.  }}
\label{fig:postex}
\end{figure}

In this section, we demonstrate how our model can be linked to the
empirical finance literature on private information and the exercise
of ESOs. In fact, our model provides a consistent theoretical
foundation for the empirical tests conducted in this literature.  A
body of papers (Aboody et al.~\cite{aboody08}, Brooks et
al.~\cite{brooks2012} and Cicero \cite{cicero09}) aim to identify and
evidence that executives use private information when exercising their
company ESOs. (Note these papers, and ours, do not take any stance on
the legality of such exercises). The idea is ``if the executive has
negative information, the stock (owned by them) would almost surely be
sold, and in all likelihood the stock would perform poorly for a
period of time thereafter." (Brooks et al.~\cite{brooks2012},
p733). These studies examine ESO exercise data in which the stock is
sold upon exercise.  The general approach is then to examine the
long-term abnormal returns after the exercise of ESOs. If the abnormal
returns are significantly negative following exercise, there is
support for the explanation of private information being a factor in
exercise decisions.  Brooks et al.~\cite{brooks2012} match firms with
ESO exercises of top executives, believed to hold private information,
to firms with no record of top executive ESO exercises, but with
similar firm characteristics. They observe one year of stock data
following each top executive option exercise, and compute the BHAR
(buy-and-hold-abnormal returns) to be the so-called insider returns
minus the matched returns.  Brooks et al.~\cite{brooks2012} find
strong evidence of ESO exercise due to insider information, via
significant negative differences in the returns. Insider exercises are
linked to significantly negative post-exercise returns over the
following year.

We use our model of differential information to generate post-exercise
returns over the following year, and compare any difference between
returns following exercises by our agent with full information versus
our agent with partial information. Our full information agent knows
the change point in the stock price process, when the expected return
of the stock drops.  Then, if our model is to be consistent with the
approach of Brooks et al.~\cite{brooks2012}, we need to demonstrate
that the difference between the average post-exercise returns from
fully and partially informed agents is also negative.  To be in line
with the literature, we consider simulated returns for one year
following each option exercise, and we only include exercises which
are more than one year before option maturity (exercises closer to
maturity are considered less likely to be information related).

In Figure \ref{fig:postex} we display results of the simulations.  The
left-hand panel uses volatility 20\% whilst the right-hand panel uses
30\%. We keep the expected return after a change point fixed at
$\mu_1= -10\%$, but take three values for the expected return $\mu_0$.
We first observe that the mean cumulative log-returns post-exercise
for the case of full information do not vary much with the different
values of initial expected return $\mu_0$. Recall from Corollary
\ref{corr:fiexbds}, with full information, and with
$\mu_0 = 8\%, 18\% > r $, exercises occur only in the bad state.  Thus
the one year log-returns are $\mu_1 - 0.5 \sigma^2$ (For the left
panel, -12\%, and for the right panel, -14.5\%).  With $\mu_0 = 2\%$,
there are some early exercises in the good state, and their occurrence
increases with volatility, as shown by the plots.  With only partial
information, the cumulative log-returns post-exercise vary much more
with the value of $\mu_0$.  We see the post-exercise returns are
worse, the higher the expected return $\mu_0$. The one year
log-returns for the partial information case vary between about -2.2\%
to -8.2\% when volatility is 20\%, and -5.3\% and 8.4\% for volatility
30\%.

Overall, the simulations support our conjecture that indeed, exercises
by the agent with full information are followed by significantly
negative stock returns, and the difference between average
post-exercise returns for fully and partially informed agents is
significantly negative.  For our simulations, this difference between
mean post-exercise returns for fully and partially informed agents
varies between about -3.8\% and -9.7\%, depending on the expected
stock return $\mu_0$ and volatility, covering the range of values
reported by Brooks et al.~\cite{brooks2012}.  Our model thus provides
theoretical support for the tests conducted in the empirical
literature to evidence so-called insider exercises.

\subsection{ESO valuation}
\label{subsec:val} 

\begin{table}
\begin{center}
\renewcommand{\arraystretch}{1.2}
\setlength\tabcolsep{2.5pt}

\begin{tabular}{|cc|dd|dd|dd|dd|dd|dd|}
\hline
\multicolumn{2}{|c|}{read:} &
  \multicolumn{6}{c|}{$\lambda = 10\%$} &
  \multicolumn{6}{c|}{$\lambda = 20\%$}
  \\
\cline{1-2}$A_V$ &
  $A_U$ &
  \multicolumn{6}{c|}{$\sigma = 20\%$} &
  \multicolumn{6}{c|}{$\sigma = 20\%$}
  \\
  $E_V$ &
  $E_U$ &
  \multicolumn{6}{c|}{$\mu_1$} &
  \multicolumn{6}{c|}{$\mu_1$}
  \\
$V$ &
  $U$ &
  \multicolumn{2}{c|}{$-2\%$} &
  \multicolumn{2}{c|}{$-5\%$} &
  \multicolumn{2}{c|}{$-10\%$} &
  \multicolumn{2}{c|}{$-2\%$} &
  \multicolumn{2}{c|}{$-5\%$} &
  \multicolumn{2}{c|}{$-10\%$}
  \bigstrut[b]\\
\hline
\multicolumn{1}{|c}{\multirow{9}[6]{*}{$\mu_0$}} &
  \multirow{3}[2]{*}{$2\%$} &
  2.9 &
  1.5 &
  5.3 &
  2.7 &
  7.7 &
  3.8 &
  4.2 &
  3.1 &
  7.4 &
  5.3 &
  10.5 &
  7.1
  \\
 &
   &
  23.1 &
  23.1 &
  19.4 &
  19.4 &
  16.0 &
  16.0 &
  18.9 &
  18.9 &
  13.5 &
  13.5 &
  8.9 &
  8.9
  \\
 &
   &
  26.0 &
  24.6 &
  24.7 &
  22.1 &
  23.7 &
  19.8 &
  23.0 &
  21.9 &
  20.9 &
  18.8 &
  19.4 &
  16.0
 \\
\cline{2-14} &
  \multirow{3}[2]{*}{$8\%$} &
  5.5 &
  0.0 &
  9.8 &
  0.3 &
  14.6 &
  1.3 &
  7.3 &
  0.7 &
  13.0 &
  2.5 &
  19.0 &
  5.0
  \\
 &
   &
  60.0 &
  60.0 &
  54.5 &
  54.5 &
  48.8 &
  48.8 &
  40.6 &
  40.6 &
  32.9 &
  32.9 &
  25.5 &
  25.5
  \\
 &
   &
  65.4 &
  60.0 &
  64.2 &
  54.8 &
  63.4 &
  50.0 &
  47.9 &
  41.3 &
  45.9 &
  35.5 &
  44.5 &
  30.5
 \\
\cline{2-14} &
  \multirow{3}[2]{*}{$18\%$} &
  12.8 &
  0.4 &
  21.9 &
  1.8 &
  33.3 &
  4.8 &
  16.1 &
  1.8 &
  27.8 &
  5.2 &
  41.5 &
  10.5
  \\
 &
   &
  210.6 &
  210.6 &
  200.5 &
  200.5 &
  188.6 &
  188.6 &
  122.8 &
  122.8 &
  109.6 &
  109.6 &
  94.9 &
  94.9
  \\
 &
   &
  223.3 &
  211.0 &
  222.4 &
  202.3 &
  221.9 &
  193.4 &
  138.9 &
  124.6 &
  137.4 &
  114.8 &
  136.4 &
  105.4
 \\
\hline
 &
   &
  \multicolumn{6}{c|}{$\sigma = 30\%$} &
  \multicolumn{6}{c|}{$\sigma = 30\%$}
\\
 &
   &
  \multicolumn{6}{c|}{$\mu_1$} &
  \multicolumn{6}{c|}{$\mu_1$}
  \\
 &
   &
  \multicolumn{2}{c|}{$-2\%$} &
  \multicolumn{2}{c|}{$-5\%$} &
  \multicolumn{2}{c|}{$-10\%$} &
  \multicolumn{2}{c|}{$-2\%$} &
  \multicolumn{2}{c|}{$-5\%$} &
  \multicolumn{2}{c|}{$-10\%$}
 \\
\cline{2-14}\multicolumn{1}{|c}{\multirow{9}[6]{*}{$\mu_0$}} &
  \multirow{3}[2]{*}{$2\%$} &
  3.6 &
  2.5 &
  6.2 &
  4.1 &
  9.4 &
  5.9 &
  5.1 &
  4.2 &
  8.8 &
  7.1 &
  13.0 &
  9.9
 \\
 &
   &
  32.2 &
  32.2 &
  27.8 &
  27.8 &
  23.2 &
  23.2 &
  27.7 &
  27.7 &
  21.2 &
  21.2 &
  14.7 &
  14.7
  \\
 &
   &
  35.8 &
  34.7 &
  34.1 &
  32.0 &
  32.6 &
  29.1 &
  32.8 &
  31.9 &
  30.1 &
  28.3 &
  27.7 &
  24.6
\\
\cline{2-14} &
  \multirow{3}[2]{*}{$8\%$} &
  5.7 &
  0.0 &
  10.0 &
  0.2 &
  15.3 &
  0.9 &
  7.8 &
  1.0 &
  13.5 &
  3.2 &
  20.3 &
  6.1
 \\
 &
   &
  68.7 &
  68.7 &
  62.7 &
  62.7 &
  55.9 &
  55.9 &
  49.5 &
  49.5 &
  41.0 &
  41.0 &
  31.8 &
  31.8
  \\
 &
   &
  74.4 &
  68.7 &
  72.7 &
  62.9 &
  71.3 &
  56.8 &
  57.3 &
  50.5 &
  54.5 &
  44.2 &
  52.2 &
  37.9
  \\
\cline{2-14} &
  \multirow{3}[2]{*}{$18\%$} &
  12.2 &
  0.0 &
  20.9 &
  0.2 &
  32.2 &
  1.2 &
  15.6 &
  0.5 &
  26.5 &
  2.0 &
  40.4 &
  5.5
\\
 &
   &
  216.0 &
  216.0 &
  205.8 &
  205.8 &
  193.3 &
  193.3 &
  130.0 &
  130.0 &
  116.4 &
  116.4 &
  100.6 &
  100.6
  \\
 &
   &
  228.2 &
  216.0 &
  226.7 &
  206.0 &
  225.5 &
  194.5 &
  145.5 &
  130.4 &
  143.0 &
  118.5 &
  141.0 &
  106.0
  \\
\hline
 &
   &
  \multicolumn{6}{c|}{$\sigma = 40\%$} &
  \multicolumn{6}{c|}{$\sigma = 40\%$}
 \\
 &
   &
  \multicolumn{6}{c|}{$\mu_1$} &
  \multicolumn{6}{c|}{$\mu_1$}
  \\
 &
   &
  \multicolumn{2}{c|}{$-2\%$} &
  \multicolumn{2}{c|}{$-5\%$} &
  \multicolumn{2}{c|}{$-10\%$} &
  \multicolumn{2}{c|}{$-2\%$} &
  \multicolumn{2}{c|}{$-5\%$} &
  \multicolumn{2}{c|}{$-10\%$}
  \\
\cline{2-14}\multicolumn{1}{|c}{\multirow{9}[6]{*}{$\mu_0$}} &
  \multirow{3}[2]{*}{$2\%$} &
  4.3 &
  3.3 &
  7.3 &
  5.4 &
  11.0 &
  7.7 &
  6.1 &
  5.3 &
  10.3 &
  8.7 &
  15.4 &
  12.5
\\
 &
   &
  40.8 &
  40.8 &
  35.9 &
  35.9 &
  30.3 &
  30.3 &
  36.0 &
  36.0 &
  28.7 &
  28.7 &
  20.6 &
  20.6
  \\
 &
   &
  45.1 &
  44.1 &
  43.2 &
  41.3 &
  41.3 &
  38.0 &
  42.1 &
  41.3 &
  39.0 &
  37.4 &
  36.0 &
  33.1
\\
\cline{2-14} &
  \multirow{3}[2]{*}{$8\%$} &
  6.3 &
  0.0 &
  10.7 &
  0.2 &
  16.4 &
  1.0 &
  8.6 &
  1.5 &
  14.6 &
  4.2 &
  22.1 &
  7.9
\\
 &
   &
  77.9 &
  77.9 &
  71.4 &
  71.4 &
  63.8 &
  63.8 &
  58.4 &
  58.4 &
  49.1 &
  49.1 &
  38.6 &
  38.6
  \\
 &
   &
  84.2 &
  77.9 &
  82.2 &
  71.6 &
  80.3 &
  64.8 &
  67.0 &
  59.8 &
  63.8 &
  53.4 &
  60.7 &
  46.5
  \\
\cline{2-14} &
  \multirow{3}[2]{*}{$18\%$} &
  12.3 &
  0.0 &
  20.8 &
  0.0 &
  32.1 &
  0.4 &
  15.8 &
  0.2 &
  26.6 &
  1.1 &
  40.6 &
  3.7
 \\
 &
   &
  223.5 &
  223.5 &
  213.0 &
  213.0 &
  199.9 &
  199.9 &
  138.3 &
  138.3 &
  124.3 &
  124.3 &
  107.3 &
  107.3
  \\
 &
   &
  235.8 &
  223.5 &
  233.8 &
  213.1 &
  232.0 &
  200.3 &
  154.1 &
  138.4 &
  150.9 &
  125.4 &
  148.0 &
  111.0
  \\
\hline
\end{tabular}
\caption{\label{table_comp} \small{ {\bf Comparative statics for the
full and partial information option values.}  Each subpanel of six
numbers contains the option values for full information in the left
column and partial information in the right column. Each column
contains (from top to bottom) the American component, the European
component and the total ESO value (sum of European and American). We
have, for the full information case, $ V=E_V + A_V$, and for the
partial information model $U = E_U + A_U$.  The option maturity is ten
years and granted at-the-money with $X_0=K=100$.  Parameter values
considered are: $\mu_0 = 2\%, 8\%, 18\%$, $\mu_1 = -2\%, -5\%, -10\%$,
transition intensity $\lambda = 10\%, 20\%$, volatility $\sigma=20\%,
30\%, 40\%$, and the riskfree rate is fixed at $r=2.5\%$. We fix
$y_0=0$.}}
\end{center}
\end{table}

\begin{table}
\begin{center}
\renewcommand{\arraystretch}{1.2}
\setlength\tabcolsep{2.5pt}

\begin{tabular}{|cc|dd|dd|dd|dd|dd|dd|}
\hline
\multicolumn{2}{|c|}{read:} &
  \multicolumn{6}{c|}{$\lambda = 10\%$} &
  \multicolumn{6}{c|}{$\lambda = 20\%$}
  \\
\cline{1-2}$A_V$ &
  $A_U$ &
  \multicolumn{6}{c|}{$ t_v = 3$ years } &
  \multicolumn{6}{c|}{$ t_v = 3 $ years }
  \\
  $E_V$ &
  $E_U$ &
  \multicolumn{6}{c|}{$\mu_1$} &
  \multicolumn{6}{c|}{$\mu_1$}
  \\
$V$ &
  $U$ &
  \multicolumn{2}{c|}{$-2\%$} &
  \multicolumn{2}{c|}{$-5\%$} &
  \multicolumn{2}{c|}{$-10\%$} &
  \multicolumn{2}{c|}{$-2\%$} &
  \multicolumn{2}{c|}{$-5\%$} &
  \multicolumn{2}{c|}{$-10\%$}
  \bigstrut[b]\\
\hline
\multicolumn{1}{|c}{\multirow{9}[6]{*}{$\mu_0$}} &
  \multirow{3}[2]{*}{$2\%$} &
    3.5 &
   2.5 &
 6.0 &
  4.1 &
 8.6  &
 5.8  &
  4.9 &
 4.2  &
  8.3 &
  7.0 &
  11.5  &
  9.5
  \\     
  &
  & 
  32.2 &
  32.2 &
  27.8 &
  27.8 &
  23.2 &
  23.2 &
  27.7 &
  27.7 &
  21.2 &
  21.2 &
  14.7 &
  14.7  
  \\ 
   &
   &
 35.7  &
 34.7  &
  33.8 &
  31.9 &
 31.8  &
 29.0  &
 32.6  &
 31.9  &
 29.5  &
 28.2  &
 26.2 &
24.2  
 \\
\cline{2-14} &
  \multirow{3}[2]{*}{$8\%$} &
    5.6 &
  0.0 &
   9.6 &
  0.2 &
  14.2 &
  0.9 &
 7.5   &
 1.0 &
 12.7  &
  3.2 &
 18.4  &
  6.1
  \\
 &
   &
  68.7 &
  68.7 &
  62.7 &
  62.7 &
  55.9 &
  55.9 &
  49.5 &
  49.5 &
  41.0 &
  41.0 &
  31.8 &
  31.8  \\
 &
   &
  74.3 &
  68.7 &
  72.3 &
  62.9 &
  70.1 &
  56.8 &
  57.0 &
  50.5 &
  53.7 &
  44.2 &
  50.2 &
  37.9
 \\
\cline{2-14} &
  \multirow{3}[2]{*}{$18\%$} &
 11.9 &
 0.0  &
  20.1 &
  0.2 &
  30.1 &
 1.2  &
   15.0 &
 0.4   &
   25.3 &
 2.0  &
    37.4 &
 5.4
  \\
 &
   &
 216.0 &
  216.0 &
  205.8 &
  205.8 &
  193.3 &
  193.3 &
  130.0 &
  130.0 &
  116.4 &
  116.4 &
  100.6 &
  100.6
  \\
 &
   &
 227.9  &
 216.0  &
 225.9  &
 206.0  &
  223.4 &
  194.5 &
   145.0 &
 130.4  &
  141.7  &
  118.4  &
  138.0 &
   106.0
 \\
\hline
 &
   &
  \multicolumn{6}{c|}{$ t_v = 5$ years } &
  \multicolumn{6}{c|}{$ t_v = 5$ years }
\\
 &
   &
  \multicolumn{6}{c|}{$\mu_1$} &
  \multicolumn{6}{c|}{$\mu_1$}
  \\
 &
   &
  \multicolumn{2}{c|}{$-2\%$} &
  \multicolumn{2}{c|}{$-5\%$} &
  \multicolumn{2}{c|}{$-10\%$} &
  \multicolumn{2}{c|}{$-2\%$} &
  \multicolumn{2}{c|}{$-5\%$} &
  \multicolumn{2}{c|}{$-10\%$}
 \\
\cline{2-14}\multicolumn{1}{|c}{\multirow{9}[6]{*}{$\mu_0$}} &
  \multirow{3}[2]{*}{$2\%$} &
  3.3 &
   1.3 &
  5.0   &
   3.8 &
  6.8 &
  5.1 &
  4.3 &
  3.8 &
  6.8 &
  6.0 &
  8.8 &
  7.7
 \\
 &
   &
  32.2 &
  32.2 &
  27.8 &
  27.8 &
  23.2 &
  23.2 &
  27.7 &
  27.7 &
  21.2 &
  21.2 &
  14.7 &
  14.7
  \\
 &
 &
  35.5 &
  34.5 &
  32.8 &
  31.6 &
  30.0 &
  28.3 &
  32.0 &
  31.5 &
  28.0 &
 27.2  &
 23.5  &
 22.4  
\\
\cline{2-14} &
  \multirow{3}[2]{*}{$8\%$} &
  5.0   &
  0.0 &
  8.2 &
  0.2 &
  11.7 &
   0.9 &
  6.5 &
 1.0  &
  10.6 &
   3.1 &
   14.5 &
   5.9
 \\
 &
   &
  68.7 &
  68.7 &
  62.7 &
  62.7 &
  55.9 &
  55.9 &
  49.5 &
  49.5 &
  41.0 &
  41.0 &
  31.8 &
  31.8
  \\
 &
   &
  73.7 &
  68.7 &
  70.9 &
  62.9 &
  67.6 &
 56.8  &
 56.0  &
  50.5 &
  51.6 &
 44.1  &
  46.3 &
  37.7
  \\
\cline{2-14} &
  \multirow{3}[2]{*}{$18\%$} &
 10.7  &
 0.0  &
 17.6  &
  0.2 &
   26.0 &
 1.2  &
 13.2  &
 0.4 &
  21.4 &
 2.0  &
 30.6  &
  5.4
\\
 &
   &
  216.0 &
  216.0 &
  205.8 &
  205.8 &
  193.3 &
  193.3 &
  130.0 &
  130.0 &
  116.4 &
  116.4 &
  100.6 &
  100.6
  \\
   &
   &
  226.7 &
  216.0  &
 223.4  &
  206.0 &
  219.3 &
  194.5 &
  143.2 &
  130.4 &
  137.8 &
  118.4  &
  131.2  &
  106.0
  \\
\hline
\end{tabular}
\caption{ \label{table_vest} \small{ {\bf{The effect of a vesting
period of 3 and 5 years on ESO valuation by agents with full and
partial information.}} We take $\sigma=30\%$ and thus values should be
compared with the middle panels of Table \ref{table_comp}.  Each
subpanel of six numbers contains the option values for full
information in the left column and partial information in the right
column. Each column contains (from top to bottom) the American
component, the European component and the total ESO value (sum of
European and American). We have, for the full information case, $
V=E_V + A_V$, and for the partial information model $U = E_U + A_U$.
The option maturity is ten years and granted at-the-money with
$X_0=K=100$.  We consider vesting periods of $t_v=3$ years and $t_v =
5$ years.  Parameter values considered are: $\mu_0 = 2\%, 8\%, 18\%$,
$\mu_1 = -2\%, -5\%, -10\%$, transition intensity $\lambda = 10\%,
20\%$, and the riskfree rate is fixed at $r=2.5\%$. We fix $y_0=0$.
}}
\end{center}
\end{table}

We now turn to the impact of differential information about the stock
price on ESO valuation by the agents themselves.  We emphasise that
the ESO values we report represent the value to the individual agent,
often termed subjective value in the literature on ESO compensation
(see Carpenter \cite{carpenter98}).  It is the value under the
$\Bbb{P}$ measure.

Table \ref{table_comp} reports the time-zero ESO values for the agent
with full information, $V= V_0$, and for the agent with partial
information, $U = U_0$. The table also gives a breakdown of each ESO
value into its European (labelled $E_V$ and $E_U$) and American early
exercise components (labelled $A_V$ and $A_U$).  This breakdown shows
the value differential arises entirely from the American early
exercise component of the ESO values. As the simulations demonstrate
in Section \ref{subsec:ex}, the agent with full information uses this
knowledge to time his option exercise advantageously.

The additional value that the agent with full information places on
the ESO is significant in magnitude. Consider the American early
exercise value as a proportion of total ESO value for each of the full
and partial information cases. For example, with $\lambda=10\%$,
$\mu_0=8\%$, $\mu_1=-5\%$, $\sigma=30\%$, the American early exercise
value represents 13.8\% (10/72.7) of the ESO value for full
information, and 0.32\% (0.2/62.9) of value for partial information.
If we compare these American-as-proportion-of-total values for the
full and partial information agents, we see that the magnitude is much
larger for the agent with full information. In our example, we see the
13.8\% is about 43 times larger than the 0.32\%.  This ratio varies
between around 1.2, up to values as high as 69.  There are also some
zero values for the American early exercise value under partial
information, which tend to be for high $\mu_0$ and the best case of
-2\% for $\mu_1$, indicating no early exercises take place. In these
scenarios, the agent with full information gains significantly as he
uses his additional information on the change point to time exercise
advantageously.

The table documents how the full and partial information ESO values
vary with changes in stock specific parameters $\mu_0, \mu_1$ and
$\sigma$, and the transition intensity $\lambda$.  The option values
under full and partial information increase with the value of expected
return $\mu_0$. Under the partial information model, the American
component of value often drops with $\mu_0$, consistent with there
being relatively few exercises for high values of $\mu_0$.

Under both full and partial information, option values decrease with
the absolute value of $\mu_1$.  However, the American component of
value increases with $|\mu_1 |$, for both full and partial
information, indicating that the ability to time the exercise of the
option is more valuable when the expected return following a change
point is worse.  For example, scenarios with a low $\mu_0$ of 2\%, the
worst case for $\mu_1$ of -10\%, and the transition probability
$\lambda=0.2$, the American component of option value can be as high
as 40-50\% of ESO value.

Volatility increases the full and partial information option
values. The European component is increasing in volatility but the
American component can increase or decrease. If $\mu_0$ is
sufficiently high, volatility can reduce the American component of
value in both full and partial information scenarios.

A higher probability of a downward jump in expected return (higher
$\lambda$) reduces the full and partial information ESO values. The
European component of value is reduced, as a higher $\lambda$ simply
means a greater chance of switching to the bad regime. However, the
American component of value increases with $\lambda$ because the
ability to time the exercise becomes more important when the chance of
the bad state is increased. This is true for both the agent with full
and the agent with partial information.

We now turn to briefly examine the impact of vesting on ESO
valuation. Section \ref{sec:vesting} described the effect of a vesting
period $[0, t_v)$ on option exercise.  Table \ref{table_vest}
documents the ESO values for both a 3- and a 5-year vesting period for
a representative subset of market parameters from Table
\ref{table_comp} and fixing volatility at $\sigma=30\%$. Hence the ESO
values should be compared to the middle panel of Table
\ref{table_comp} where the same volatility is used but no vesting
period.

As we anticipate, the American early exercise values are
non-increasing as $t_v$ increases, as the option becomes
un-exercisable for a larger share of the life of the option. For
example, when $\mu_0=2\%, \mu_1=-5\%, \sigma=30\%$ and $\lambda=10\%$,
the early exercise value for full information falls from 6.2, to 6, to
5, as $t_v$ increases from 0, to 3 years, to 5 years. Corresponding
early exercise values in the partial information setting are 4.1, 4.0,
3.8.  For some parameters, say when $\mu_0$ is high, the early
exercise value in the case with partial information did not vary with
$t_v$, as these are situations where there are no exercises taking
place when there is no vesting period, and thus additional exercise
restrictions via vesting do not alter the agent's value.

{\small

\bibliography{esoexcp_refs}

\bibliographystyle{siam}

}

\end{document}